\documentclass[12pt]{article}

\RequirePackage[colorlinks,citecolor=blue,urlcolor=blue]{hyperref}

\usepackage{amsmath}
\usepackage{graphicx}
\usepackage{natbib}
\usepackage{url} 
\usepackage{indentfirst}

\RequirePackage{amsthm,amsfonts,amssymb}
\usepackage[shortlabels]{enumitem}
\usepackage{booktabs,tabularx}
\RequirePackage{graphicx}

\makeatother
\def\|{\Vert}
\def\tilde {\widetilde}
\def\hat {\widehat}
\def\bar {\overline}

\def\mC{\mathcal C}

\def\mE{\mathcal E}
\def\mI{\mathcal I}
\def\mG{\mathcal G}
\def\mL{\mathcal L}
\def\mN{\mathcal N}
\def\mP{\mathcal P}
\def\mQ{\mathcal Q}
\def\mR{\mathbb  R}
\def\mS{\mathcal S}
\def\mV{\mathcal V}

\def\bB{\boldsymbol B}
\def\bC{\boldsymbol C}
\def\bK{\boldsymbol K}
\def\bt{\boldsymbol t}
\def\bu{\boldsymbol u}
\def\bY{\boldsymbol Y}
\def\bx{\boldsymbol x}
\def\balpha{\boldsymbol \alpha}

\def\var{\operatorname{var}}
\def\bias{\operatorname{bias}}

\def\tr{\operatorname{tr}}
\def\cov{\operatorname{cov}}
\def\cor{\operatorname{cor}}
\def\sign{\mbox{sign}}

\def\AMISE{\mbox{AMISE}}

\def\CV{\mbox{CV}}
\def\argmin{\operatorname*{argmin}}

\def\GPA{\textup{GPA}}
\def\PGPA{\textup{PGPA}}
\def\OS{\textup{OS}}
\def\MA{\textup{MA}}
\def\BE{\textup{BE}}
\def\KRR{\textup{KRR}}

\theoremstyle{plain}

\newtheorem{theorem}{Theorem}
\newtheorem{lemma}{Lemma}

\newcommand{\blind}{0}

\usepackage{xcolor}

\def\spacingset#1{\renewcommand{\baselinestretch}%
{#1}\small\normalsize} \spacingset{1}

\begin{document}

\date{}


\if0\blind
{
  \title{\bf Grid Point Approximation for Distributed Nonparametric Smoothing and Prediction}
 \author{ Yuan Gao$^1$,  Rui Pan$^2$, 
   Feng Li$^1$,  Riquan Zhang$^3$, Hansheng Wang$^1$   \vspace{0.2cm} \\ 
    \small $^1$Guanghua School of Management, Peking University \\
     \small $^2$School of Statistics and Mathematics, Central University of Finance and Economics\\
     \small $^3$School of Statistics and Information, Shanghai University of International Business and Economics
}
  \maketitle
} \fi

\if1\blind
{
  \bigskip
  \bigskip
  \bigskip
  \begin{center}
    {\LARGE\bf Grid Point Approximation for Distributed Nonparametric Smoothing and Prediction}
\end{center}
  \medskip
} \fi

\bigskip
\begin{abstract}
Kernel smoothing is a widely used nonparametric method in modern statistical analysis.
The problem of efficiently conducting kernel smoothing for a massive dataset on a distributed system is a problem of great importance.
In this work, we find that the popularly used one-shot type estimator is highly inefficient for prediction purposes.
To this end, we propose a novel grid point approximation (GPA) method, which has the following advantages. 
First, the resulting GPA estimator is as statistically efficient as the global estimator under mild conditions.
Second, it requires no communication and is extremely efficient in terms of computation for prediction.
Third, it is applicable to the case where the data are not randomly distributed across different machines.
To select a suitable bandwidth, two novel bandwidth selectors are further developed and theoretically supported.
Extensive numerical studies are conducted to corroborate our theoretical findings. 
Two real data examples are also provided to demonstrate the usefulness of our GPA method.
\end{abstract}

\noindent%
\vfill

\newpage
\spacingset{1.65} 

\section{Introduction} \label{sec:intro}

Kernel smoothing refers to a general class of techniques for nonparametric estimation \citep{wand1994kernel}. Kernel smoothing is a powerful statistical tool and has been found to be extremely useful for density estimation \citep{parzen1962estimation,sheather2004density}, nonparametric and semiparametric regression \citep{fan1992design,fan1999statistical,xia2004efficient,li2008variable}, longitudinal data analysis \citep{fan2004new,sun2007estimation}, covariance modelling \citep{yin2010nonparametric}, and many other applications. The basic idea of nonparametric kernel smoothing is to conduct statistical estimation based on data from a sufficiently compact local region. Depending on the type of statistical method for local estimation, different kernel smoothing methods can be employed. The two most popularly used methods are the local constant method \citep{nadaraya1965non,watson1964smooth}, and the local linear regression method \citep{fan1996local,loader2006local}. However, the effectiveness of kernel smoothing methods depends on the effective size of the local region, which is well controlled by the bandwidth parameter \citep{jones1996brief}. If the bandwidth is too large, then the subsequent estimation bias could be substantial. In contrast, if the bandwidth is too small, the variability of the resulting estimator could be substantial. Consequently, the trade-off between estimation bias and variability in terms of bandwidth selection has been a critical research problem in the literature \citep{hardle1985optimal, fan1995data,xia2002asymptotic}.

Datasets of massive size are often encountered in modern statistical analysis. This leads to both opportunities and challenges. On the one hand, massive datasets lead to extremely large sample sizes. This gives nonparametric kernel smoothing powerful support to combat the curse of dimensionality and, thus, the ability to capture more flexible function shapes. On the other hand, large sample sizes also pose serious challenges to computation. Often, the amount of data is too large to be loaded into the computer memory simultaneously. Consequently, the full dataset must be divided into many nonoverlapping subsets with much reduced sample sizes. Subsequently, subsets can be processed on different computers in parallel. The results are then reported to a central computer for final assembly, which leads to a more powerful final estimator. This is an effective method for handling massive datasets and is referred to as divide-and-conquer \citep{gao2022review}.

In fact, divide-and-conquer is an extremely useful technique for the analysis of massive datasets. For parametric model estimation, the one-shot averaging approach has been frequently adopted \citep{zhang2013communication,volgushev2019distributed}. The basic idea is to assemble a global estimator by averaging the local estimators. The advantage of the one-shot approach is that the communication cost is minimal since only one round of communication between the central and local machines is needed. However, the disadvantage is also apparent; the statistical efficiency may not be optimal unless a sufficient amount of data are randomly assigned to local machines \citep{wang2017efficient, jordan2019communication, wang2021distributed}.
As an alternative, some iterative approaches have been developed that implement multiple rounds of communication. One important iterative approach is the one-step method. 
For instance, \cite{huang2019distributed} propose a one-step approach that utilizes a single Newton-Raphson update with one additional round of communication. The idea of one-step updating is also adopted for quantile regression in a recent work of \cite{pan2022note}. 
For high-dimensional response data, \cite{hector2020doubly} and \cite{hector2021distributed} propose one-step meta-estimators under the scheme of distributed computation.
However, many one-step methods require computing and transmitting matrices (e.g., the Hessian matrix) to every local machine. To reduce the communication cost, \cite{jordan2019communication} develop a communication-efficient surrogate likelihood framework, where the local Hessian matrix is used for updating the parameter, as opposed to the whole sample matrix. Consequently, the communication cost due to the transmission of matrices is reduced. This technique has been used in many iterative methods; for example, see 
\cite{wang2019distributed} and \cite{fan2021communication}.

Despite their practical usefulness for massive data analysis, most existing divide-and-conquer methods focus on distributed methods for various parametric models. Nonparametric and/or semiparametric models are extremely useful in practice but are much less well developed for application to massive datasets. Successful applications include financial econometrics \citep{fan2000prospects}, longitudinal data analysis \citep{wu2006nonparametric}, functional data analysis \citep{ferraty2006nonparametric}, and many others. For a good summary, we refer to \cite{hardle2004nonparametric}. Consequently, there is a natural and practical need to develop distributed methods for nonparametric models. However, research in this area has been very limited. One possible reason is that developing distributed estimation methods for nonparametric models is considerably more challenging than that for parametric models. The key difference is that the target parameter of parametric models is of finite dimension. In contrast, the parameter in nonparametric models is often of infinite dimension (e.g., a nonparametric smooth curve). The statistically and computationally efficient implementation of distributed estimation for nonparametric models is a problem of great interest \citep{zhang2015divide,shang2017computational,chang2017divide}.

Other researchers have studied this problem and there have been two main branches of research.
This first is the reproducing kernel Hilbert space (RKHS) based method, where the associated estimator is obtained by solving a regularized least squares problem in a given RKHS.
To construct the distributed estimator, the one-shot approach is commonly used; see, for example, \cite{zhang2015divide}, \cite{lin2017distributed}, and \cite{shang2017computational}.
The second is the kernel smoothing method. In this area, two important pioneering works have been published in the literature. The first is \cite{li2013statistical}, who developed a divide-and-conquer type approach for nonparametric kernel density estimation. Under the assumption that a globally optimal bandwidth $O(N^{-1/5})$ can be used where $N$ is the total sample size, the authors show that the resulting density estimator can be as good as the optimal estimator. 
The second work is developed by \cite{chang2017divide}, who study the problem of distributed nonparametric regression by kernel smoothing. 
The authors propose a one-shot type estimation method under the assumption that the data are randomly distributed across local machines. They investigate the theoretical properties and find that some restrictive conditions must be imposed on the number of local machines to achieve the optimal convergence rate.
To this end, they provide two improved algorithms based on different bandwidth selection procedures.
In this paper, we restrict our attention to nonparametric kernel smoothing.

These two pioneering works provide some important contributions to the existing literature. First, to the best of our knowledge, they are among the first group of researchers to study the distributed computation problem for nonparametric kernel smoothing. Second, they present the fundamental theoretical framework for distributed nonparametric kernel smoothing. Their research inspires a number of interesting topics for future study. First, in many cases, data are not randomly distributed across local machines \citep{zhu2021least}. Thus, developing a statistically efficient estimator not depending on randomly distributed data is important. Second, as pointed out by \cite{li2013statistical}, for optimal estimation accuracy, the globally optimal bandwidth should be used. However, another challenging problem is how to practically estimate this bandwidth parameter for a massive dataset. To solve this problem, \cite{li2013statistical} propose a heuristic bandwidth selection rule, but its theoretical properties are not provided. \cite{chang2017divide} also suggest two bandwidth selection procedures. However, they require a whole sample-based cross-validation step. This is practically extremely difficult to implement, especially when the massive dataset is split between many different local machines.

Lastly, we want to emphasize that nonparametric prediction is a problem with real-time applications but also poses great challenges, especially for massive datasets. On one hand, many real-time applications require extensive predictions. For example, air ticket or travel insurance agencies need to provide customers with real-time predictions regarding flight delays for each departure. These predictions are crucial as they are influenced by the daily variability in weather conditions. Similarly, a web-based agency specializing in used car sales might experience a surge in inquiries, with millions of requests for price prediction on second-hand cars with specific conditions within minutes during promotion events. To our best knowledge, existing techniques cannot effectively address the aforementioned practical demands on real-time predictions. 
On the other hand, nonparametric prediction by kernel smoothing can be computationally expensive because it requires searching through the entire training dataset for each new data point \citep{kpotufe2010curse}.
The RKHS based method also encounters a similar computational issue. Specifically, one needs to invert an $N\times N$ Gram matrix, where $N$ is the size of the training dataset \citep{zhang2005learning}. 
This can be computationally expensive or even infeasible when $N$ is very large. Therefore, conducting nonparametric prediction in a distributed way is an important issue. 
Recent distributed computing methods based on the one-shot approach can help reduce the computational burden; see, for example, \cite{zhang2015divide}, \cite{lin2017distributed}, and \cite{chang2017divide}. 
However, these methods need to send every newly observed data point to each local machine to compute the local estimate. Subsequently, these local estimates are sent back to the central machine to form the final predictor.
This procedure not only incurs substantial communication costs but also requires that all local machines are running normally throughout the entire prediction phase \citep{lv2022discussion}. Then, how to make predictions on a distributed system in a more efficient way becomes a problem of great interest.

The above discussion is the inspiration for this work. Specifically, we propose a novel distributed estimation and prediction method for nonparametric regression. For convenience, we refer to it as a method of grid point approximation (GPA) and the resulting estimators are referred to as the GPA estimators. 
To implement the GPA method, the following key steps are needed. First, pre-specify some grid points and calculate a number of local moment statistics at these points on each local machine. The idea of grid points selection is in spirit of binning techniques, which are proposed by \cite{fan1994fast} and \cite{wand1994fast}. This technique is particularly useful in accelerating the computation of kernel estimators. Second, transmit these local moment statistics to the central machine so that the global kernel estimators for these grid points can be exactly constructed. These global estimators are essential for computing the GPA estimator. Last, on the central machine, linear or higher order interpolation is conducted to estimate or predict a given observation by using the global estimators computed on the grid points.

The newly proposed method has the following important features. First, our method imposes no requirement for the distribution of data across local machines. Any type of data distribution (e.g., homogeneous or heterogeneous distribution) is permitted. Second, our method does not require a minimal sample size for local machines. The sample sizes on different local machines can be arbitrarily large or small. Third, the newly proposed distributed GPA method is extremely efficient in communication in the sense that prediction can be conducted on the central machine with zero communication cost.  We also want to emphasize that an industry standard for handling these intensive predictions is to cache or persist the prediction model in memory for fast prediction and reduced latency. In-memory caching allows the system to handle a high volume of queries efficiently by reducing the load on the raw database. With such an infrastructure, when a user queries the price of a specific car, the system quickly retrieves the necessary model parameters from the cache to provide an instant price estimate. Our GPA methods can be easily cached in the system for intensive prediction tasks, such as users frequently querying the estimated price of a used car. 
We show theoretically and numerically that the GPA estimators can be statistically as efficient as the globally optimal estimator. We also develop two novel bandwidth selectors for the GPA estimator with rigorous theoretical support.
Finally, we expand the GPA method by incorporating a higher order polynomial interpolation and extending it to the multivariate scenario.

The rest of the article is organized as follows. In Section 2, we introduce our newly proposed GPA method, including the estimation, prediction, and bandwidth selection procedures. 
In Section 3, we discuss two extensions of the GPA method.
In Section 4, we present numerical studies to support our theoretical findings.
Section 5 includes real data analysis for illustrative purposes.  
The article is concluded in Section 6. All the technical proofs and extra numerical studies are included in the appendices.

\section{The Main Methodology}
\label{sec:meth}

\subsection{Model and Notations}

Let $(X_i, Y_i)$ be the observation collected from the $i$-th subject with $ 1\le i \le N $, where $N$ represents the entire sample size. Furthermore, $ Y_i \in \mR $ is the response of interest, and $ X_i = (X_{i1},\dots,X_{ip})^\top\in \mR^p $ is the associated $p$-dimensional covariate.
To flexibly model the regression relationship between $Y_i$ and $X_i$, a nonparametric regression model can be assumed as follows \citep{wand1994kernel,fan1996local}:
\begin{align*}
	Y_i = \mu(X_i) + \varepsilon_i,
\end{align*}
where $ \varepsilon_i $ is the independent random noise term with mean $0$ and variance $\sigma^2$,  and $ \mu(\cdot) $ is an unknown but sufficiently smooth function.
For simplicity, we first assume that $X_i$ is a random and univariate predictor (i.e., $p=1$) and is compactly supported on $[0,1]$. Furthermore, we also assume that the random noise $\varepsilon_i$ is independent of $X_i$.

The objective is to estimate $\mu(\cdot)$ nonparametrically. To this end, the standard nonparametric kernel smoothing method is frequently used.
Specifically, let $K(\cdot)$ be a kernel function. It is typically taken to be a continuous probability density function symmetric about $0$.
Next, define $K_h(u) = K(u/h)/h$ with some positive bandwidth $h>0$.
Then, for an arbitrarily fixed point $x\in (0,1)$, we can estimate $\mu(x)$ by minimizing the following locally weighted least squares objective function:
\begin{align}\label{eq:NW_est}
	\hat\mu(x) = \argmin_{a\in\mR} \sum_{i=1}^N (Y_i - a)^2 K_h(X_i - x) = \frac{\sum_{i =1}^N K_h(X_i - x)  Y_i }{\sum_{i =1}^N K_h(X_i - x)}.
\end{align}
The resulting estimator $\hat \mu(x)$ is often referred to the local constant estimator, also known as the Nadaraya-Watson (NW) estimator \citep{wand1994kernel,fan1996local,li2007nonparametric}.
Under appropriate regularity conditions, we know that 
\begin{align*}
	\bias\big\{\hat \mu (x)\big\} = E\big\{\hat\mu(x)\big\} -\mu(x) = B(x)h^2 + o(h^2),\  
	\var\big\{\hat \mu (x)\big\} = \frac{V(x)}{Nh} + o\left( \frac{1}{Nh} \right),
\end{align*}
where
\begin{align}\label{eq:B&V}
	B(x) = \frac{\kappa_{2}}{2} \left\{\ddot\mu(x) + 2\dot\mu(x)\frac{\dot f(x)}{f(x)} \right\}, \ V(x)=\frac{\nu_0 \sigma^2}{f(x)},
\end{align}
$f(\cdot)$ is the probability density function of $X_i$, $\kappa_r = \int u^r K(u) du$ and $\nu_r  =\int  u^r K^2(u) du $ with $r\ge 0$ are constants related to $K(\cdot)$; see \cite{fan1996local} and \cite{li2007nonparametric}. Throughout this article, we refer to $\hat\mu(x)$ in \eqref{eq:NW_est} as the \textit{global estimator} to emphasize the fact that it is obtained from the entire dataset.

\subsection{Divide-and-Conquer Methods} 
\label{subsec:global&OS}

Despite the popularity of the nonparametric kernel smoothing method in the literature, its development for massive datasets occurred only recently. Meanwhile, the idea of divide-and-conquer has been extremely useful for distributed estimation. It is thus of great interest to explore the possibility of developing divide-and-conquer type methods for nonparametric kernel smoothing. Specifically, let $\mS = \{1,\dots,N\}$ be the collection of sample indexes. Assume that the whole sample can be divided into $M$ nonoverlapping subsets such that $ \mS = \cup_{m=1}^M \mS_m $ with $ \mS_{m_1} \cap \mS_{m_2}  = \emptyset$ for any $ m_1 \ne m_2$, where $\mS_m$ refers to the index set of the sample on the $m$-th local machine. For simplicity, we further assume that all the subsets are equally sized with $|\mS_m|=n$. We immediately know that $N=Mn$. 
Then there are two divide-and-conquer strategies to construct distributed kernel smoothing estimators. Their details are given as follows.

\textsc{Strategy 1}. The first one is the one-shot type distributed estimator developed in \citep{chang2017divide}. Specifically, consider the local NW estimator on the $m$-th machine, that is, $\hat\mu_m(x) = \sum_{i \in \mS_m} K_h(X_i - x)  Y_i  \Big/ \sum_{i\in \mS_m} K_h(X_i - x)$. Then, by averaging all the $M$ local NW estimators, the resulting one-shot estimator can be constructed as
$\hat \mu_\textup{OS}(x) = M^{-1}\sum_{m=1}^M \hat \mu_m (x)$. Under appropriate regularity conditions, we can show that if data
are randomly assigned to local machine and local sample size $n$ is sufficiently large, then the one-shot estimator $\hat \mu_\textup{OS}(x)$ can be as efficient as the global estimator $\hat \mu(x)$. The technical details are given in Appendix A.1. 

\textsc{Strategy 2}. We can see that the global NW estimator in \eqref{eq:NW_est} can be written as 
\begin{align} \label{eq:nw_dist}
	\hat \mu(x) =\frac{\sum_{i=1}^N K_h(X_i - x)  Y_i }{\sum_{i=1}^N K_h(X_i - x)}= \frac{\sum_{m=1}^M\sum_{i \in \mS_m} K_h(X_i - x)  Y_i }{\sum_{m=1}^M\sum_{i \in \mS_m} K_h(X_i - x)}.
\end{align}
This immediately suggests a more direct way for distributed computing of the global estimator $\hat \mu(x)$. 
Specifically, let local machines compute and transfer to the central machine $ \sum_{i\in \mS_m}K_h(X_i - x)  Y_i$ and $\sum_{i\in \mS_m} K_h(X_i - x)$ with $1 \le m\le M$.
Next, the central machine assembles these quantities according to \eqref{eq:nw_dist} to obtain the final estimator.
We should emphasize that this estimator is exactly the global estimator but computed in a distributed way.
It solely depends on the whole dataset and has nothing to do with the local sample sizes or the local sample distributions. Consequently, it should have the best statistical efficiency in terms of mean squared error (MSE). A slightly more elaborated asymptotic analysis of $\hat \mu(x)$ can also be found in Appendix A.1.

Under appropriate regularity conditions, we can prove theoretically that both types of the distributed estimators can achieve comparable performance.
However, we find that they face serious challenges for prediction in both computation and communication.
Consider for example $N^*$ new observations $\{X_i^*:\ 1\le i\le N^* \}$ with only feature $X_i^*$s observed.
In order to make a prediction for the unobserved response $Y_i^*$, we need to first broadcast the corresponding feature $X_i^*$ to every local machine. 
Then the $m$-th local machine needs to compute some specific local statistics with respect to $X_i^*$ based on the local sample $\mS_m$ for each $1\le m\le M$.
Specifically, $\hat \mu_m(X_i^*)$ should be computed according to \textsc{Strategy 1}, or $ \sum_{j\in \mS_m}K_h(X_j - x)  Y_i$ and $\sum_{j\in \mS_m} K_h(X_j - X_i^*)$ should be computed according to \textsc{Strategy 2}.
Subsequently, the central machine collects all those local statistics and then aggregates them to obtain $\hat\mu_\textup{OS}(X_i^*)$ or $\hat\mu(X_i^*)$ according to different strategies as the predicted value for $Y_i^*$.
To predict all $N^*$ observations, the above distributed computation process must be replicated $N^*$ times. 
This leads to a huge time cost for both computation and communication.
Thus, it is of great interest to develop a more computationally and communicationally efficient prediction method for a distributed system.

\subsection{Distributed Grid Point Approximation}
\label{subsec:GPA}

To solve this prediction problem, we develop here a novel method called grid point approximation (GPA).
Specifically, let $J$ be a prespecified positive integer, which determines the number of grid points.
With a given $J$, we define the $j$-th grid point as $x_j^*  = j/J$ with $0\le j\le J$. 
Obviously, the distance between two adjacent grid points should be $\Delta = 1/J$.
We then estimate the global estimator $\hat\mu(x)$ according to \eqref{eq:nw_dist} on the grid points $\{x_j^*:\ 0\le j\le J\}$ only.
This is a training process having nothing to do with prediction.
Consequently, it can be conducted in the model training stage before prediction occurs.
Once we have obtained the estimates on each grid point, prediction can be implemented using linear interpolation.
Specifically, suppose $x\in [x_j^* , x_{j+1}^*]$ for some $j$, we can predict $\mu(x)$ by
\begin{align} \label{eq:interpolation}
	\hat\mu_\textup{GPA}(x) = & \frac{x_{j+1}^* - x}{\Delta} \hat \mu(x_j^* ) + \frac{ x- x_j^*}{\Delta} \hat \mu(x_{j+1}^*).
\end{align}
We refer to $ \hat\mu_\textup{GPA}(x)$ as the GPA estimator. 
It is noteworthy that the GPA estimator can be computed on the central machine without incurring any communication cost.
In addition, only simple interpolations are needed for making predictions.
Consequently, the computational cost is also minimal.
Intuitively, the number of grid points $J$ should not be too large; otherwise, the computation and communication cost required for the grid point estimation could be large.
On the other hand, $J$ should not be too small. Otherwise, the bias due to linear interpolation could be substantial.
Therefore, understanding how the number of grid points $J$, together with the bandwidth $h$, affects the statistical efficiency of the GPA estimator becomes a key issue.

To study the statistical properties of the GPA estimator, we need to assume the following technical conditions.
\begin{enumerate}[(C1)] 
	\item (\textsc{Smoothness Condition}) The probability density function $f(\cdot)$ is sufficiently smooth so that its second-order derivative $\ddot f(x) = d^2 f(x) / dx^2$ is continuous. 
	The mean function $\mu(\cdot)$ is sufficiently smooth so that its second-order derivative $\ddot \mu (x)= d^2 \mu(x) / dx^2$ is continuous. \label{cond:smoothness}
	\item (\textsc{Kernel Function}) The kernel function $K(\cdot)$ is a continuous probability density function and symmetric about $0$ with compact support $[-1,1]$. \label{cond:kernel}
	\item (\textsc{Optimal Bandwidth}) Assume $h = CN^{-1/5}$ for some positive constant $C$.\label{cond:bandwidth}
\end{enumerate}
\noindent 
The smoothness conditions in \ref{cond:smoothness} are standard and have been used in previous studies; for example, see \cite{fan1992design}, \cite{fan1996study}, and \cite{li2007nonparametric}. 
A kernel function that satisfies condition \ref{cond:kernel} is often referred to as a second-order kernel \citep{fan1992bias, ullah1999nonparametric,wand1994kernel}. 
This type of kernel functions are usually used to construct kernel estimators for twice continuously differentiable functions. 
Last, condition \ref{cond:bandwidth} directly assumes the bandwidth should be of the optimal convergence rate \citep{ullah1999nonparametric}. 
By this bandwidth, the optimal convergence rate achieved by a global NW estimator in terms of the asymptotic mean squared error is given by $O(1/(Nh)) = O(N^{-1/5})$ under the conditions \ref{cond:smoothness} and \ref{cond:kernel}.
We then have the following theorem.

\begin{theorem}\label{thm:GPA_asymptotic}
	Assume conditions \ref{cond:smoothness}--\ref{cond:bandwidth} hold and $f(x)>0$. Further assume that $Jh\to \infty$ as $N\to \infty$. Then we have $\hat\mu_\textup{GPA}(x) - \mu(x) = \tilde Q_0 + \tilde Q_1 + \tilde Q_2 + \tilde \mQ$,
	where $\tilde Q_0$ is a non-stochastic term satisfying $\tilde Q_0 = O(1/J^2)$, $ \sqrt{Nh}\left\{\tilde  Q_1 - B(x)h^2\right\}\to_d \mN\Big(0, V(x) \Big)$, $E\tilde Q_2 = O\left(h/N+h^4\right)$, $\var(\tilde Q_2)=O\left\{1/(Nh)^2\right\}$, and $\tilde \mQ = O_p\left\{ 1/ (Nh)^{3/2} \right\}$.
	In addition to that, we have
	\begin{align*}
		\sqrt{Nh}\Big\{\hat\mu_\textup{GPA}(x) - \mu(x) - B(x)h^2\Big\}\to_d \mN\Big(0, V(x) \Big).
	\end{align*}
\end{theorem}

The proof of Theorem \ref{thm:GPA_asymptotic} can be found in Appendix A.2.  
Theorem \ref{thm:GPA_asymptotic} implies that as long as $1/J = o(h)$, the GPA estimator $\hat\mu_\textup{GPA}(x)$ should have the same asymptotic distribution as the global estimator $\hat\mu(x)$. Recall that $h =CN^{-1/5}$ for some constant $C>0$.
This means that the number of grid points should satisfy $J / N^{1/5}\to \infty$.
This seems to be an extremely mild condition. 
To clarify the idea, consider, for example, $N=10^{10}$ (i.e., 10 billion). Then, as long as $J\gg 10^2=100$, the GPA estimator could perform as well as the global estimator. 
Further discussion on the case when the covariate $X_i$ does not have compact support can be found in Appendix C.

\subsection{Bandwidth Selection} \label{subsec:bandwidth}

Next, we consider the problem of bandwidth selection. 
To achieve the best asymptotic efficiency, the optimal bandwidth should be used.
By condition (C3), we know that $h=CN^{-1/5}$ with some constant $C>0$.
Then, how to optimally choose the positive constant $C$ becomes the problem of interest.
For this problem, cross-validation (CV) has been frequently used.
Specifically, let $\hat \mu^{(-i)} (x)$ be the global estimator obtained without the $i$-th sample. We use this estimator to conduct an out-sample prediction for the $i$-th sample as $\hat \mu^{(-i)}(X_i)$.
Its prediction error is then given by $\big\{Y_i - \hat\mu^{(-i)}(X_i)\big\}^2$. Averaging these values gives the CV score as 
\begin{align}\label{eq:CV_score}
	\CV(h) = \frac{1}{N}\sum_{i=1}^N \Big\{Y_i - \hat \mu^{(-i)}(X_i) \Big\}^2 w(X_i),
\end{align}
where $w(\cdot)\ge 0$ is a prespecified weight function that trims out-of-boundary observations \citep{racine2004nonparametric}.
Then, an empirically optimal bandwidth is obtained as $\hat h =\argmin_h \CV(h)$.
Note that $\CV(h)$ can be viewed as an empirical alternative to the theoretical asymptotic mean integrated squared error (AMISE) as
\begin{align*}
	\AMISE( h) = \int \bigg\{ B^2(x)h^4 + \frac{V(x)}{Nh}\bigg\} w(x)f(x)dx= \bar B h^4 + \frac{\bar V }{Nh},
\end{align*}
where $\bar B = \int B^2(x) w(x) f(x) dx$, $\bar V = \int V(x) w(x) f(x) dx$.
Then, the asymptotically optimal bandwidth can be defined as
\begin{align}\label{eq:bandwidth_opt}
	h_\textup{opt} = \argmin_h \AMISE(h) = C_\textup{opt}N^{-1/5},
\end{align}
where $C_\textup{opt} = \big\{\bar V / (4 \bar B) \big\}^{1/5}$. 
As demonstrated by \cite{racine2004nonparametric}, we have $\CV(h) = \AMISE(h)\{1+o_p(1)\} + C_1$ under appropriate regularity conditions, where $C_1$ is a constant independent of $h$. To see this, we first write $\CV(h) = \sum_{i=1}^N \{\mu(X_i) - \hat \mu^{(-i)}(X_i)\}^2 w(X_i)/N + \sum_{i=1}^N \{\mu(X_i) - \hat \mu^{(-i)}(X_i)\}\varepsilon_i  w(X_i)/N + \sum_{i=1}^N \varepsilon_i^2 w(X_i) /N= \Omega_1 + \Omega_2 + \Omega_3,$ where $\Omega_1$ is an empirical approximation of $\AMISE(h)$, $\Omega_2$ is a negligible term compared with $\Omega_1$, and $\Omega_3$ is a term independent of the bandwidth $h$ and thus can be approximately treated as a constant. Consequently, we should have $\CV(h) = \AMISE(h)\{1+o_p(1)\} + C_1$, where $C_1$ is some constant independent of $h$.
Consequently, we should reasonably expect $\hat h$ to be a consistent estimator of $h_\textup{opt}$. 
In fact, by Theorem 2.2 in \cite{racine2004nonparametric}, we know that $\hat h / h_\textup{opt} \to_p 1$ as $N\to \infty$ under appropriate conditions.

Unfortunately, such a straightforward method cannot be directly implemented for massive data analysis because the associated computation cost is too high. 
Therefore, we propose two bandwidth selectors which can be computed in a distributed manner.
The first one is the one-shot type bandwidth selector, which is inspired by the idea of one-shot estimation.
To be more specific, let $\hat h_m = \argmin_h \CV_{m}(h)$ be the optimal bandwidth selected on the $m$-th local machine by the typically used CV method. Next, these selected bandwidths are passed to the central machine and averaged to yield a more stable value $\bar h = M^{-1}\sum_{m=1}^M \hat h_m$.
Unfortunately, this bandwidth is not a good estimator for $h_\textup{opt} = C_\textup{opt}N^{-1/5}$. 
This is because $\hat h_m$ is computed on the local machine with $n=N/M$ samples only. 
Hence, the local bandwidth selector $\hat h_m$ should be an estimator for $C_\textup{opt}n^{-1/5}$, other than $C_\textup{opt}N^{-1/5}$.
Consequently, the average of them $\bar h$ is not a good estimator for $C_\textup{opt}N^{-1/5}$ either.
The problem can be easily fixed by rescaling.
This leads to the final one-shot bandwidth selector as $\hat h_\textup{OS} = (N/n)^{-1/5}\bar h =  M^{-1/5} \bar h$.

However, similar to the one-shot estimator, the one-shot bandwidth selector $\hat h_\textup{OS}$ also requires that the entire sample is randomly distributed across different local machines. Otherwise, the distribution of the subsamples on the local machine can be substantially different from that of the whole sample. In this case, the ``optimal bandwidth" computed on a local machine could be a seriously biased estimator for the globally optimal bandwidth as defined in \eqref{eq:bandwidth_opt} even after appropriate re-scaling. To address this issue, we propose a pilot sample-based bandwidth selector, which is particularly useful for nonrandomly distributed samples \citep{pan2022note}.
Specifically, let $n_0$ be the prespecified pilot sample size.
Then, an index set $\mP_m$ is sampled from $\mS_m$ by simple random sampling without replacement such that $|\mP_m|=n_0 / M $ for each $1\le m\le M$. Recall that $\mS_m$ is the index set of the sample on the $m$-th local machine.
We refer to $\mP = \cup_{m=1}^M \mP_m$ as the pilot sample.
Next, these observations in $\mP$ are transferred from the local machines to the driver machine, where the final results are assembled.
Let $\hat h_0 = \argmin_h\CV_{\mP}(h)$ be the optimal bandwidth selected by the CV method based on the pilot sample $\mP$. 
Similarly, we rescale it to obtain the pilot sample-based bandwidth selector as $\hat h_\textup{PLT} = (N/n_0)^{-1/5} \hat h_0$.
The theoretical properties of the two bandwidth selectors are summarized in the following theorem.
\begin{theorem}\label{thm:bandwidth}
	Assume conditions \ref{cond:smoothness}--\ref{cond:kernel} and (C4)-(C5) in Appendix A.3 hold. 
	Then, we have $\hat h_\textup{OS}/ h_\textup{opt} \to_p 1$ as $n\to \infty$ and $\hat h_\textup{PLT} / h_\textup{opt} \to_p 1 $ as $n_0\to \infty$.
\end{theorem}

The proof of Theorem \ref{thm:bandwidth} is given in Appendix A.3. 
By Theorem \ref{thm:bandwidth}, we know that the two bandwidth selectors should be ratio consistent for the optimal bandwidth $\hat h_\textup{opt}$, as long as the local sample size $n$ or the pilot sample size $n_0$ diverges to infinity as $N\to \infty$.
Our exhaustive simulation experiments further confirm their excellent finite sample performance.

\section{Some Extensions}

In this section, we consider two extensions of the GPA method. They are, respectively, high order polynomial interpolation based GPA and multivariate GPA.

\subsection{High Order Polynomial Interpolation}
\label{subsec:GPA_nu}

One way to further improve the proposed GPA method is to replace the linear interpolation method with higher order polynomial interpolation.
Specifically, recall that $\hat\mu(x_j^*)$ is the NW estimator on the $j$-th grid point. 
For an arbitrary $x\in (0,1)$, we first find its $\nu+1$ nearest grid points.
Without loss of generality, we assume that the $\nu+1$ grid points are $x_{j}^*,\dots, x_{j+\nu}^*$ for some $j$.
We then estimate $\mu(x)$ by the $\nu$-th order polynomial interpolation as
\begin{align*}
	\hat\mu_{\textup{PGPA},\nu}(x) = \sum_{k =j}^{j+\nu} q_k(x) \hat \mu(x_j^*),
\end{align*}
where $q_{k}(x)=\prod_{i=j, i\ne k}^{j+\nu}(x-x_i^*)/(x_k^*-x_i^*) $ for $j\le k\le j+\nu$ are the corresponding Lagrange interpolation coefficients \citep{suli2003introduction}. 
To investigate the theoretical properties of $\hat\mu_{\textup{PGPA},\nu}(x)$, we assume the following technical conditions.
\begin{enumerate}[label=(A\arabic*)]
	\item (\textsc{Smoothness Condition}) The probability density function $f(\cdot)$ is sufficiently smooth so that its $(\nu+1)$-th order derivative $ f^{(\nu+1)}(\cdot)$ is continuous. The mean function $\mu(\cdot)$ is sufficiently smooth so that its $(\nu+1)$-th order derivative $\mu^{(\nu+1)} (\cdot)$ is continuous. \label{cond:smoothness_nu}
	
	\item (\textsc{Higher Order Kernel}) Assume that $K(\cdot)$ is a $(\nu+1)$-th order continuous kernel function with compact support $[-1,1]$. That is $K(\cdot)$ satisfies (i) $\kappa_0=\int K(u)du=1$, (ii) $\kappa_r = \int u^r K(u) du = 0$ for $1\le r\le \nu$, and (iii) $\int |u^{\nu+1} K(u)| du < \infty$.\label{cond:kernel_nu}
	
	\item (\textsc{Optimal Bandwidth}) Assume $h = CN^{-1/(2\nu+3)}$ for some positive constant $C$.\label{cond:bandwidth_nu}
\end{enumerate}
{
	The smoothness conditions in \ref{cond:smoothness_nu} are needed for $\nu$-th order polynomial interpolation. Note that condition \ref{cond:smoothness_nu} is the same as condition \ref{cond:smoothness} if $\nu=1$. Condition \ref{cond:kernel_nu} introduces the higher order kernel function, which is useful to further reduce the bias of the kernel estimator \citep{fan1992bias, ullah1999nonparametric}. The construction of higher order kernel functions can be found in \cite{fan1992bias} and \cite{hansen2005exact}. 
	Last, condition \ref{cond:bandwidth_nu} assumes that the bandwidth should be of the compatibly optimal rate when a higher order kernel is used \citep{ullah1999nonparametric}.
	By this bandwidth, the optimal convergence rate achieved by a global NW estimator in terms of the asymptotic mean squared error is given by $O(1/(Nh)) = O(N^{-2(\nu+1) / (2\nu+3)})$ under the conditions \ref{cond:smoothness_nu} and \ref{cond:kernel_nu}.
	The theoretical properties of the $\nu$-th order polynomial interpolation based GPA estimator $\hat\mu_{\textup{PGPA},\nu}(x)$ are then summarized in the following theorem.
}

\begin{theorem}\label{thm:GPA_nu}
	Assume conditions \ref{cond:smoothness_nu}--\ref{cond:bandwidth_nu} hold and $f(x)>0$. Further, assume that $Jh\to\infty$ as $N\to\infty$.
	Then, we have $\hat\mu_{\textup{PGPA},\nu}(x) - \mu(x) = \tilde Q_{\nu,0} + \tilde Q_{\nu,1}  + \tilde \mQ_{\nu}$,
	where $\tilde Q_0$ is a non-stochastic term satisfying $\tilde Q_{\nu,0} = O(1/J^{\nu+1})$, $ \sqrt{Nh}\Big\{\tilde  Q_{\nu,1} - B_{\nu}(x)h^{\nu+1}\Big\} \to_d \mN\Big(0, V(x) \Big)$, and $\tilde \mQ_{\nu} = o_p(1/\sqrt{Nh})$.
	Here, $B_\nu(\cdot)$ is defined in \textup{(S.11)} in Appendix A.4, 
	and $V(\cdot)$ is the same as that in \eqref{eq:B&V}.
	In addition to that, we have
	\begin{align*}
		\sqrt{Nh}\Big\{\hat\mu_{\textup{PGPA},\nu}(x) - \mu(x) - B_{\nu}(x)h^{\nu+1}\Big\}\to_d \mN\Big(0, V(x) \Big).
	\end{align*}
\end{theorem}
The proof of Theorem \ref{thm:GPA_nu} is given in Appendix A.4. 
By Theorem \ref{thm:GPA_nu}, we find that, as long as $Jh\to \infty$ as $N\to\infty$, the $\nu$-th order polynomial interpolation based GPA estimator $\hat\mu_{\textup{PGPA},\nu}(x)$ can achieve the same statistical efficiency as the global estimator. Recall that $h = CN^{-1/(2\nu+3)}$ for some constant $C>0$. This implies that the number of grid points should satisfy $J \gg N^{1/(2\nu+3)}$. This is in line with the results in Theorem \ref{thm:GPA_asymptotic} for linear interpolation based GPA estimator $\hat\mu_{\textup{GPA}}(x)$ with $\nu=1$. 
Furthermore, we can see that fewer number of grid points are required as long as $\nu>1$. This suggests that higher order polynomial interpolation leads to further reduced communication cost without any sacrifice of asymptotic statistical efficiency. 
However, this benefit comes at the cost of imposing more stringent smoothness conditions on the mean function $\mu(\cdot)$ and density function $f(\cdot)$. 
In practice, we often do not know the smoothness of the unknown functions $\mu(\cdot)$ and $f(\cdot)$. To select an appropriate interpolation order $\nu$, we can randomly select a validation set $\mV$ from the whole dataset. 
Then the performance for different interpolation orders $1\le \nu \le \nu_{\max}$ can be evaluated on $\mV$, where $\nu_{\max}\ge 1$ is a prespecified integer. 
Specifically, for each candidate order $\nu$, we first choose a $(\nu+1)$-th order kernel function that satisfies condition \ref{cond:kernel_nu}. Subsequently, the $\nu$-th order polynomial interpolation based GPA estimator $\hat\mu_{\textup{PGPA},\nu}$ can be used to make prediction for the validation set $\mV$. Finally, we can choose the order $\nu$ minimizing the prediction error on $\mV$.

\subsection{Multivariate GPA}

We now consider how to extend the GPA estimator from the situation with a univariate predictor (i.e. $p=1$) to a more general situation with a multivariate predictor (i.e., $p>1$). 
We assume that $X_i=(X_{i1},\dots,X_{ip})^\top$ is a random vector compactly supported on $[0,1]^p \subset\mR^p$. 
Then for any fix point $\bx=(x_1,\dots,x_p)^\top \in [0,1]^p $, the multivariate NW estimator can be defined as \citep{li2007nonparametric}: 
\begin{align}\label{eq:nw_multi}
	\hat\mu(\bx) = \frac{\sum_{i =1}^N \bK_h(X_i - \bx)  Y_i }{\sum_{i =1}^N \bK_h(X_i - \bx)} = \frac{\sum_{m=1}^M\sum_{i \in \mS_m} \bK_h(X_i - \bx)  Y_i }{\sum_{m=1}^M\sum_{i \in \mS_m} \bK_h(X_i - \bx)},
\end{align}
where $\bK_h(X_i- \bx) = h^{-p}\prod_{s=1}^p K\big((X_{is} - x_s)/h\big)$, and $K(\cdot)$ is a univariate kernel function satisfying condition \ref{cond:kernel}. 
Similarly, let $J$ be a prespecified positive integer. 
We then design the grid points set as $\mG= \{\bx_j^* = (j_1/J,\dots, j_p/J)^\top \in \mR^p: 0\le j_1,\dots,j_p\le J\}$. This leads to a total of $|\mG|=(J+1)^p$ grid points. 
We can then compute the global estimator $\hat\mu(\bx_j^*)$ according to \eqref{eq:nw_multi} for each grid point $\bx_j^* \in \mG$ in a distributed way. 
For any $\bx=(x_1,\dots,x_p)^\top \in (0,1)^p $, we can find its $p+1$ nearest linearly independent grid points $\bx_{j_1}^*, \dots, \bx_{j_{p+1}}^* \in \mG$.
Let $\mC(\bx) = \big\{\sum_{k=1}^{p+1} w_k\bx_{j_k}^* : \sum_{k=1}^{p+1} w_k = 1 \textup{ and } w_k \ge 0 \textup{ for each } k \big\} $ be the $p$-simplex with these $p+1$ grid points as vertices. We can verify that $\bx \in \mC(\bx)$. 
Subsequently, we can predict $\mu(\bx)$ by multivariate linear interpolation as 
\begin{align*}
	\hat\mu_{\textup{MGPA},p}(\bx) = \sum_{k=1}^{p+1} q_k(\bx) \hat \mu(\bx_{j_k}^*),
\end{align*}
where $q_k(\bx),\, 1\le k\le p+1$ are interpolation coefficients \citep{waldron1998error}. 
To investigate the theoretical properties of multivariate GPA estimator $\hat\mu_{\textup{MGPA},p}(\bx)$, we assume the following technical conditions.
\begin{enumerate}[label=(A\arabic*)]
	\setcounter{enumi}{3}
	\item (\textsc{Smoothness Condition}) The probability density function $f(\cdot)$ is sufficiently smooth so that its second-order derivative $\ddot f(\bx)= \partial^2 f(\bx) / (\partial \bx\partial \bx^\top) \in \mR^{p\times p}$ is continuous. \label{cond:smoothness_p}
	
	The mean function $\mu(\cdot)$ is sufficiently smooth so that its its second-order derivative $\ddot \mu(\bx)= \partial^2 \mu(\bx) / (\partial \bx\partial \bx^\top)\in \mR^{p\times p}$ is continuous. 
	\item (\textsc{Optimal Bandwidth}) Assume $h = C N^{-1/(p+4)}$ for some positive constant $C$.\label{cond:bandwidth_p}
\end{enumerate}
The condition \ref{cond:smoothness_p} is a multivariate version of the univariate smoothness condition \ref{cond:smoothness}. This condition is needed for multivariate linear interpolation.
Condition \ref{cond:bandwidth_p} assumes that the bandwidth should be of the compatible optimal rate for the multivariate NW estimator \citep{hardle2004nonparametric,li2007nonparametric}. 
By this bandwidth, the optimal convergence rate achieved by a global multivariate NW estimator in terms of the asymptotic mean squared error is given by $O(1/(Nh^p)) = O(N^{-4/(p+4)})$ under conditions \ref{cond:kernel} and \ref{cond:smoothness_p}. 
The theoretical properties of the multivariate GPA estimator $\hat\mu_{\textup{MGPA},p}(\bx)$ are then summarized in the following theorem. 

\begin{theorem} \label{thm:GPA_multi}
	Assume conditions \ref{cond:kernel}, \ref{cond:smoothness_p} and \ref{cond:bandwidth_p} hold and $f(\bx)>0$. 
	Further, assume that $Jh\to\infty$ as $N\to\infty$. 
	Then we have $\hat\mu_{\textup{MGPA},p}(\bx) - \mu(\bx) = \tilde Q_{p,0} + \tilde Q_{p,1}  + \tilde \mQ_p$, where $\tilde Q_{p,0} $ is a non-stochastic term satisfying $\tilde Q_{p,0} = O(1/J^2)$, $\sqrt{Nh^{p}}\big\{ \tilde Q_{p,1} - B_p(\bx)h^2 \big\} \to_d \mN\Big(0,V_p(\bx)\Big)$, 
	and $\tilde \mQ_p = o_p(1 / \sqrt{Nh^p})$. 
	Here, $B_p(\cdot)$ and $V_p(\cdot)$ are defined in \textup{(S.13)} in Appendix A.5. 
	In addition to that, we have
	\begin{align*}
		\sqrt{Nh^p}\Big\{\hat\mu_{\textup{MGPA},p}(\bx) - \mu(\bx) - B_p(\bx)h^2\Big\}\to_d \mN\Big(0, V_p(\bx) \Big).
	\end{align*}
\end{theorem}
The proof of Theorem \ref{thm:GPA_multi} is given in Appendix A.5. 
By Theorem \ref{thm:GPA_multi}, we know that the global statistical efficiency can be achieved if $Jh\to \infty$ as $N\to\infty$. Note that $h = C N^{-1/(p+4)}$ for some constant $C>0$. This implies that the number of grid points should satisfy $|\mG| = (J+1)^4 \gg N^{p/(p+4)}$. This is also in line with the results in Theorem \ref{thm:GPA_asymptotic} for univariate linear interpolation based GPA estimator $\hat\mu_{\textup{GPA}}(x)$ with $p=1$.

\section{Numerical Studies}

In this section, we conduct a number of numerical studies to demonstrate the performance of our proposed GPA methods.
We first investigate the performance of the different distributed kernel estimators in terms of the prediction accuracy and the time cost.
Next, we verify the finite performance of the two bandwidth selectors.
Finally, we compare our GPA method with other competing methods (i.e., the basis expansion method and the RKHS based method).
Specifically, we generate $X_i$ according to a probability density function $f(x)$ defined on $[0,1]$. Two different distributions are considered. They are, respectively, the uniform distribution $\text{Unif}(0,1)$ and the Beta distribution $\text{Beta}(2,3)$.
Once $X_i$ is obtained, $Y_i$ is generated according to $Y_i = \mu(X_i) + \varepsilon_i$ with $\varepsilon_i$ simulated from the standard normal distribution.
Here, two different mean functions are considered. They are, respectively,
$\mu_1(x) = 4(x-0.5) + 2\exp\{-128(x-0.5)^2 \}$, and 
$\mu_2(x) = \sin\{8(x-0.5)\} + 2\exp\{-128(x-0.5)^2 \}$.
The two mean functions are revised from \cite{fan1996local}. 
Additional numerical experiments, in comparison to other methods, can be found in Appendix E of the supplementary materials.

\subsection{Different Distributed Kernel Estimators}
\label{subsec:dist_KE}
Once the full training sample $\{(X_i,Y_i):\ 1\le i\le N \}$ is simulated, we consider two different local sample allocation strategies as follows.

    \textsc{Strategy 1. (Random Partition)}  We randomly and evenly divide the whole sample $\mS$ into $M$ disjoint subsets $\mS_m\ (1\le m\le M)$ such that the data distribution on different machines are the same. This is an ideal scenario for one-shot type estimators.
    
    \textsc{Strategy 2. (Nonrandom Partition)} We assign the samples to different machines according to the covariate value. Specifically, let $X_{(1) } \le X_{(2) } \le \cdots\le X_{(N) } $ be the order statistic of $X_i$s. Thus, $(i)$ represents the index of the $i$-th smallest sample point for each $1\le i\le N$. We then assign $\mS_m = \{(i): (m-1)n+1 \le i \le mn \}$ to the $m$-th local machine, where $n = N / M$ is the local sample size. In this case, the data distributions on different machines are very different.

Subsequently, different estimators are computed, including the one-shot estimator $\hat \mu_\textup{OS}$, the GPA estimator $\hat \mu_\textup{GPA}$, and the global estimator $\hat \mu$ with a distributed implementation described in Section \ref{subsec:global&OS}.
For all the estimators, the Epanechnikov kernel $K(u) = (3/4) (1-u^2)I(|u|\le 1)$ is adopted, where $I(\cdot)$ is the indicator function. The asymptotically optimal bandwidth $h_{\textup{opt}}$ defined in \eqref{eq:bandwidth_opt} is used for all estimators. 
For a comprehensive evaluation, three sample sizes ($N = 1\times 10^4, 2\times 10^4, 5\times 10^4$) are studied.
For the GPA estimator $\hat\mu_\textup{GPA}$, a total of $J + 1= [h_\textup{opt}^{-1}\log\log N]+1$ grid points are equally spaced on the interval $[0,1]$, where $[r]$ denotes the integer part of $r\in\mR$. 
To assess the prediction performance of each estimator, we independently generate a testing sample set $\big\{ \big(X_i^*,\mu(X_i^*)\big): 1\le i\le N^*\big\}$ according to the same probability density function $f(x)$ and mean function $\mu(x)$ as the training sample. The testing sample size is set to be $N^*=N/2$ and the number of machine is $M=50$.
Let $\hat \mu^{(b)}$ be one particular estimator (e.g., the global estimator) obtained in the $b$-th simulation iteration with $1\le b\le B = 100$. We then compute the root mean squared error (RMSE) as $\text{RMSE}(\hat\mu^{(b)}) = \sqrt{ (N^*)^{-1}\sum_{i=1}^{N^*} \big\{\hat\mu^{(b)}(X_i^*) - \mu(X_i^*) \big\}^2 }$.
In addition, both the computation and communication times are recorded. 
The averaged RMSE values, computation times ($t_\textup{cp}$, in seconds), and communication times ($t_\textup{cm}$, in seconds) over $B=100$ replications are presented in Tables \ref{tab:random} and \ref{tab:nonrandom} for the two different local sample allocation strategies, respectively.

\begin{table}[htbp]
	\caption{The averaged RMSE values, computation times and communication times for three different estimators by the \textsc{Random Partition} strategy in four different settings.}
	\label{tab:random}
	\centering
	\begin{tabular}{cccc|ccc|ccc} 
		\toprule
		 & \multicolumn{3}{c|}{RMSE} & \multicolumn{3}{c|}{$t_\textup{cp}$ (sec.)} & \multicolumn{3}{c}{$t_\textup{cm}$ (sec.)} \\
        Est.&  $\hat \mu$ & $\hat \mu_\textup{OS}$ & $\hat \mu_\textup{GPA}$  &  $\hat \mu$ & $\hat \mu_\textup{OS}$ & $\hat \mu_\textup{GPA}$ &  $\hat \mu$ & $\hat \mu_\textup{OS}$ & $\hat \mu_\textup{GPA}$\\
		\hline
		\multicolumn{1}{c}{$N\ (\times 10^4)$} 
		&\multicolumn{8}{c}{\small \textsc{Setting 1}: $\mu = \mu_1,\ f\sim \text{Unif}(0,1)$} \\
		$1$& 0.046 & 0.048 & 0.046 & 0.108 & 0.097 & 0.004 & 0.177 & 0.174 & 0.129 \\  
		$2$& 0.034 & 0.034 & 0.034 & 0.431 & 0.453 & 0.005 & 0.276 & 0.264 & 0.118 \\  
		$5$& 0.024 & 0.024 & 0.024 & 2.743 & 2.831 & 0.013 & 0.585 & 0.622 & 0.122 \\  
		\hline
		
		\multicolumn{1}{c}{ } 
		&\multicolumn{8}{c}{\small \textsc{Setting 2}: $\mu = \mu_1,\ f\sim \text{Beta}(2,3)$} \\
		$1$& 0.047 & 0.049 & 0.048 & 0.114 & 0.075 & 0.003 & 0.186 & 0.179 & 0.126 \\   
		$2$& 0.035 & 0.037 & 0.035 & 0.440 & 0.465 & 0.006 & 0.272 & 0.276 & 0.120 \\   
		$5$& 0.025 & 0.025 & 0.025 & 2.821 & 2.846 & 0.014 & 0.591 & 0.645 & 0.122 \\  
		\hline
		
		\multicolumn{1}{c}{ } 
		&\multicolumn{8}{c}{\small \textsc{Setting 3}: $\mu = \mu_2,\ f\sim \text{Unif}(0,1)$} \\
		$1$& 0.048 & 0.049 & 0.048 & 0.107 & 0.075 & 0.003 & 0.182 & 0.178 & 0.128 \\   
		$2$& 0.035 & 0.036 & 0.035 & 0.440 & 0.458 & 0.006 & 0.268 & 0.254 & 0.120 \\   
		$5$& 0.025 & 0.025 & 0.025 & 2.706 & 2.632 & 0.012 & 0.629 & 0.573 & 0.122 \\ 
		\hline
		
		\multicolumn{1}{c}{ } 
		&\multicolumn{8}{c}{\small \textsc{Setting 4}: $\mu = \mu_2,\ f\sim \text{Unif}(2,3)$} \\
		$1$& 0.048 & 0.050 & 0.048 & 0.110 & 0.083 & 0.004 & 0.182 & 0.176 & 0.127 \\ 
		$2$& 0.036 & 0.038 & 0.036 & 0.456 & 0.470 & 0.006 & 0.265 & 0.262 & 0.122 \\ 
		$5$& 0.025 & 0.026 & 0.025 & 2.699 & 2.845 & 0.014 & 0.604 & 0.583 & 0.122 \\ 
		\bottomrule
	\end{tabular}
\end{table}

\begin{table}[htbp]
	\caption{The averaged RMSE values, computation times and communication times for three different estimators by the \textsc{Nonrandom Partition} strategy in four different settings.}
	\label{tab:nonrandom}
	\centering
	\begin{tabular}{cccc|ccc|ccc} 
		\toprule
		 & \multicolumn{3}{c|}{RMSE} & \multicolumn{3}{c|}{$t_\textup{cp}$ (sec.)} & \multicolumn{3}{c}{$t_\textup{cm}$ (sec.)} \\
        Est.&  $\hat \mu$ & $\hat \mu_\textup{OS}$ & $\hat \mu_\textup{GPA}$  &  $\hat \mu$ & $\hat \mu_\textup{OS}$ & $\hat \mu_\textup{GPA}$ &  $\hat \mu$ & $\hat \mu_\textup{OS}$ & $\hat \mu_\textup{GPA}$\\
		\hline
		\multicolumn{1}{c}{$N\ (\times 10^4)$} 
		&\multicolumn{8}{c}{\small \textsc{Setting 1}: $\mu = \mu_1,\ f\sim \text{Unif}(0,1)$} \\
		$1$& 0.044 & NA & 0.045 & 0.090 & NA & 0.004 & 0.170 & NA & 0.126 \\  
		$2$& 0.034 & NA & 0.034 & 0.435 & NA & 0.006 & 0.260 & NA & 0.111 \\  
		$5$& 0.024 & NA & 0.024 & 2.749 & NA & 0.012 & 0.590 & NA & 0.111 \\ 
		\hline
		
		\multicolumn{1}{c}{ } 
		&\multicolumn{8}{c}{\small \textsc{Setting 2}: $\mu = \mu_1,\ f\sim \text{Beta}(2,3)$} \\
		$1$& 0.045 & NA & 0.046 & 0.087 & NA & 0.004 & 0.177 & NA & 0.129 \\ 
		$2$& 0.036 & NA & 0.036 & 0.448 & NA & 0.006 & 0.267 & NA & 0.112 \\ 
		$5$& 0.025 & NA & 0.025 & 2.734 & NA & 0.014 & 0.620 & NA & 0.110 \\ 
		\hline
		
		\multicolumn{1}{c}{ } 
		&\multicolumn{8}{c}{\small \textsc{Setting 3}: $\mu = \mu_2,\ f\sim \text{Unif}(0,1)$} \\
		$1$& 0.046 & NA & 0.046 & 0.088 & NA & 0.004 & 0.176 & NA & 0.127 \\ 
		$2$& 0.035 & NA & 0.035 & 0.434 & NA & 0.005 & 0.267 & NA & 0.116 \\ 
		$5$& 0.025 & NA & 0.025 & 2.738 & NA & 0.012 & 0.558 & NA & 0.116 \\ 
		\hline
		
		\multicolumn{1}{c}{ } 
		&\multicolumn{8}{c}{\small \textsc{Setting 4}: $\mu = \mu_2,\ f\sim \text{Unif}(2,3)$} \\
		$1$& 0.047 & NA & 0.047 & 0.089 & NA & 0.004 & 0.172 & NA & 0.129 \\  
		$2$& 0.036 & NA & 0.036 & 0.448 & NA & 0.006 & 0.273 & NA & 0.114 \\  
		$5$& 0.026 & NA & 0.026 & 2.710 & NA & 0.016 & 0.676 & NA & 0.113 \\
		\bottomrule
	\end{tabular}
\end{table}

Table \ref{tab:random} reports the detailed simulation results for the random partition strategy. By Table \ref{tab:random} we obtain the following interesting observations. 
First, we find that the RMSE values steadily decrease as the training sample size $N$ increases for all estimators. This is expected because all the estimators are consistent for the mean function $\mu(\cdot)$. In addition, the RMSE values of the one-shot estimator are slightly larger than those of the global estimator when whole sample size $N=1\times 10^4$. This is expected because the one-shot estimator generally requires sufficient local samples, however, the local sample size is only $n=N/M=200$ in this case.
In contrast, the GPA estimator performs almost identically to the global estimator in terms of the RMSE values. This is expected because according to Theorem \ref{thm:GPA_asymptotic}, the GPA estimator should be asymptotically as efficient as the global estimator. This is a desirable property independent of local sample size.
Second, regarding computation time (i.e., $t_\textup{cp}$), it can be observed that the computation times for the global and the one-shot estimators are comparable, and substantially higher than those of the GPA estimator.
This is particularly true when the sample size $N$ is large. For example, when $N=5\times 10^4$, the computation time for the GPA estimator takes no more than 0.02 seconds, while the other two estimators require over 2 seconds.
Lastly, regarding communication time (i.e., $t_\textup{cm}$), we can find that the global and the one-shot estimators have similar pattern: as the sample size increases, the communication times grow from about 0.18 seconds to about 0.6 seconds. 
In contrast, the communication times for the GPA estimator remain around 0.12 seconds. 
This is because the GPA estimator only transmits statistics related to a small amount of grid points during the training phase, while the other two estimators need to transmit a large number of testing observations. 

Table \ref{tab:nonrandom} reports the results for the nonrandom partition strategy. 
By Table \ref{tab:nonrandom} we find that NA values are reported for the one-shot estimator $\hat \mu_{\OS}$. This is because the support of $X_i$s on the local machine may not cover all $X_i^*$s in the testing sample in this case. Consequently, these machines cannot make predictions for these testing observations, and thus $\hat \mu_{\OS}$ is not well defined. In contrast, both the global and the GPA estimators can still make predictions for testing samples. Other results regarding RMSE values, computation and communication times are similar to those in Table \ref{tab:random}.

To further demonstrate the practical feasibility of the GPA method on massive datasets, we present another simulation study. 
For this study, we set $N = 10^8, 10^9, 10^{10}$. 
Note that $N= 10^{10}$ is 10 billion! Simulating a dataset of this size is extremely challenging for a single computer. As a result, we turn to a powerful distributed system for simulation. The distributed system used here is a Spark-on-YARN cluster with one master node and six computing nodes. Each computing node has 32 cores and 64 GB of RAM.
Thus, we have a total of 192 cores for computation. We generate the data according to Setting 1 described above but with testing sample size $N^* = 10,000$. Once the data are simulated, they are placed on a Hadoop file system (HDFS) so that the data are stored distributed on different local computers. We then randomly split the data into a total of $M = 128$ partitions. 
The advantage of the GPA method is that the distributed computation only needs to be conducted on gird points with $J = [h_\textup{opt}^{-1}\log\log N]$ having nothing to do with $N^* = 10,000$.
Next, the prediction of the testing sample is implemented on the master machine. Table \ref{tab:massive} reports the simulation results. 
The total time required required by the Spark system increases from 0.09 hours with $N=10^8$ to 2.38 hours with $N = 10^{10}$.
This is a single time consumption for our GPA methods. 
Once the grid point estimation is accomplished, the prediction for $N^* = 10,000$ testing samples is computed on the master machine, which takes no more than 1 second. However, for the global and the one-shot estimators, the distributed computation needs to be conducted $N^*$ times. When $N^*$ is large, this is extremely time consuming.

\begin{table}[htbp]
	\caption{The total time cost (in hours) for computing all grid points on the Spark system. The RMSE values of the GPA estimator are also reported. 
	}
	\label{tab:massive}
	\centering
	\begin{tabular}{ *{2}{>{\centering\arraybackslash}p{0.1\textwidth}}>{\centering\arraybackslash}p{0.15\textwidth} *{2}{>{\centering\arraybackslash}p{0.3\textwidth}} }
		\toprule
		$N$ & $J$ & Time & RMSE ($\times 10^{-2}$)\\
		\hline
		$10^8$    & 509  & 0.09 & 0.139 \\
		$10^9$    & 838  & 0.24 & 0.051 \\
		$10^{10}$ & 1375 & 2.38 & 0.020 \\
		\bottomrule
	\end{tabular}
\end{table}

\subsection{The Two Bandwidth Selectors}

In this subsection, we discuss the bandwidth selection problem. The models are the same as above. Three sample sizes ($N = 1\times 10^4, 2\times 10^4, 5\times 10^4$) are considered, and the number of local machines is fixed to be $M = 50$. We estimate $\hat h_{\textup{OS}}$ and $\hat h_\textup{PLT}$ via cross-validation described in Section \ref{subsec:bandwidth}. 
For $\hat h_\textup{PLT}$, the pilot sample sizes are chosen to be $n_0 =(1000,1500,3000)$ according to different full sample sizes. 
Denote $\hat h^{(b)}$ as the estimated bandwidth in the $b$-th replication for $\hat h_{\textup{OS}}$ or $\hat h_\textup{PLT}$. We then compute the mean relative absolute error (MRAE) of $\hat h$ as $\text{MRAE}(\hat h) =  B^{-1}\sum_{b=1}^B|\hat h^{(b)}-h_\textup{opt}|/h_\textup{opt}$. The results are presented in Table \ref{tab:bandwidth}.
By Table \ref{tab:bandwidth}, we find that the MRAE values steadily decrease as $N$ increases for both bandwidth selectors. This consists with Theorem \ref{thm:bandwidth} that both two bandwidth selectors are ratio consistent of the asymptotically optimal bandwidth.
The MRAE values of $\hat h_{\textup{OS}}$ are smaller than those of $\hat h_\textup{PLT}$.
This is expected since $\hat h_{\textup{OS}}$ is computed based on the full sample, while $\hat h_\textup{PLT}$ is computed based on the pilot sample with a much smaller size.
However, the drawback of the $\hat h_{\textup{OS}}$ is that it requires a homogeneous data distribution on different workers. 
In contrast, $\hat h_\textup{PLT}$ is free of such an assumption.

\begin{table}[h]
	\caption{The mean relative absolute error (MRAE) of the two bandwidth selectors in four different settings.}
	\label{tab:bandwidth}
	\centering
	\begin{tabular}{ >{\centering\arraybackslash}p{0.1\textwidth}| *{3}{>{\centering\arraybackslash}p{0.115\textwidth}}|*{3}{>{\centering\arraybackslash}p{0.115\textwidth}} }
		\toprule
		$N\  (\times 10^4)$& $1$ & $2$ & $5$ & $1$ & $2$ & $5$ \\
		\hline
		Band.&\multicolumn{3}{c|}{\small \textsc{Setting 1}: $\mu = \mu_1,\ f\sim \text{Unif}(0,1)$} 
		&\multicolumn{3}{c}{\small \textsc{Setting 2}: $\mu = \mu_1,\ f\sim \text{Beta}(2,3)$}\\
		$\hat h_\textup{OS}$ & 0.099 & 0.044 & 0.024 & 0.060 & 0.036 & 0.023 \\
		$\hat h_\textup{PLT}$& 0.186 & 0.159 & 0.153 & 0.172 & 0.155 & 0.146 \\
		\hline
		&\multicolumn{3}{c|}{\small \textsc{Setting 3}: $\mu = \mu_2,\ f\sim \text{Unif}(0,1)$} 
		&\multicolumn{3}{c}{\small\textsc{Setting} 4: $\mu = \mu_2,\ f\sim \text{Beta}(2,3)$}\\
		$\hat h_\textup{OS}$ & 0.067 & 0.034 & 0.021 & 0.044 & 0.030 & 0.022\\
		$\hat h_\textup{PLT}$& 0.173 & 0.153 & 0.147 & 0.164 & 0.146 & 0.141 \\
		\bottomrule
	\end{tabular}
\end{table}

To assess the performance of the different bandwidths, we compute the one-shot estimator $\hat\mu_\textup{OS}$ and the GPA estimator $\hat\mu_\textup{GPA}$ based on both $\hat h_{\textup{OS}}$ and $\hat h_\textup{PLT}$.
The RMSE values of these estimators are then computed based on an independently generated testing sample of size $N^* = 1,000$.
For comparison, we also compute the RMSE values of the global estimator $\hat\mu$ based on the asymptotically optimal bandwidth $h_\textup{opt}$.
The RMSE values averaged over $B=100$ replications are presented in Table \ref{tab:RMSE_OS_PLT}.
As seen in Table \ref{tab:RMSE_OS_PLT}, the RMSE values decrease as the sample size $N$ increases for all estimators. 
In addition, all the estimators exhibit performance comparable to that of the global estimator when $N = 5\times 10^4$.
This confirms the usefulness of the two bandwidth selectors.
We can also see that the $\hat h_{\textup{OS}}$-based estimators perform better than the $\hat h_\textup{PLT}$-based estimators in terms of RMSE values.
This is expected because according to the results in Table \ref{tab:bandwidth}, $\hat h_{\textup{OS}}$ is closer to the asymptotically optimal bandwidth. 
Furthermore, the GPA estimator generally performs better than the one-shot estimator in terms of the RMSE values.

\begin{table}[htbp]
	\caption{The averaged RMSE values of the global estimator, the one-shot estimator, and the GPA estimator with different bandwidths in four different settings.}
	\label{tab:RMSE_OS_PLT}
	\centering
	\begin{tabular}{ >{\centering\arraybackslash}p{0.15\textwidth} *{5}{>{\centering\arraybackslash}p{0.12\textwidth}} }
		\toprule
		Est. & $\hat \mu$ & $\hat \mu_\textup{OS}$ & $\hat \mu_\textup{GPA}$ & 
		$\hat \mu_\textup{OS}$ & $\hat \mu_{\textup{GPA}}$  \\
		
		Band. & $h_\textup{opt}$ & $\hat h_\textup{OS}$  & $\hat h_\textup{OS}$  &  $\hat h_\textup{PLT}$& $\hat h_\textup{PLT}$\\
		\hline
		\multicolumn{1}{c}{$N\ (\times 10^4)$} 
		&\multicolumn{5}{c}{\small \textsc{Setting 1}: $\mu = \mu_1,\ f\sim \text{Unif}(0,1)$} \\
		$1$&  0.046 & 0.048 & 0.048  & 0.051 & 0.049  \\
		$2$&  0.034 & 0.035 & 0.034  & 0.036 & 0.036  \\
		$5$&  0.024 & 0.024 & 0.024  & 0.025 & 0.025  \\
		\hline
		
		\multicolumn{1}{c}{ } 
		&\multicolumn{5}{c}{\small \textsc{Setting 2}: $\mu = \mu_1,\ f\sim \text{Beta}(2,3)$} \\
		$1$&  0.047 & 0.049 & 0.048  & 0.051 & 0.050  \\
		$2$&  0.036 & 0.037 & 0.036  & 0.038 & 0.037  \\
		$5$&  0.025 & 0.025 & 0.025  & 0.026 & 0.026  \\
		\hline
		
		\multicolumn{1}{c}{ } 
		&\multicolumn{5}{c}{\small \textsc{Setting 3}: $\mu = \mu_2,\ f\sim \text{Unif}(0,1)$} \\
		$1$&  0.048 & 0.050 & 0.049  & 0.052 & 0.051 \\
		$2$&  0.035 & 0.036 & 0.036  & 0.037 & 0.037 \\
		$5$&  0.025 & 0.025 & 0.025  & 0.026 & 0.026 \\
		\hline
		
		\multicolumn{1}{c}{ } 
		&\multicolumn{5}{c}{\small \textsc{Setting 4}: $\mu = \mu_2,\ f\sim \text{Unif}(2,3)$} \\
		$1$&  0.048 & 0.050 & 0.049  & 0.052 & 0.050  \\
		$2$&  0.037 & 0.038 & 0.037  & 0.039 & 0.038  \\
		$5$&  0.025 & 0.026 & 0.025  & 0.026 & 0.026  \\
		\bottomrule
	\end{tabular}
\end{table}

\section{Real Data Analysis}
\label{sec: application}

\subsection{Airline Dataset}

We first apply our method to the US airline dataset (\url{http://stat-computing.org/dataexpo/2009}). 
It consists of flight arrival and departure details for all commercial flights within the US from October 1987 to April 2008. The dataset contains over 123 million records, occupying 12 GB of space. To demonstrate the usefulness of our proposed methodology, we choose the arrival delay (\textsf{ArrDelay}, recorded in minutes) of the flight as the response. However, the response \textsf{ArrDelay} has a very heavy-tailed distribution and contains negative values (early arrival). Therefore, we transform the response via a modified logarithmic transformation, $\sign(y)\cdot\log (1 + |y|)$, preserving both sign and scale information. We next choose the scheduled departure time (\textsf{CRSDepTime}) as the predictor.
After data cleaning, a total of $120,947,440$ observations are retained. 
Among them, $N=120,790,000$ observations are randomly selected as the training sample set, $n_\textup{val}=10,000$ observations are randomly selected as the validation set $\mV$, and the remaining $N^* = 147,440$ observations form the testing sample set. 
A total of $M=50$ machines are used. 
We first select an appropriate interpolation order $\nu\ (1\le \nu \le \nu_{\max}=4)$ using the procedures described at the end of Section \ref{subsec:GPA_nu}.
Specifically, for each order $\nu$, we compute the root mean prediction error (RMPE) of the $\nu$-th order polynomial interpolation based GPA estimator $\hat\mu_{\GPA, \nu}$ on the validation set $\mV$ as $\textup{RMPE}(\hat\mu_{\GPA, \nu}) =  \sqrt{ (n_\textup{val})^{-1}\sum_{i\in\mV} \big\{\hat\mu_{\GPA, \nu}(X_i) - Y_i\big\}^2 }$. Here, the pilot sample-based bandwidth selector $\hat h_\textup{PLT}$ is used with a pilot sample of size $n_0 = 10,000$, and $J$ is set to be $J = [24\hat h_\textup{PLT}^{-1} \log \log N  ] $ accordingly. 
We find that the linear interpolation based GPA estimator $\hat\mu_{\GPA,1}$ (i.e., $\hat\mu_{\GPA}$) shows the smallest RMPE value. The corresponding bandwidth is $\hat h_\textup{PLT} =  0.549$, and a total of $J+1=164$ grid points are equally spaced within the interval $[0, 24]$.
We then use $\hat\mu_{\GPA}$ to predict the testing sample set. The training and prediction of the GPA estimator take approximately 25 seconds.
For comparison, we also use the one-shot estimator and the global estimator with distributed implementation to make prediction for the testing sample. 
Both methods show RMPE values close to the GPA estimator, but each takes more than 3.5 hours to complete the prediction.
Additionally, implementing $\hat\mu_{\GPA,\nu}$ with different $\nu\ (1\le \nu \le 4)$ on the testing sample confirms that the linear interpolation based GPA estimator $\hat\mu_{\GPA,1}$ (i.e., $\hat\mu_{\GPA}$) still yields the smallest RMPE value.
For illustration, we plot the prediction results of $\hat\mu_{\GPA}$ as a curve in Figure \ref{fig:airline}. 
As shown in Figure \ref{fig:airline}, the curve is not very smooth. This partially explains why previous procedures select the linear interpolation based GPA estimator. Since according to theorem \ref{thm:GPA_nu}, higher order interpolation requires stringent smoothness conditions on $\mu$ and $f$. 
By Figure \ref{fig:airline}, we can also observe that the early flights (before 7:00) show fewer arrival delays. In fact, flights departing during this time tend to arrive earlier than scheduled. This is partly because airports are relatively less crowded in the morning. However, delays start to increase after 7:00 and peak around 20:00.

\begin{figure}[htbp]
\centering
\includegraphics[width = 0.9\textwidth]{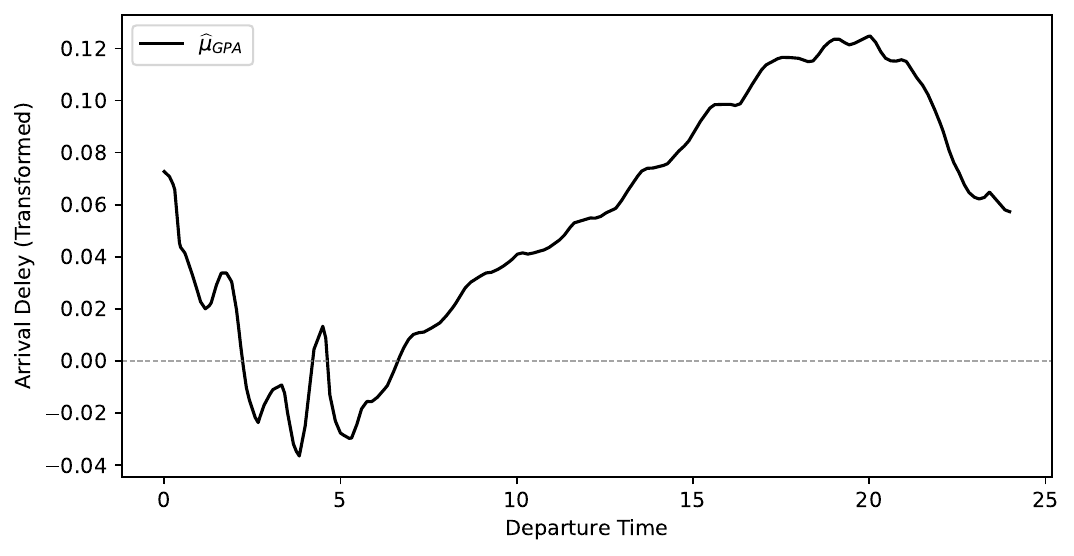}
\caption{The estimated curve of the transformed arrival delay vs. departure time.}
\label{fig:airline}
\end{figure}

\begin{figure}[htbp]
\centering
\includegraphics[width = 0.9\textwidth]{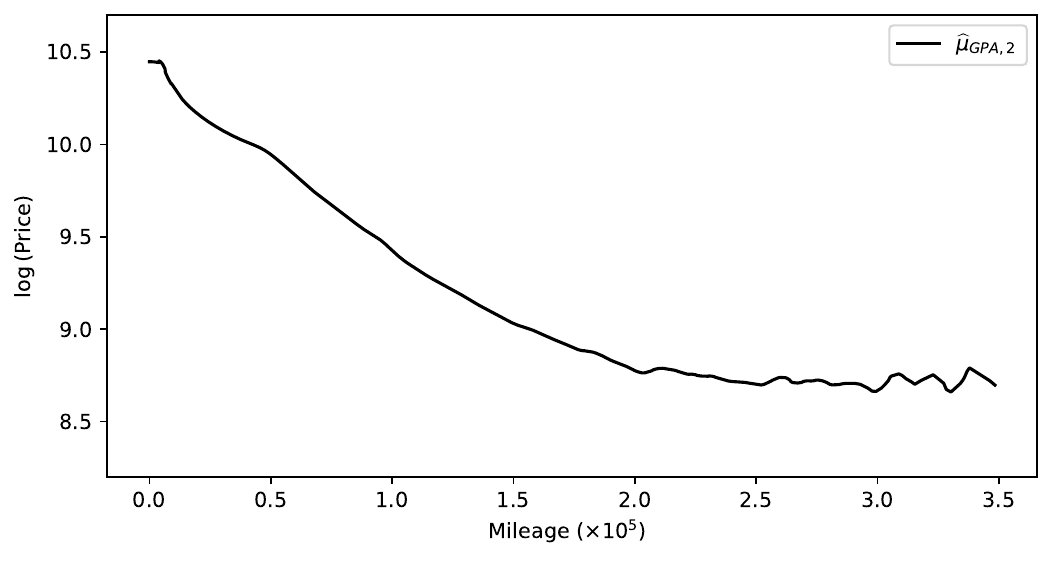}
\caption{The estimated curve of the log-transformed price vs. mileages (in miles).}
\label{fig:usedcar}
\end{figure}

\subsection{Used Cars Dataset}
{
We next apply our method to the US used cars dataset (\url{https://www.kaggle.com/datasets/ananaymital/us-used-cars-dataset}). 
The whole dataset contains details of 3 millions real-world used cars, taking over 9 GB of space. 
We choose the price of a car as the response of interest, and the mileage (i.e., the miles ($\times 10^5$) traveled of a car since manufactured) as the predictor. 
Our preceding analysis suggests that the distribution of the response variable is highly right-skewed. Therefore, we apply a logarithmic transformation to it. 
After data cleaning, a total of $2,610,715$ observations are retained. 
Among them, $N=2,490,000$ observations are randomly selected as the training sample set, $n_\textup{val}=10,000$ observations are randomly selected as the validation set $\mV$, and the remaining $N^* = 110,715$ observations are used as the testing sample set. 
A total of $M=50$ machines are used. 
We use the same procedures as the previous example to select the interpolation order and the bandwidth. 
In this example, we find that the quadratic interpolation based GPA estimator $\hat\mu_{\GPA,2}$ shows the smallest RMPE value. The corresponding bandwidth is $\hat h_\textup{PLT} =  0.069$, and a total of $J+1 = [3.5\hat h_\textup{PLT}^{-1} \log \log N  ]+1 = 137 $ grid points are equally spaced within the interval $[0, 3.5]$. 
We then use $\hat\mu_{\GPA,2}$ to predict the testing observations. The training and prediction of the GPA estimator $\hat\mu_{\GPA,2}$ take about 4 seconds. 
For comparison, we also use the one-shot estimator and the global estimator with the same bandwidth and kernel function to predict the testing observations. 
While both estimators yield RMPE values similar to $\hat\mu_{\GPA,2}$, their prediction times exceed 65 seconds.
In addition, on the prediction sample, the GPA estimator $\hat\mu_{\GPA,\nu}$ with $\nu=2$ still yields the smallest RMPE value compared to the other interpolation orders.
For illustration, we plot the prediction results of $\hat\mu_{\GPA,2}$ as a curve in Figure \ref{fig:usedcar}.
By Figure \ref{fig:usedcar}, we can observe that the curve is relatively smooth. This partially explains why quadratic interpolation beats the linear interpolation. 
In addition, we find that the price of used cars almost monotonically decreases with mileage increasing. It is noted that when the mileage is less than $2\times 10^5$ miles, the log-transformed price drops very quickly as the mileage increases. However, when the mileage exceeds $2\times 10^5$ miles, the impact of mileage on the log-transformed price becomes no longer substantial.
}

\section{Concluding Remarks}

In this paper, we study the problem of nonparametric smoothing on a distributed system. 
We first investigate the popularly used one-shot type estimator and the distributed implementation of the global estimator. 
However, we find that the two estimators are highly inefficient for prediction in terms of both computation and communication costs.
This drives us to develop a novel GPA estimator, which requires no further communication when making predictions.
Subsequently, the statistical properties of the GPA estimator are carefully investigated.
We find that the GPA estimator has the same statistical efficiency as the global estimator under very mild conditions.
In addition, two novel bandwidth selectors are developed for the distributed estimators, and their theoretical properties are then investigated.
To further extend the GPA method, we also explore the application of a higher-order polynomial interpolated GPA estimator and a multivariate GPA estimator.
Finally, these theoretical results are carefully validated through extensive numerical studies and real data analysis.

To conclude the article, we provide several directions for future studies. 
First, we investigate the nonparametric estimation problem on a centralized distributed system. It is worthwhile to explore whether the GPA method can be extended to a decentralized distributed system.
Second, our numerical studies show that the one-shot bandwidth selector exhibits excellent performance.
It is interesting to compare it with the global bandwidth selector in terms of the finite sample performance.
Last, many semiparametric models are popularly used in modern statistical inference, including the single index model, partially linear model, varying coefficient model, and many others.
It would be very interesting to apply our GPA method to these models.

\bibliographystyle{apalike}
\bibliography{ref.bib}

\clearpage

\iftrue{

\setcounter{table}{0}
\renewcommand{\thetable}{S.\arabic{table}}

\setcounter{equation}{0}
\renewcommand{\theequation}{S.\arabic{equation}}

\setcounter{theorem}{0}
\renewcommand{\thetheorem}{S.\arabic{theorem}}

\setcounter{lemma}{0}
\renewcommand{\thelemma}{S.\arabic{lemma}}

\section*{Appendix}

\begin{appendix}

\section{Proofs of the Theorems}

\subsection{Properties of the Global and the One-Shot Estimators}
\label{append:thm:global&OS}

{
For the sake of completeness, we present the properties of the global and the one-shot estimators in the following theorem.
 \begin{theorem} \label{thm:global&OS}
	Assume conditions \ref{cond:smoothness}--\ref{cond:bandwidth} hold and $f(x)>0$. Then, we have
	\begin{enumerate}[(a)]
		\item $\hat \mu(x) - \mu(x) = Q_1 + Q_2 + \mQ$, where $\sqrt{Nh}\Big\{Q_1 - B(x)h^2\Big\} \to_d \mN\Big(0, V(x)\Big)$, $\var(Q_2)=O\left\{1/(Nh)^2\right\}$, $\mQ = O_p\left\{ 1/ (Nh)^{3/2} \right\}$, and 
		\begin{align*}
			E(Q_2) = -\left[ \frac{\nu_2\big\{\ddot\mu(x) f(x) + 2\dot\mu(x)\dot f(x) \big\}}{2f^2(x)}\right]\left(\frac{h}{N}\right)+o\left(\frac{h}{N}\right) + O(h^4).
		\end{align*}
		
		\item Further assume that $nh/ \log M \to\infty$ as $N\to \infty$ and $|Y_i|<C_Y$ for some constant $C_Y>0$, then, we have $\hat\mu_\textup{OS}(x) - \mu(x) = Q_1 + \bar Q_2 + \bar \mQ$, where $\var(\bar Q_2)=O\left\{1/(Nnh^2)\right\}$, $\bar\mQ = O_p\left\{ 1/ (nh)^{3/2} \right\}$, and
		\begin{align*}
			&E(\bar Q_2) = -\left[ \frac{\nu_2\big\{\ddot\mu(x) f(x) + 2\dot\mu(x)\dot f(x) \big\} } {2f^2(x)}\right]\left(\frac{h}{n}\right)+o\left(\frac{h}{n}\right)+ O(h^4).
		\end{align*}
	\end{enumerate}
\end{theorem}

By Theorem \ref{thm:global&OS}(a), we know that $Q_2$ and $\mQ$ are of the same order $o_p(1/\sqrt{Nh})$.
Consequently, $\hat \mu(x) - \mu(x)$ and $Q_1$ should share the same asymptotic distribution as
\begin{align*}
	\sqrt{Nh}\Big\{ \hat \mu(x) - \mu(x) - B(x)h^2\Big\} \to_d \mN\Big(0, V(x)\Big).
\end{align*}
This is in line with the classical results on the NW estimator \citep[e.g.,][ Theorem 2.2]{li2007nonparametric}.
However, we try to make further contributions to the literature by providing a slightly more elaborated second-order asymptotic analysis of $Q_2$.
By Theorem \ref{thm:global&OS}(b), we require that {\color{red} $nh/\log M\rightarrow \infty$} as $N\rightarrow\infty$ to ensure the consistency of the one-shot estimator. Recall that $h = CN^{-1/5}$ with some constant $C>0$. This implies that the local sample size $n$ should be no smaller than $N^{1/5}\log M$. Furthermore, we know that $\bar Q_2 = O_p(h/n+ 1/\sqrt{Nnh^2}) = o_p(1 / \sqrt{Nh})$ as long as $nh\to \infty$.
Furthermore, if $\bar \mQ = o_p(1 / \sqrt{Nh})$, the one-shot estimator $\hat\mu_\textup{OS}(x)$ should be asymptotically as efficient as the global estimator $\hat \mu(x)$.
In fact, this condition can be satisfied provided $n\gg N^{7/15} $.
We note that this is a relatively mild condition. 
In contrast, for a distributed parametric problem, we typically require that $n\gg N^{1/2}>N^{7/15}$; for example, see \cite{zhang2013communication} and \cite{huang2019distributed}.

}

\begin{proof}

	We use $C_k$ with $ k\ge 0$ to denote positive constants, whose values might change from occurrence to occurrence. 
	In the following, we prove theorem conclusions (a) and (b) subsequently.
	
	\noindent\textbf{Proof of }(a). For some fixed point $x\in(0,1)$, in a slight abuse of notation, we define
	\begin{align*}
		\hat f = \hat f(x) = \frac{1}{Nh} \sum_{i=1}^N K\left(\frac{X_i - x}{h} \right), \quad \hat g = \hat g(x) =\frac{1}{Nh}  \sum_{i=1}^N K\left(\frac{X_i - x}{h} \right)Y_i.
	\end{align*}
	It is easy to prove that $\hat f$ and $\hat g$ are consistent estimators of $f$ and $g=\mu f$ respectively, under the conditions \ref{cond:smoothness}--\ref{cond:bandwidth}. 
	Then it follows from Taylor expansion that $\hat \mu - \mu =$
	\begin{align*}
		\frac{\hat g}{\hat f} - \frac{g}{f} 
		=& \left\{- \frac{\mu}{f}(\hat f - f) +  \frac{1}{f} (\hat g - g)\right\} + \left\{\frac{\mu}{f^2} (\hat f - f)^2 - \frac{1}{f^2} (\hat f - f)(\hat g - g)\right\} \\
		&+ \left\{-\frac{\tilde g }{\tilde f^4} (\hat f - f)^3 + \frac{1}{\tilde f ^3 }(\hat f - f)^2(\hat g - g) \right\}\\
		=& \Big\{Q_{1}\Big\} + \Big\{Q_{2}\Big\} + \Big\{\mQ\Big\},
	\end{align*}
	where $\tilde f $ is between $\hat f$ and $f$, and $\tilde g$ is between $\hat g$ and $g$. We then consider $Q_{1}$, $Q_{2}$ and $\mQ$ respectively.

	We first derive the limit distribution of the term $Q_1$. 
	In fact, by Cram\'er-Wold device and Lyapunov's central limit theorem, we can establish that 
	\begin{align} \label{eq:f&g}
		\sqrt{Nh}\left(\hat f - f - B_f h^2 ,\  \hat g - g - B_g h^2\right) \to_d \mN(0, \Sigma),
	\end{align}
	under the assumed conditions.
	Here, $\Sigma = [\sigma_{ff}, \sigma_{fg}; \sigma_{fg}, \sigma_{gg}] \in \mR^{2\times 2}$ with $\sigma_{ff} = \nu_0 f$, $\sigma_{gg} = (\mu^2 + \sigma^2) \nu_0 f$ and $\sigma_{fg} = \mu \nu_0 f$, and $B_f = \kappa_2 \ddot f / 2$ and $B_g = \kappa_2 (\ddot \mu f + \mu \ddot f + 2 \dot\mu \dot f) / 2$.
	Then, by Delta method and \eqref{eq:f&g}, we have 
	\begin{align}\label{eq:Q1_normal}
		\sqrt{Nh}\Big\{Q_1(x)-B(x)h^2 \Big\} \to_d \mN\Big(0, V(x)\Big),
	\end{align}
	where $B = (-\mu / f)B_f + (1/f)B_g  = \kappa_2 (\ddot \mu + 2 \dot \mu \dot f  / f) / 2$, and $V = (\mu^2 / f^2) \sigma_{ff} + (1/f^2)\sigma_{fg} - 2(\mu /f^2) \sigma_{fg} =\nu_0 \sigma^2  / f $ are the same as those defined in \eqref{eq:B&V}.
	
	We next compute the expectation and variance of $Q_2$.
	Similar to the proofs of \eqref{eq:Efm}, \eqref{eq:varfm} and \eqref{eq:Efmgm} in Lemma \ref{lemma:fmgm}, we can obtain the expression of $E(\hat f - f)^2$ and $E(\hat f -f)(\hat g - g)$. 
	Then we should have $EQ_2 =$
	\begin{align*}
		\frac{\mu}{f^2} E(\hat f - f)^2 -\frac{1}{f^2} E(\hat f -f)(\hat g - g) = -\left\{ \frac{\nu_2(2\dot\mu\dot f + \ddot\mu f) } {2f^2}\right\}\frac{h}{N} + O(h^4) +o\left(\frac{h}{N}\right).
	\end{align*}
	Similarly, from the proof of \eqref{eq:varfm2}, \eqref{eq:varfmgm} and \eqref{eq:covfm2fmgm} in Lemma \ref{lemma:fmgm}, we can obtain the expression of
	$\var\left\{(\hat f- f)^2 \right\}$, $\var\Big\{ ( \hat f - f)(\hat g - g)  \Big\}$, and $\cov\left\{(\hat f - f)^2, ( \hat f - f)(\hat g - g)\right\}$.
	Thus, we have 
	\begin{align*}
		\var( Q_2)=&\frac{\mu^2}{f^4}\var\left\{(\hat f- f)^2 \right\} +   \frac{1}{f^4} \var\Big\{ ( f_m - f)(\hat g - g)  \Big\} \\
		&- \frac{2\mu}{f^4} \cov\left\{(\hat f - f)^2, ( \hat f - f)(\hat g - g)\right\}
		=O\left(\frac{1}{N^2h^2}\right)
	\end{align*}
	
	Finally, we consider the remainder term $\mQ$. 
	Recall that $\tilde f = f + o_p(1)$, $\tilde g = g + o_p(1)$, $\hat f - f = O_p(1/\sqrt{Nh})$, and $\hat g - g = O_p(1/\sqrt{Nh})$. Then we should have $\mQ  = -(\tilde g/\tilde f^4)(\hat f - f)^3 + (1/\tilde f^3)(\hat f - f)^2(\hat g - g)= O_p\{ 1  / (Nh)^{3/2}\}$. 
	This completes the proof of conclusion (a).

	\noindent\textbf{Proof of }(b). 
	Analogously, we define the local statistics on the $m$-th machine as 
	\begin{align*}
		\hat f_m = \hat f_m(x) = \frac{1}{n} \sum_{i\in \mS_m} K_h(X_i - x), \quad \hat g_m = \hat g_m(x) =  \frac{1}{n} \sum_{i\in \mS_m} K_h(X_i - x)Y_i.
	\end{align*}
	Then the one-shot estimator can be written as $\hat \mu_\textup{OS} = M^{-1} \sum_{m-1}^M \hat \mu_m$, where $\hat\mu_m = \hat g_m / \hat f_m$ is the local NW estimator.
	It follows from Taylor expansion that $\hat \mu_m - \mu=$
	\begin{align*}
		\frac{\hat g_m}{\hat f_m} - \frac{g}{f} 
		=& \left\{- \frac{\mu}{f}(\hat f_m - f) +  \frac{1}{f} (\hat g_m - g)\right\} + \left\{\frac{\mu}{f^2} (\hat f_m - f)^2 - \frac{1}{f^2} (\hat f_m - f)(\hat g_m - g)\right\} \\
		&+\left\{- \frac{\tilde g_m }{\tilde f_m^4} (\hat f_m - f)^3 + \frac{1}{\tilde f_m ^3 }(\hat f_m - f)^2(\hat g_m - g) \right\}\\
		=& \Big\{Q_{m1}\Big\} + \Big\{Q_{m2}\Big\} + \Big\{\mQ_{m}\Big\},
	\end{align*}
	where $\tilde f_m $ is between $\hat f_m$ and $f$, and $\tilde g_m$ is between $\hat g_m$ and $g$.
	Consequently, $\hat \mu_\textup{OS} - \mu = \bar Q_1 + \bar Q_2 + \bar \mQ$, where $\bar Q_{1} = M^{-1} \sum_{m=1}^M Q_{m1},\ \bar Q_{2} = M^{-1} \sum_{m=1}^M Q_{m2}  $, and $ \bar \mQ = M^{-1} \sum_{m=1}^M \mQ_m  $. Note that $\hat f = M^{-1}\sum_{m-1}^M \hat f_m$ and $\hat g= M^{-1}\sum_{m-1}^M \hat g_m$. Thus, we should have $\bar Q_1 = Q_1$. Hence, we only need to investigate $\bar Q_2$ and $\bar\mQ$.
	
	We first compute the expectation and variance of $\bar Q_2$. 
	Recall that $Q_m\ (1\le m \le M)$ are independent and identically distributed. 
	Hence $E (\bar Q_2) = E(Q_{m2})$ and $\var(\bar Q_2) = M^{-1} \var (Q_{m2})$.
	Then by equations \eqref{eq:Efm}, \eqref{eq:varfm} and \eqref{eq:Efmgm} in Lemma \ref{lemma:fmgm}, we should have $E(\bar Q_2) = E( Q_{m2}) =$
	\begin{align*}
		\frac{\mu}{f^2} E(\hat f_m - f)^2 -\frac{1}{f^2} E(\hat f_m -f)(\hat g_m - g) 
		= -\left\{ \frac{\nu_2(2\dot\mu\dot f + \ddot\mu f) } {2f^2}\right\}\frac{h}{n} + O(h^4) +o\left(\frac{h}{n}\right).
	\end{align*}
	Similarly, by equations \eqref{eq:varfm2}. \eqref{eq:varfmgm} and \eqref{eq:covfm2fmgm} in Lemma \ref{lemma:fmgm}, we have $\var(\bar Q_2) =$
	\begin{align*}
		M^{-1} \var (Q_{m2})
		=&M^{-1}\Bigg[\frac{\mu^2}{f^4}\var\left\{(\hat f_m - f)^2 \right\} +   \frac{1}{f^4} \var\left\{ ( f_m - f)(\hat g_m - g)  \right\}  \\
		&~~~~~~~~ - \frac{2\mu}{f^4} \cov\left\{(\hat f_m - f)^2, ( \hat f_m - f)(\hat g_m - g)\right\}\Bigg]\\
		=&M^{-1}O\left(\frac{1}{n^2h^2}\right)=O\left(\frac{1}{Nnh^2}\right).
	\end{align*}
	
	Finally, we need to prove $\bar \mQ = O_p\left\{ 1 / (nh)^{3/2}\right\}$. In fact, it is sufficient to verify that for an arbitrarily small $\epsilon>0$, there exists a sufficiently large $t>0$ such that $\liminf_N P\big\{ (nh)^{3/2}|\bar \mQ| > t \big\} \le \epsilon $.
	To this end, we define the event $$\Omega = \Big\{ \max_{1\le m\le M} \Big(|\tilde g_m / \tilde f_m^4|, |1 / \tilde f_m^3| \Big)> C_0\Big\},$$ where $C_0>0$ is some sufficiently large constant. Then by the probability inequality \eqref{ineq:bound} in Lemma \ref{lemma:concen}, we know that $P(\Omega) \le C_1 \exp\left( \log M - C_2 nh \right)$, where $C_1,\ C_2>0$ are some constants.
	Then we have 
	\begin{align}
		&P\Big\{ (nh)^{3/2}|\bar \mQ| > t \Big\} 
		\le   P\Big\{ (nh)^{3/2}|\bar \mQ| > t, \Omega^c \Big\} + P(\Omega)\nonumber\\
		\le & P\Big\{ C_0(nh)^{3/2} M^{-1}\sum_{m} \left(|\hat f_m - f|^3 + |(\hat f_m - f)^2 (\hat g_m - g)|\right)>t  \Big\} + P(\Omega) \nonumber\\
		\le& { C_0 (nh)^{3/2}E\left(|\hat f_m - f|^3 + |(\hat f_m - f)^2 (\hat g_m - g)| \right)}\Big/ t + P(\Omega) \label{ineq:barmQ},
	\end{align}
	where the last line follows from Markov's inequality.
	By similar calculations in the proof of Lemma \ref{lemma:fmgm} and H\"older's inequality, we can obtain that 
	\begin{align*}
		&E|\hat f_m - f|^3 \le \sqrt{E(\hat f_m - f)^6 } = O\left\{\frac{1}{(nh)^{3/2}} + h^6 \right\},\\
		&E|(\hat f_m - f)^2 (\hat g_m - g)|\le \sqrt{E(\hat f_m - f)^4 } \sqrt{E(\hat g_m - g)^2} = O\left\{\frac{1}{(nh)^{3/2}} + h^6 \right\}.
	\end{align*}
	Note that $(nh)^{3/2} h^6 = (nh^5)^{3/2} \le  (Nh^5)^{3/2} = O(1)$. Together with \eqref{ineq:barmQ}, we can obtain that $\liminf_N P\big\{ (nh)^{3/2}|\bar \mQ| > t \big\} \le \epsilon$ by choosing a sufficiently large $t$. 
 This completes the proof of conclusion (b).
\end{proof}

	\subsection{Proof of Theorem \ref{thm:GPA_asymptotic}}
	\label{append:thm:GPA_asymptotic}
	
	Recall that $\Delta = 1/J$ is the distance between any two adjacent grid points. Since $f(x)>0$, $f(\cdot)$ continuous, and $\max_{k=j,j+1} |x_k^*-x|\le \Delta\to 0$, we should have $f(x_k^*)>0$ for $k=j,j+1$.
	Consequently, the conclusions in Theorem \ref{thm:global&OS} (a) hold for $\hat \mu(x)$ and $\hat\mu(x_k^*),\ k=j,j+1$.
	Then by interpolation formula \eqref{eq:interpolation}, we should have 
	$\hat\mu_\textup{GPA}(x) - \mu(x) = \tilde{Q}_0  + \tilde{Q}_1 + \tilde Q_2 + \tilde \mQ $, where 
	$\tilde Q_0 = \tilde{Q}_0(x) = \{ (x_{j+1}^* - x) / \Delta \} \mu(x_j^* ) + \{ (x - x_j^* ) / \Delta \} \mu(x_{j+1}^*) - \mu(x)$,
	$\tilde Q_1 = \tilde{Q}_1(x) = \{ (x_{j+1}^* - x) / \Delta \} Q_1(x_j^* ) + \{ (x - x_j^* ) / \Delta \} Q_1(x_{j+1}^*)$, 
	$\tilde Q_2 = \tilde{Q}_2(x) =\{ (x_{j+1}^* - x) /  \Delta \} Q_2(x_j^* ) + \{ (x - x_j^*  ) / \Delta \} Q_2(x_{j+1}^*)$, and 
	$\tilde \mQ   = \tilde \mQ(x)   = \{ (x_{j+1}^* - x) /  \Delta \}  \mQ(x_j^* ) + \{ (x - x_j^*  ) / \Delta \} \mQ(x_{j+1}^*) $. 
	Then by conclusions in Theorem \ref{thm:global&OS}, we can obtain the results about $\tilde Q_2$ and $\tilde \mQ$ in Theorem \ref{thm:GPA_asymptotic} immediately.
	In the following, we focus on the analysis of $\tilde Q_0$ and $\tilde Q_1$.
	
	By condition \ref{cond:smoothness} and error formula of linear interpolation \citep[Theorem 6.2]{suli2003introduction}, we have
	\begin{align*}
		\tilde Q_0 =  \left\{\frac{x_{j+1}^*-x}{\Delta}\mu(x_j^* ) + \frac{x-x_j^* } {\Delta}\mu(x_{j+1}^*)\right\} - \mu(x)= -\frac{(x-x_j^* )(x-x_{j+1}^*)}{2}\mu''( \tilde x),
	\end{align*}
	where $\tilde x \in [x_j^* , x_{j+1}^*]$. Since $|(x-x_j^* )(x-x_{j+1}^*)|\le \Delta^2 / 4$, we conclude that $\tilde Q_0 = O(\Delta^2) = O(1/J^2)$.
	
	We next establish the asymptotic normality of $\tilde Q_1$. 
	Recall that
	\begin{align}\label{eq:Q1xl}
		Q_1(x) = - \frac{\mu(x)}{f(x)}\Big\{\hat f(x) - f(x)\Big\} +  \frac{1}{f(x)} \Big\{\hat g(x) - g(x)\Big\}.
	\end{align}
	Similar to the proof of \eqref{eq:Efm} and \eqref{eq:Egm} in Lemma \ref{lemma:fmgm}, we can calculate that $EQ_1(x) = B(x)h^2 + o(h^2)$ and $EQ_1(x_k^*) = B(x_k^*)h^2 + o(h^2)$ for $k =j,j+1$, where $B(\cdot)$ is defined in \eqref{eq:B&V}.
	Under the assumed conditions, we know that $B(\cdot)$ is a continuous function. 
	Thus, we should have $EQ_1(x_k^*) = B(x_k^*)h^2 + o(h^2) = B(x)\{1+o(1)\} + o(h^2) = EQ_1(x)  + o(h^2)$, since $\max_{k = j,j+1}|x_k ^*-x| \le \Delta \to 0$.
	By \eqref{eq:varQ1xl} in Lemma \ref{lemma:Q1jQ1j+1}, we have $\var\{Q_1(x)\} = \left\{\nu_0\sigma^2 / f(x)\right\}/Nh + o\{1/(Nh)\}$ and  $\var\{Q_1(x_k^*)\} = \var\{Q_1(x)\} \{1+o(1)\}$ for $k=j,j+1$.
	This implies that $\{EQ_1(x_k^*)  - E Q_1(x) \} / \sqrt{\var(Q_1)} =o(h^2) O(\sqrt{Nh}) = o(\sqrt{Nh^5} )= o(1)$ for $k=j,j+1$ under condition \ref{cond:bandwidth}. 
	Furthermore, by \eqref{eq:covQ1jQ1j+1} in Lemma \ref{lemma:Q1jQ1j+1}, we have $\cov\{Q_1(x_k^*), Q_1(x) \} = \var\{Q_1(x)\} \{1+o(1)\}$ for $k=j,j+1$. 
	Recall that $ \tilde{Q}_1(x) = \{ (x_{j+1}^* - x) / \Delta \} Q_1(x_j^* ) + \{ (x - x_j^* ) / \Delta \} Q_1(x_{j+1}^*)$ and $\{ (x_{j+1}^* - x) / \Delta \} + \{ x - x_j^*) / \Delta \}=1$.
	Then by Lemma \ref{lemma:difference}, we conclude that $\tilde Q_1(x) - Q_1(x) = o_p(\sqrt{\var\{Q_1(x)\}}) = o_p(1/\sqrt{Nh})$. This implies that $ \sqrt{Nh} \{\tilde Q_1(x) - Q_1(x)\} = o_p(1)$. 
	By \eqref{eq:Q1_normal}, we know that $ \sqrt{Nh}\{Q_1(x)-B(x)h^2\} \to_d \mN(0, V(x))$ as $N\to\infty$. Then by Slutsky's theorem, we conclude that $ \sqrt{Nh}\{\tilde Q_1(x)-B(x)h^2\} \to_d \mN(0, V(x))$ as well.
	
	Note that $\sqrt{Nh} \tilde Q_0(x) = \sqrt{Nh} O(1/J^2)= o(\sqrt{Nh^5}) = o(1)$, since $\Delta = 1 / J=o(h)$ under the assumed conditions. Together with $\sqrt{Nh} \tilde Q_2(x) = o_p(1)$ and $\sqrt{Nh} \tilde\mQ(x) = o_p(1)$, it follow from Slutsky's theorem that
	\begin{align*}
		\sqrt{Nh}\Big\{\hat\mu_\textup{GPA}(x) - \mu(x) -B(x)h^2 \Big\} \to_d \mN\Big(0, V(x)\Big).
	\end{align*}
	Therefore, we accomplish the proof of the theorem.

	\subsection{Proof of Theorem \ref{thm:bandwidth}} 
	\label{append:thm:bandwidth}
	
	We first state some additional conditions.
	\begin{enumerate}[(C\arabic*)]
		\setcounter{enumi}{3}
		\item (\textsc{Candidate Bandwidth}) Let $H_n = \big[C_H ^{-1}n^{-1/5}, C_H n^{-1/5}\big]$ be the candidate bandwidth set of the local CV, and $H_{n_0} = \big[ C_H ^{-1}n_0^{-1/5}, C_H n_0^{-1/5}\big]$ be that of the pilot sample based CV, where $C_H>C_\textup{opt}$ is a positive constant. \label{cond:cand_bandwidth}
		
		\item (\textsc{Weight Function}) The weight function $w(\cdot)$ is bounded and supported on a compact set with a nonempty interior. In addition, the density function $f(\cdot)>\tau$ for some constant $\tau>0$ on the support of $w(\cdot)$. \label{cond:weight_fun}
	\end{enumerate}
	\noindent
	Condition \ref{cond:cand_bandwidth} assures that the candidate bandwidth set includes the optimal bandwidth and has the same order as it. This can be easily satisfied in practice, since we usually select the bandwidth of the optimal order. A similar condition is also assumed in \cite{hardle1985optimal} and \cite{racine2004nonparametric}.
	Condition \ref{cond:weight_fun} is used to establish uniform convergence of $\hat f_m$ and $\hat \mu_m$. The same condition is assumed in \cite{hardle1985optimal} and \cite{racine2004nonparametric}.
	
	The proofs of the two conclusions are very similar, and the first one is much more involved.
	Hence, we only provide the proof of the first conclusion, i.e., $\hat h_\textup{OS}/ h_\textup{opt} \to_p 1$ as $n\to \infty$.
	
	We start by proving that $\hat h_m / h_{\textup{opt},m} \to_p 1 $ as $n\to \infty$, where $\hat h_m = \argmin_{h\in H_n} \CV_{m}(h)$ is the optimal bandwidth selected by local CV, and $h_{\textup{opt},m} = C_\textup{opt}n^{-1/5}$ is the asymptotically local optimal bandwidth.
	Firstly, we can compute that $\CV_{m}(h)=$
	\begin{align}
		&\frac{1}{n}\sum_{i \in \mS_m} \Big\{Y_i - \hat\mu_m^{(-i)}(X_i) \Big\}^2 w(X_i)
		= \frac{1}{n}\sum_{i \in \mS_m} \Big\{\mu(X_i) + \varepsilon_i - \hat\mu_m^{(-i)}(X_i) \Big\}^2 w(X_i) \nonumber\\
		=&\frac{1}{n}\sum_{i \in \mS_m} \Big\{\mu(X_i) - \hat\mu_m^{(-i)}(X_i) \Big\}^2 w(X_i)+ \frac{2}{n}\sum_{i \in \mS_m} \Big\{\mu(X_i) - \hat\mu_m^{(-i)}(X_i) \Big\}\varepsilon_i w(X_i) \nonumber \\
		&+\frac{1}{n}\sum_{i \in \mS_m} \varepsilon_i^2 w(X_i). \label{eq:CVlocal}
	\end{align}
	Here, $\hat\mu_m^{(-i)}(X_i) = \hat g_m^{(-i)}(X_i) \Big/ \hat\mu_m^{(-i)}(X_i)$ is the leave-one-out estimator, where
	\begin{align}
		\hat g_m^{(-i)}(X_i) =  \frac{1}{n}\sum_{j\in \mS_m,j\ne i} K_h(X_j - X_i) Y_j, \quad
		\hat f_m^{(-i)}(X_i) = \frac{1}{n}\sum_{j\in\mS_m, j\ne i}K_h(X_j-X_i).\label{eq:g&floo}
	\end{align}
	Since the last term in \eqref{eq:CVlocal} is independent of $h$, we ignore it in the following derivation.
	Write $\mu(X_i) - \hat\mu_m^{(-i)}(X_i) =$
	\begin{align}
		\frac{\Big\{\mu(X_i) - \hat\mu_m^{(-i)}(X_i)\Big\} \hat f_m^{(-i)}(X_i)} {f(X_i)} + \frac{\Big\{\mu(X_i) - \hat\mu_m^{(-i)}(X_i)  \Big\}\Big\{f(X_i) - \hat f_m^{(-i)}(X_i) \Big\}}{f(X_i)}. \label{eq:mu-mu^(-i)}
	\end{align}
	By similar arguments in the proof of Lemma 1 in \cite{hardle1985optimal}, we can derive the uniform convergence of $\hat f_m$ to $f$ and $\hat \mu_m$ to $\mu$. Thus, the second term should be of smaller order than the first term in \eqref{eq:mu-mu^(-i)}.
	Replacing $\mu(X_i) - \hat\mu_m^{(-i)}(X_i)$ by the first term of \eqref{eq:mu-mu^(-i)} and substituting it into \eqref{eq:CVlocal}, we can obtain the leading term of $\CV_{m}(h)$, which is denoted by 
	\begin{align}
		\begin{aligned}
			\CV_{m,1}(h) =&\frac{1}{n}\sum_{i \in \mS_m} \frac{\Big\{\mu(X_i) - \hat\mu_m^{(-i)}(X_i) \Big\}^2 \big\{\hat f_m^{(-i)}(X_i)\big\}^2   w(X_i)}{f^2(X_i)}\\
			&+ \frac{2}{n}\sum_{i \in \mS_m} \frac{\Big\{\mu(X_i) - \hat\mu_m^{(-i)}(X_i) \Big\} \hat f_m^{(-i)}(X_i) \varepsilon_i w(X_i)}{f(X_i)}.\label{eq:CVleading}
		\end{aligned} 
	\end{align}
	For the simplicity of the notations, we denote $K_{h,ij} = K_h(X_j - X_i)$, $w_i = w(X_i)$, $\mu_i = \mu(X_i)$, $f_i = f(X_i)$, $\hat \mu_m^{(-i)} = \hat\mu_m^{(-i)}(X_i)$, and $\hat f_m^{(-i)} = \hat f_m^{(-i)}(X_i)$.
	Substituting \eqref{eq:g&floo} into \eqref{eq:CVleading}, we can obtain that $\CV_{m,1}(h)=$
	\begin{align*}
		&n^{-3}\sum_{i\in \mS_m} \sum_{j \ne i}\sum_{k \ne i} w_i (\mu_i - Y_j)(\mu_i - Y_k )K_{h, ij} K_{h, ik } /f_i^2 + 2n^{-1}\sum_{i\in \mS_m} w_i \varepsilon_i (\mu_i - \hat \mu_m^{(-i)} )\hat f_m^{(-i)} / f_i\\
		=&\Bigg\{n^{-3}\sum_{i\in \mS_m} \sum_{j \ne i}\sum_{k \ne i} w_i (\mu_i - \mu_j)(\mu_i - \mu_k )K_{h, ij} K_{h, ik } /f_i^2 \Bigg\} \\
		&+\Bigg\{n^{-3}\sum_{i\in \mS_m} \sum_{j \ne i}\sum_{k \ne i} w_i \varepsilon_j  \varepsilon_k  K_{h, ij} K_{h, ik } /f_i^2  - 2n^{-2}\sum_{i\in \mS_m}\sum_{j\ne i}w_i \varepsilon_i \varepsilon_j K_{h, ij} /f_i \Bigg\}\\
		&+ 2 \Bigg\{n^{-2}\sum_{i\in \mS_m} \sum_{j \ne i} w_i \varepsilon_i (\mu_i - \mu_j)K_{h, ij}  /f_i - n^{-3}\sum_{i\in \mS_m}\sum_{j \ne i}\sum_{k \ne i} w_i\varepsilon_k   (\mu_i - \mu_j) K_{h, ij} K_{h, ik } /f_i^2  \Bigg\}\\
		=&\Big\{T_1 \Big\} + \Big\{T_2 \Big\} + 2\Big\{T_3 \Big\}.
	\end{align*}
	Then, by similar arguments in the proof of Lemma 1-4 in \cite{racine2004nonparametric}, we can show that $T_1 = \bar B h^4 + o_p\left(h^4 \right), T_2 = \bar V / (nh) + o_p\{1 / (nh) \},  T_3 = O_p({h^2}/{\sqrt{n}})$ and $\CV_{m}(h) - \CV_{m,1}(h) = o_p\{ h^4 + 1/ (nh) \}$. 
	Consequently, we conclude that $\CV_{m}(h) =  \left\{\bar B h^4 + \bar V / (nh)\right\} \left\{1+o_p(1)\right\}$.
	Thus we should have $\hat h_m = h_{\textup{opt},m}\{1+o_p(1)\}$, where $h_{\textup{opt},m}  = C_\textup{opt}n^{-1/5}$ and $C_\textup{opt} = \left\{\bar V / (4 \bar B) \right\}^{1/5}$.
	
	We next show that $\bar h = M^{-1} \sum_{m=1}^M \hat h_m$ is also a ratio consistent estimator for $h_{\textup{opt},m}$, i.e., $\bar h / h_{\textup{opt},m}\to_p 1$.
	By condition \ref{cond:cand_bandwidth} we know that $\hat h_m \in H_n$ and ${h_{\textup{opt},m}}\in H_n$, and thus $|(\hat h_m -h_{\textup{opt},m})/h_{\textup{opt},m}|\le (C_H + C_\textup{opt}) / C_\textup{opt}$.
	Then by dominated convergence theorem and the fact that $(\hat h_m -h_{\textup{opt},m})/h_{\textup{opt},m} \to_p 0$, we have
	\begin{align*}
		E\left|\frac{\bar h - h_{\textup{opt},m}}{h_{\textup{opt},m}} \right|\le \frac{1}{M}\sum_{m=1}^M E\left|\frac{ \hat h_m - h_{\textup{opt},m}}{h_{\textup{opt},m}} \right|= E\left|\frac{ \hat h_m - h_{\textup{opt},m}}{h_{\textup{opt},m}} \right| \to 0.
	\end{align*}
	This implies that $(\bar h -h_{\textup{opt},m})/h_{\textup{opt},m} \to_p 0$, i.e., $\bar h / h_{\textup{opt},m}  \to_p 1$ as $n\to \infty$.
	Consequently, $\hat h_\textup{OS} / h_\textup{opt} = (M^{-1/5} \bar h ) / (M^{-1/5} h_{\textup{opt},m}) \to_p 1$ as $n\to \infty$.
	This completes the proof.

	\subsection{Proof of Theorem \ref{thm:GPA_nu}}
	\label{append:thm:GPA_nu}
	
	In the following two steps, we first investigate the global NW estimator $\hat \mu(x)$ with the $(\nu+1)$-th order kernel. Subsequently, we derive the results for the $\nu$-th order polynomial interpolation based GPA estimator $\hat \mu_{\GPA, \nu}(x)$.
	
	\noindent\textbf{Step 1.}
	Recall that 
	\begin{align*}
		&\hat f= \hat f(x) = \frac{1}{Nh} \sum_{i=1}^N K\left(\frac{X_i-x}{h} \right),\ \hat g= \hat g(x) = \frac{1}{Nh} \sum_{i=1}^N K\left(\frac{X_i-x}{h} \right)Y_i.
	\end{align*}
	We can verify that $\hat f$ and $\hat g$ are consistent estimators of $f$ and $g = \mu f$, respectively, under the assumed conditions. Then by Taylor expansion we have $\hat \mu - \mu=$
	\begin{align*}
		\frac{\hat g}{\hat f} - \frac{g}{f} 
		=& \left\{- \frac{\mu}{f}(\hat f - f) +  \frac{1}{f} (\hat g - g)\right\} + \left\{\frac{\tilde g}{\tilde f^3} (\hat f - f)^2 - \frac{1}{\tilde f^2} (\hat f - f)(\hat g - g)\right\} \\
		=& \Big\{Q_{\nu,1}\Big\} + \Big\{\mQ_{\nu} \Big\},
	\end{align*}
	where $\tilde f $ is between $\hat f$ and $f$, and $\tilde g$ is between $\hat g$ and $g$. We next investigate $Q_{\nu,1}$ and $\mQ_{\nu}$ in the following two steps.
	
	\noindent\textbf{Step 1.1.}
	As the proof in Appendix \ref{append:thm:global&OS}, we can derive that
	\begin{align} \label{eq:f&g_nu}
		\sqrt{Nh}\left(\hat f - f - B_f h^{\nu+1} ,\  \hat g - g - B_g h^{\nu+1}\right) \to_d \mN(0, \Sigma),
	\end{align}
	by results in Lemma \ref{lemma:fg_nu}.
	Here, $\Sigma = [\sigma_{ff}, \sigma_{fg}; \sigma_{fg}, \sigma_{gg}] \in \mR^{2\times 2}$ with $\sigma_{ff} = \nu_0 f$, $\sigma_{gg} = (\mu^2 + \sigma^2) \nu_0 f$ and $\sigma_{fg} =  \nu_0 \mu f$, where $\nu_0 = \int K^2(u) du$. In addition, $B_{\nu,f} = \kappa_{\nu+1}  f^{(\nu+1)} / (\nu+1)!$ and $B_{\nu,g} = \kappa_{\nu+1} g^{(\nu+1)} / (\nu+1)!$ with $g^{(\nu+1)}=\sum_{s=0}^{\nu+1} \binom{\nu+1}{s} \mu^{(s)} f^{(\nu+1-s)}$.
	Then, by Delta method and \eqref{eq:f&g_nu}, we have
	\begin{align}\label{eq:Q_nu_1_normal}
		\sqrt{Nh}\Big\{Q_{\nu,1}(x)-B_{\nu}(x)h^{\nu+1}\Big\} \to_d \mN\Big(0, V(x)\Big),
	\end{align}
	where $V = (\mu^2 / f^2) \sigma_{ff} + (1/f^2)\sigma_{fg} - 2(\mu /f^2) \sigma_{fg} =\nu_0 \sigma^2  / f $ is the same as that in \eqref{eq:B&V}, and
	\begin{align}\label{eq:B_nu}
		B_{\nu} = -\frac{\mu}{f} B_{\nu,f} + \frac{1}{f} B_{\nu,g}  = \frac{\kappa_{\nu+1}}{(\nu+1)!} \frac{\sum_{s=1}^{\nu+1}\binom{\nu+1}{s} \mu^{(s)} f^{(\nu+1-s)}}{f},
	\end{align}
	where $\kappa_{\nu+1} = \int u^{\nu+1} K(u) du$.

	\noindent\textbf{Step 1.2.}
	By \eqref{eq:Ef_nu}--\eqref{eq:varg_nu} in Lemma \ref{lemma:fg_nu}, we know that $\tilde f = f + o_p(1)$, $\tilde g = g + o_p(1)$, $\hat f - f = O_p(1/\sqrt{Nh})$, and $\hat g - g = O_p(1/\sqrt{Nh})$ . Then we should have $\mQ_{\nu}  = {\tilde g} / {\tilde f^3} (\hat f - f)^2 - ({1}/{\tilde f^2}) (\hat f - f)(\hat g - g) = O\{1 / (Nh)\} = o_p(1/\sqrt{Nh})$. 
	
	From the above results, we can conclude that, for any fixed $x\in(0,1)$, we have
	\begin{align*}
		\sqrt{Nh}\Big\{\hat \mu(x) -\mu(x)- B_{\nu}(x)h^{\nu+1}\Big\}\to_d \mN\Big(0, V(x)\Big),
	\end{align*}
	where $B_{\nu}(x)$ is defined in \eqref{eq:B_nu}, and $V(x)$ is defined in \eqref{eq:B&V}.
	This ends the proof of Step 1.

	\noindent\textbf{Step 2.}
	From the proof of Step 1, we can obtain that
	$\hat \mu_{\textup{PGPA},\nu}(x) - \mu(x) = \tilde Q_{\nu,0} + \tilde Q_{\nu,1}  + \tilde \mQ_\nu$, where $\tilde Q_{\nu,0} = \tilde Q_{\nu,0}(x) = \sum_{k = j}^{j+\nu} q_k  (x) \mu(x_k ^*) - \mu(x)$, $\tilde Q_{\nu,1} = \tilde Q_{\nu,1}(x) = \sum_{k = j}^{j+\nu} q_k  (x) Q_{\nu,1}(x_k ^*)$, and $\tilde \mQ_{\nu}= \tilde \mQ_{\nu}(x) = \sum_{k = j}^{j+\nu} q_k  (x) \mQ_{\nu}(x_k ^*)$. 
	Note that $q_{k}(x)=\prod_{i=j, i\ne k}^{j+\nu}(x-x_i^*)/(x_k^*-x_i^*) $ for $j\le k\le j+\nu$ are bounded for any $x\in[0,1]$.
	In Step 1.2, we have shown that $\mQ_{\nu} = o_p(1/\sqrt{Nh})$.
	Then we can conclude that $\tilde Q_{\nu,2}=\sum_{k = j}^{j+\nu} q_k  (x) \mQ_{\nu}(x_k ^*)=o_p(1/\sqrt{Nh})$ as well.
	We next derive the results about $\tilde Q_{\nu,0}$ and $\tilde Q_{\nu,1}$ in the following two steps.
	
	\noindent\textbf{Step 2.1.}
	Since $\mu(\cdot)$ is assumed to be $(\nu+1)$-times continuously differentiable, then by error formula of polynomial interpolation \citep[Theorem 6.2]{suli2003introduction}, we have
	\begin{align*}
		\tilde Q_{\nu,0} = \sum_{k = j}^{j+\nu} q_k  (x) \mu(x_k ^*) - \mu(x)=-\frac{1}{(\nu+1)!} \prod_{k=j}^{j+\nu} (x-x_k^*) \mu^{(\nu+1)}(\tilde x),
	\end{align*}
	where $\tilde x \in [x_{j}^*, x_{j+\nu}^*]$. Since $|\prod_{k=j}^{j+\nu} (x-x_k^*)|\le (\nu/J)^{\nu+1}$, we conclude that $\tilde Q_{\nu,0} = O(1/J^{\nu+1})$.
	
	\noindent\textbf{Step 2.2.}
	We next establish the asymptotic normality of $\tilde Q_{\nu,1}$. 
	First, we consider the polynomial function $\mL (x) = \sum_{k = j-1}^{j+1} q_k  (x) -1$. 
	Note that $\mL (x)$ is a polynomial whose degree is at most $\nu$. 
	It is easy to verify that $\mL (x_k ^*) = 0$ for $j \le k \le j+\nu$, that is $\mL (x) $ has $\nu+1$ zeros.
	Then we must have $\mL(x) \equiv 0$, or equivalently, $\sum_{k = j}^{j+\nu} q_k  (x) \equiv 1$. 
	By \eqref{eq:Ef_nu} and \eqref{eq:Eg_nu} in Lemma \ref{lemma:fg_nu}, we can show that $E Q_{\nu,1}(x) = B_{\nu}(x) h^{\nu+1} + o(h^{\nu+1})$, and $E Q_{\nu,1}(x_k^*) = B_{\nu}(x_k^*) h^{\nu+1} + o(h^{\nu+1})$ for each $j\le k\le j+\nu$. Furthermore, by \eqref{eq:varf_nu} and \eqref{eq:varg_nu} in Lemma \ref{lemma:fg_nu}, we can show that $\var\{Q_{\nu,1}(x)\} = V(x)/(Nh) +  o(1/(Nh))$, and $\var\{Q_{\nu,1}(x_k^*)\} = V(x_k^*)/(Nh) +  o\{1/(Nh)\}$ for each $j\le k\le j+\nu$. 
	Since $B_{\nu}(\cdot)$ and $V(\cdot)$ are continuous under the assumed conditions, and $\max_{j\le k\le j+\nu} | x_k^* - x | \le \nu / J \to 0$ as $J\to\infty$, we have $E Q_{\nu,1}(x_k^*) - E Q_{\nu,1}(x) = o(h^{\nu+1})$ and $\var\{Q_{\nu,1}(x_k^*)\} = \var\{Q_{\nu,1}(x)\}\{1+o(1)\}$. Recall that $h = CN^{-1/(2\nu+3)}$ for some constant $C>0$. Then we have $\{E Q_{\nu,1}(x_k^*) - E Q_{\nu,1}(x)\} / \sqrt{\var\{Q_{\nu,1}(x)\}} = o(\sqrt{Nh^{2\nu + 3}}) = o(1)$. 
	Similar to the proof of Lemma \ref{lemma:Q1jQ1j+1}, we can show that $\cov\{Q_{\nu,1}(x_k^*),Q_{\nu,1}(x) \} =  \var\{Q_{\nu,1}(x)\}\{1+o(1)\}$ for each $j\le k\le j+\nu$. Together with above results, and using  Lemma \ref{lemma:difference}, we have
	\begin{align*}
		\tilde Q_{\nu,1}(x) - Q_{\nu,1}(x) =  \sum_{k=j}^{j+\nu} q_k(x) Q_{\nu,1}(x_k^*) - Q_{\nu,1}(x)= o_p\Big(\sqrt{\var\{Q_{\nu,1}(x)\}}\Big) = o_p(1/\sqrt{Nh}).
	\end{align*}
	This implies $\sqrt{Nh} \big\{\tilde Q_{\nu,1}(x) - Q_{\nu,1}(x)\big\} \to_p 0$ as $N\to\infty$. 
	Consequently, by Slutsky's theorem and \eqref{eq:Q_nu_1_normal}, we conclude that 
	\begin{align*}
		\sqrt{Nh}\Big\{\tilde Q_{\nu,1}(x)-B_{\nu}(x)h^{\nu+1}\Big\} \to_d \mN\Big(0, V(x)\Big),
	\end{align*}
	as $N\to \infty$. This completes the proof of Step 2. 
	
	Furthermore, by Slutsky's theorem and the results in Step 2, we can obtain that 
	\begin{align*}
		\sqrt{Nh}\Big\{\hat\mu_{\textup{PGPA},\nu}(x) - \mu(x) - B_{\nu}(x)h^{\nu+1}\Big\}\to_d \mN\Big(0, V(x) \Big).
	\end{align*}
	This ends the proof of the theorem.

	\subsection{Proof of Theorem \ref{thm:GPA_multi}}
	\label{append:thm:GPA_multi}
	
	In the following two steps, we first investigate the global multivariate NW estimator $\hat \mu(\bx)$. Subsequently, we derive the results for the multivariate linear interpolation based GPA estimator $\hat \mu_{\GPA, p}(\bx)$.
	
	\noindent\textbf{Step 1.}
	Recall that 
	\begin{align*}
		\hat f= \hat f(\bx) = \frac{1}{Nh^p} \sum_{i=1}^N \bK\left(\frac{X_{i}-\bx}{h} \right),\ \hat g= \hat g(\bx) = \frac{1}{Nh^p} \sum_{i=1}^N  \bK\left(\frac{X_{i}-\bx}{h} \right)Y_i,
	\end{align*}
	where $\bK(\bt) = \prod_{s=1}^p K(t_s)$ is the multivariate kernel function.
	Similar to the proof of Theorem \ref{thm:GPA_nu}, we can show that $\hat \mu - \mu=$
	\begin{align*}
		\frac{\hat g}{\hat f} - \frac{g}{f} 
		=& \left\{- \frac{\mu}{f}(\hat f - f) +  \frac{1}{f} (\hat g - g)\right\} + \left\{\frac{\tilde g}{\tilde f^3} (\hat f - f)^2 - \frac{1}{\tilde f^2} (\hat f - f)(\hat g - g)\right\} \\
		=& \Big\{Q_{p,1}\Big\} + \Big\{\mQ_{p} \Big\},
	\end{align*}
	where $\tilde f $ is between $\hat f$ and $f$, and $\tilde g$ is between $\hat g$ and $g$.
	Furthermore, by results in Lemma \ref{lemma:fg_p}, we can derive that 
	\begin{align}\label{eq:Q_p_1_normal}
		\sqrt{Nh^{p+1}}\Big\{Q_{p,1}(\bx)-B_{p}(\bx)h^{2}\big\}\to_d \mN\Big(0, V_p(\bx)\Big),
	\end{align}
	where
	\begin{align}\label{eq:B&V_p}
		B_p(\bx) = \frac{\kappa_2}{2} \left\{\tr\{\ddot \mu(\bx)\} + 2 \frac{\dot \mu(\bx) ^\top \dot f(\bx)}{f(\bx)} \right\},\ V_p(\bx) = \frac{\nu_0^p \sigma^2}{f(\bx)}.
	\end{align}
	Recall that $\kappa_2 = \int u^2 K(u) du$ and $\nu_0 = \int K^2(u) du$.
	In addition, by \eqref{eq:Ef_p}--\eqref{eq:varg_p} in Lemma \ref{lemma:fg_p}, we can obtain that $\mQ_{p}  = {\tilde g} / {\tilde f^3} (\hat f - f)^2 - ({1}/{\tilde f^2}) (\hat f - f)(\hat g - g) = O\{1 / (Nh^p)\} = o_p(1/\sqrt{Nh^{p}})$.
	
	From the above results, we can conclude that, for any fixed $\bx\in(0,1)^p$, we have
	\begin{align*}
		\sqrt{Nh^p}\Big\{\hat \mu(\bx) -\mu(\bx)- B_p (\bx)h^{2}\Big\}\to_d \mN\Big(0, V_p(\bx)\Big),
	\end{align*}
	where $B_{p}(\bx)$ and $V_p(\bx)$ are defined in \eqref{eq:B&V_p}.
	This ends the proof of Step 1.

	\noindent\textbf{Step 2.}
	From the proof of Step 1, we can obtain that
	$\hat \mu_{\textup{MGPA},p}(\bx) - \mu(\bx) = \tilde Q_{p,0} + \tilde Q_{p,1}  + \tilde \mQ_p$, where $\tilde Q_{p,0} = \tilde Q_{p,0}(x) = \sum_{k = 1}^{p+1} q_k  (\bx) \mu(\bx_{j_k}^*) - \mu(\bx)$, $\tilde Q_{p,1} = \tilde Q_{p,1}(\bx) = \sum_{k = 1}^{p+1} q_k (\bx) Q_{p,1}(\bx_{j_k} ^*)$, and $\tilde \mQ_{p}= \tilde \mQ_{p}(\bx) = \sum_{k = 1}^{p+1} q_k (\bx) \mQ_{p}(\bx_{j_k} ^*)$. 
	Note that $q_k(\bx),\ 1\le k\le p+1$ are linear functions that should be bounded for any $\bx\in [0,1]^p$. 
	In Step 1, we have shown that $\mQ_p = o_p(1/\sqrt{Nh^{p}})$. Then we can conclude that $\tilde \mQ_{p}= \tilde \mQ_{p}(\bx) = \sum_{k = 1}^{p+1} q_k (\bx) \mQ_{p}(\bx_{j_k} ^*) = o_p(1/\sqrt{Nh^{p}})$ as well. We next investigate the two terms $\tilde Q_{p,0}$ and $\tilde Q_{p,1}$.
	
	\noindent\textbf{Step 2.1.}
	Recall that $\mC(\bx) = \{\sum_{k=1}^{p+1} w_k\bx_{j_k}^* : \sum_{k=1}^{p+1} w_k = 1 \textup{ and } w_k \ge 0 \textup{ for each } k\} $ is a $p$-simplex with grid points $\bx_{j_k}^*\ (1\le k\le p+1)$ as vertices.
	Since $\mu(\cdot)$ is a twice continuously differentiable function, by error bound derived in Theorem 3.1 of \cite{waldron1998error}, we have
	\begin{align*}
		|\tilde Q_{p,0}(\bx)| = \left|\sum_{k = 1}^{p+1} q_k  (\bx) \mu(\bx_{j_k}^*) - \mu(\bx)\right| \le  \frac{p}{2J^2} \sup_{\bx' \in \mC(\bx)} \| \ddot \mu(\bx')\|.
	\end{align*}
	We then conclude that $\tilde Q_{p,0}(\bx) = O(1/J^2)$.
	
	\noindent\textbf{Step 2.2.} We next establish the asymptotic normality of $\tilde Q_{p,1}$. Note that $\mL(\bx) = \sum_{k=1}^{p+1} q_k(\bx)$ is a linear function interpolates $(\bx_{j_k}^*, 1)$ for each $ 1\le k\le p+1$. Note that $\{(\bx_{j_k}^*, 1):\ 1\le k\le p+1\}$ also uniquely determines the linear function $\mL_0(\bx)\equiv 1$. Then we must have $\mL(\bx)=\sum_{k=1}^{p+1} q_k(\bx)=\mL_0(\bx) \equiv 1$.
	By \eqref{eq:Ef_p} and \eqref{eq:Eg_p} in Lemma \ref{lemma:fg_p}, we can show that $E Q_{\nu,1}(\bx) = B_{p}(\bx) h^2 + o(h^2)$, and $E Q_{p,1}(\bx_{j_k}^*) = B_p(\bx_{j_k}^*) h^2 + o(h^2)$ for each $1 \le k\le p+1$. Furthermore, by \eqref{eq:varf_p} and \eqref{eq:varg_p} in Lemma \ref{lemma:fg_p}, we can show that $\var\{Q_{p,1}(\bx)\} = V_p(\bx)/(Nh^p) +  o(1/(Nh^p))$, and $\var\{Q_{\nu,1}(\bx_{j_k}^*)\} = V_p(\bx_{j_k}^*)/(Nh^p) +  o\{1/(Nh^p)\}$ for each $1\le k\le p+1$. 
	Recall that $B_p(\cdot)$ and $V_p(\cdot)$ defined in \eqref{eq:B&V_p} are continuous functions under the assumed conditions. 
	Since $\max_{1\le k\le p+1} \|\bx_{j_k}^* - \bx\| \le \sqrt{p} / J \to 0  $, we have $EQ_{p,1}(\bx_{j_k}^*) - EQ_{p,1}(\bx) = o(h^2)$ and $\var \{Q_{p,1}(\bx_{j_k}^*) \}= \var \{Q_{p,1}(\bx) \} \{1+o(1)\}$. Recall that $h = C N^{-1/(p+4)}$ for some constant $C>0$. Then we have 
	$\{EQ_{p,1}(\bx_{j_k}^*) - EQ_{p,1}(\bx) \} / \sqrt{\var\{Q_{p,1}(\bx)} \} = o(\sqrt{Nh^{p+4}})= o(1)$. Similar to the proof of Lemma \ref{lemma:Q1jQ1j+1}, we can show that $\cov\{Q_{p,1}(\bx_{j_k}^*) , Q_{p,1}(\bx)  \} = \var \{Q_{p,1}(\bx) \}\{ 1+o(1)\}$. Together with above results and using Lemma \ref{lemma:difference}, we should have
	\begin{align*}
		\tilde Q_{p,1}(\bx ) - Q_{p,1}(\bx ) =  \sum_{k=1}^{p+1} q_k(\bx ) Q_{p,1}(\bx_{j_k}^*) - Q_{p,1}(\bx)= o_p\Big(\sqrt{\var\{Q_{p,1}(\bx)\}}\Big) = o_p(1/\sqrt{Nh^p}).
	\end{align*}
	This implies $\sqrt{Nh^p} \big\{\tilde Q_{p,1}(\bx) - Q_{p,1}(\bx)\big\} \to_p 0$ as $N\to\infty$. 
	Consequently, by Slutsky's theorem and \eqref{eq:Q_p_1_normal}, we conclude that 
	\begin{align*}
		\sqrt{Nh^p}\Big\{\tilde Q_{p,1}(\bx)-B_{p}(\bx)h^2\Big\} \to_d \mN\Big(0, V_p(\bx)\Big),
	\end{align*}
	as $N\to \infty$. This completes the proof of Step 2. 
	
	Furthermore, by Slutsky's theorem and the results in Step 2, we can obtain that 
	\begin{align*}
		\sqrt{Nh^p}\Big\{\hat\mu_{\textup{MGPA},p}(\bx) - \mu(\bx) - B_{p}(\bx)h^2\Big\}\to_d \mN\Big(0, V_p(\bx) \Big).
	\end{align*}
	This ends the proof of the theorem.


	\section{Useful Lemmas and Their Proofs}
	\label{append:lemmas}
	
	\begin{lemma} \label{lemma:fmgm}
		Under the conditions assumed in Theorem \ref{thm:global&OS}, we have
		\begin{align}
			&E \hat f_m = f + \frac{\kappa_2}{2} \ddot f h^2 + o(h^2)  \label{eq:Efm}\\
			&\var(\hat f_m) = \frac{\nu_0  f}{nh} - \frac{ f^2}{n} + \left(\frac{\nu_2 \ddot f}{2}\right)\frac{h}{n} + o\left(\frac{h}{n}\right), \label{eq:varfm} \\
			&E \hat g_m = \mu f + \frac{\kappa_2}{2} (\ddot \mu f + 2\dot \mu \dot f+ \mu \ddot f ) h^2 + o(h^2), \label{eq:Egm}\\
			&\var(\hat g_m) = \frac{(\mu^2 + \sigma^2)\nu_0 f }{nh} - \frac{ \mu^2 f^2}{n} \nonumber\\
			&~~~~~~~~~~~~~ + \left\{\frac{\nu_2(2\mu\ddot \mu f + 4\mu\dot\mu\dot f+2\mu\dot\mu^2f + \mu^2\ddot f + \sigma^2 \ddot f) }{2}\right\}\frac{h}{n}+ o\left(\frac{h}{n}\right), \label{eq:vargm}\\
			&\cov(\hat f_m, \hat g_m) = \frac{\nu_0 \mu f }{nh} -  \frac{ \mu f^2}{n} + \left\{\frac{\nu_2(\ddot\mu f + 2\dot \mu \dot f+\mu \ddot f)}{2}\right\}\frac{h}{n} +o\left(\frac{h}{n} \right), \label{eq:covfmgm}\\
			&E(\hat f_m - f)(\hat g_m - g) = \cov(\hat f_m, \hat g_m)  + O(h^4), \label{eq:Efmgm}\\
			&\var\left\{(\hat f_m - f)^2 \right\}= O\left(\frac{1}{n^2 h^2}\right), \label{eq:varfm2}\\
			&\var\left\{ (\hat f_m - f)(\hat g_m - g)  \right\} = O\left(\frac{1}{n^2 h^2}\right),\label{eq:varfmgm}\\
			&\cov\left\{(\hat f_m - f)^2, ( \hat f_m - f)(\hat g_m - g)\right\}	= O\left(\frac{1}{n^2 h^2}\right)\label{eq:covfm2fmgm}.
		\end{align}
	\end{lemma}
	
	\begin{proof}
		The proofs of these results are very similar, so we only calculate \eqref{eq:varfm}, \eqref{eq:covfmgm}, \eqref{eq:varfm2} and \eqref{eq:varfmgm} as examples in the following.
		
		\noindent\textbf{Proof of \eqref{eq:varfm}.} Note that $\var(\hat f_m) = n^{-1} \var \{K_h(X_i - x) \} = n^{-1} h^{-2}E\xi_i^2$, where $\xi_i = K\{(X_i - x) / h\}  - EK\{(X_i - x) / h\}$. In fact, 
		\begin{align*}
			EK^2\left(\frac{X_i- x}{h} \right) =& \int K^2\left(\frac{u- x}{h} \right) f(u)du
			= h\int K^2(t) f(x+ht)dt\\
			=& h\int K^2(t) \left\{f(x)+ \dot f(x)ht + \ddot f(\tilde x) h^2t^2 / 2  \right\}dt,\\
			=& \nu_ 0 fh + \nu_2 \ddot f h^3 / 2 + o(h^3),
		\end{align*}
		where $\tilde x$ is between $x$ and $x+ht$, and the last equality uses condition \ref{cond:smoothness} and dominated convergence theorem. Similarly, $\left[ EK\left\{(X_i- x)/h\right\}\right]^2 = f^2 h^2 + o(h^3)$.
		Then we have $E\xi_i^2 = \nu_0 f h  - f^2 h^2 + \nu_2 \ddot f h^3 / 2+ o(h^3) $, and thus \eqref{eq:varfm} follows.
		
		\noindent\textbf{Proof of \eqref{eq:covfmgm}.} Note that $\cov(\hat f_m, \hat g_m) = n^{-1}h^{-2} E\xi_i \zeta_i$, where $\xi_i = K\{(X_i - x) / h\}  - EK\{(X_i - x) / h\}$ and $\zeta_i = K\{(X_i - x) / h\}Y_i  - EK\{(X_i - x) / h\}Y_i$. By some calculations, we can obtain that 
		\begin{align*}
			E K^2\left(\frac{X_i- x}{h} \right)Y_i =& E K^2\left(\frac{X_i- x}{h} \right)\mu(X_i) = h \int K^2(t) f(x+ht)\mu(x+ht) dt \\
			=& \nu_0 \mu f h + \nu_2(\ddot\mu f+ 2\dot \mu \dot f+\mu \ddot f )h^3/2 + o(h^3),
		\end{align*}
		and  $\left[EK\{(X_i - x) / h\}\right]\left[EK\{(X_i - x) / h\}Y_i\right] = \mu f^2 h^2 + o(h^3)$. Then we have $E\xi_i\zeta_i = \nu_0 \mu f h -\mu f^2h^2+ \nu_2(\ddot\mu f+ 2\dot \mu \dot f+\mu \ddot f )h^3/2 + o(h^3) $, and thus \eqref{eq:covfmgm} follows.
		
		\noindent  \textsc{Proof of \eqref{eq:varfm2}.} Note that $\var\left\{(\hat f_m - f)^2 \right\} \le E(\hat f_m - f)^4=  E(\hat f_m - E\hat f_m)^4 + 4E(\hat f_m - E\hat f_m)^3 (E\hat f_m - f) + 6E(\hat f_m - E\hat f_m)^2(E\hat f_m - f)^2 + (E\hat f_m - f)^4$. We can calculate that $E(\hat f_m - E \hat f_m)^4=$
		\begin{align*}
			(nh)^{-4} \left\{ n E\xi_i^4 + 3n(n-1) (E\xi_i^2)^2  \right\}
			= &O\left\{ (nh)^{-3}\right\} + 3\nu_0^2 f^2 (nh)^{-2} \left\{ 1 + o(1)\right\}\\
			=& 3\nu_0^2 f^2 (nh)^{-2}+ o \left\{(nh)^{-2}\right\},
		\end{align*}
		and $ E(\hat f_m - E \hat f_m)^3 = O\left\{  ( nh)^{-2}\right\}$. Recall that $E\hat f_m - f = O(h^2)$, $E(\hat f_m - E\hat f_m)^2 = \var(\hat f_m)  = O\left\{(nh)^{-1} \right\}$, and $E(\hat f_m - f)^2 =  O\left\{ (nh)^{-1} + h^4\right\}  $. Together with these results and the facts $h^3/n = (nh^5)/(n^2h^2) =O \left\{(Nh^5 /M)(nh)^{-2}\right\} = O\{(nh)^{-2}\}$ and $h^8 = (nh^5)^2 / (n^2h^2) = O\left\{ (Nh^5/M)^2  (nh)^{-2} \right\} = O\{(nh)^{-2}\}$, we have $\var\left\{(\hat f_m - f)^2 \right\} = O \left\{(nh)^{-2}\right\}$. This ends the proof of \eqref{eq:varfm2}.
		
		\noindent  \textsc{Proof of \eqref{eq:varfmgm}.} Note that $\var\left\{ (\hat f_m - f)(\hat g_m - g)  \right\} \le E (\hat f_m - f)^2(\hat g_m - g) ^2 = E \{( \hat f_m - E \hat f_m) + (E \hat f_m - f)\}^2 \{(\hat g_m - E \hat g_m) + (E \hat g_m -  g)\}^2 = 
		E (\hat f_m - E\hat f_m )^2(\hat g_m - E\hat g_m)^2 
		+ E (\hat f_m - E\hat f_m )^2( E\hat g_m - g)^2 
		+ 2E (\hat f_m - E\hat f_m )^2(\hat g_m - E\hat g_m) ( E\hat g_m - g) 
		+  ( E\hat f_m - f)^2  E(\hat g_m - E\hat g_m)^2
		+  ( E\hat f_m - f)^2 ( E\hat g_m - g)^2
		+ 2 E(\hat f_m - E\hat f_m ) ( E\hat f_m - f) (\hat g_m - E\hat g_m)^2
		+ 4 E(\hat f_m - E\hat f_m)( E\hat f_m - f) (\hat g_m - E\hat g_m)(\hat g_m - g ) $. We can calculate that $E (\hat f_m - E\hat f_m )^2(\hat g_m - E\hat g_m)^2=$
		\begin{align*}
			& (nh)^{-4}\left\{ n E\xi_i^2\zeta_i^2 + n(n-1) E\xi_i^2  E\zeta_i^2 + 2n(n-1)(E\xi_i\zeta_i)^2  \right\}\\
			=& O\left\{ (nh)^{-3}\right\} + (\mu^2 + \sigma^2)\nu_0^2 f^2  (nh)^{-2}\left\{ 1 + o(1)\right\} + 2 \nu_0^2 \mu^2 f^2 (nh)^{-2}  \left\{ 1 + o(1)\right\}\\
			=& (3\mu^2 + \sigma^2)\nu_0^2 f^2 (nh)^{-2}  + o\left\{ (nh)^{-2}\right\}.
		\end{align*}
		Since $E\hat g_m - g = O(h^2)$, we have $E (\hat f_m - E\hat f_m )^2(\hat g_m - E\hat g_m) ( E\hat g_m - g) = O(n^{-2}) = o\left\{ (nh)^{-2} \right\}$. Similarly, we can show that all other terms in the expression of $E ( f_m - f)^2(\hat g_m - g)^2$ are $O\left\{ (nh)^{-2} \right\}$ by using the fact that $nh^5 = (Nh^5) / M = O(1)$. Together with above results, we can obtain that $\var\left\{ (\hat f_m - f)(\hat g_m - g)  \right\}  = O\left\{ (nh)^{-2}\right\} $. This proves \eqref{eq:varfmgm}.

	\end{proof}

	\begin{lemma} \label{lemma:concen}
		Under the conditions assumed in Theorem \ref{thm:global&OS} (b), we have
		\begin{align}
			&P\Big( |\hat f_m - E \hat f_m|>t \Big) \le 2 \exp\left[-\frac{nht^2}{2\nu_0 f \{1+o(1)\} + 2 C_K t / 3}\right], \label{ineq:fm_concen}\\
			&P \Big( |\hat g_m - E \hat g_m|>t \Big) \le 2 \exp\left[-\frac{nht^2}{2 (\mu ^2 + \sigma^2)\nu_0 f \{1+o(1)\} + 2 C_ K C_Y t / 3}\right], \label{ineq:gm_concen}
		\end{align}
		where $C_K = \sup_{u} |K(u)|$.
		Furthermore, we have
		\begin{align}
			P\left\{ \max_{1\le m\le M} \left(|\tilde g_m / \tilde f_m^4|, |1 / \tilde f_m^3| \right)> C_0\right\} \le C_1 \exp\left( \log M - C_2nh \right), \label{ineq:bound}
		\end{align}
		if $n$ is sufficiently large, where $C_1,\ C_2>0$ are some constants.
		Here $\hat f_m$, $\tilde f_m$, $\hat g_m$, and $\tilde g_m$ are defined in the proof of Theorem \ref{thm:global&OS}.
	\end{lemma}
	
	\begin{proof}
		Note that $\hat f_m - E\hat f_m = (nh)^{-1} \sum_{i\in \mS_m} \xi_i$, where $\xi_i = K\{(X_i - x) / h\}  - EK\{(X_i - x) / h\}$. Then by Bernstein's inequality \citep{bennett1962probability}, we have 
		\begin{align*}
			P\left(|\hat f_m - E \hat f_m|>t \right)  = P\left(\left|\sum_{i\in \mS_m}  \xi_i \right| > nht \right)
			\le  2\exp\left\{ - \frac{(nht)^2 }{2 \var(\sum_i \xi_i) + 2C_K nht / 3} \right\}.
		\end{align*}
		Recall that $ \var(\sum_i \xi_i) = n \var(\xi_i) =  \nu_0 f nh\{1+o(1)\}$. Thus, the probability inequality \eqref{ineq:fm_concen} follows immediately. The inequality \eqref{ineq:gm_concen} can be derived similarly. 
		
		To derive \eqref{ineq:bound}, we need to bound $|\tilde g_m / \tilde f_m^4|$ and $|1/\tilde f_m^3|$ for $1\le m\le M$ uniformly. We only deal with $\max_m |\tilde g_m / \tilde f_m^4|$ in the following, since $\max_m |1/\tilde f_m^3|$ can be handled similarly. Then, it is sufficient to show that 
		\begin{align}
			&P\left( \max_m |\tilde f_m - f| > f/2 \right) \le C_3 \exp(\log M - C_4 nh),\label{ineq:tildefm} \\
			&P\left( \max_m |\tilde g_m - g| > C_5 \right) \le C_6 \exp(\log M - C_7 nh).\label{ineq:tildegm}
		\end{align}
		Note that $|\tilde f_m - f| \le |\hat f_m - f| \le |\hat f_m - E\hat f_m| + |E\hat f_m - f|$, and $|E\hat f_m - f| = O(h^2)$. Then by \eqref{ineq:fm_concen}, we have
		\begin{align*}
			P\left( \max_m |\hat f_m - f| > f/2 \right)\le& P\left( \max_m |\tilde f_m - E\hat f_m| > f/4 \right) + P\left( |E\hat f_m - f| > f/4 \right)\\
			\le & 2M\exp( -C_8 nh ),
		\end{align*}
		if $h$ is small enough such that $h^2 \ll f/4$. Thus, the inequality \eqref{ineq:tildefm} follows. The inequality \eqref{ineq:tildegm} can be derived in almost the same way.
		This completes the proof of the lemma.
		
	\end{proof}

	\begin{lemma} \label{lemma:Q1jQ1j+1}
		Under the conditions assumed in Theorem \ref{thm:GPA_asymptotic}, then for any $x_j^*$ satisfies $|x_j^* - x| / h \to \infty$, we have
		\begin{align}
			&\cov\Big\{\hat f(x_j^*), \hat f(x) \Big\} = \frac{\nu_0 f(x)}{Nh} + o\left(\frac{1}{Nh}\right), \label{eq:covfjfj+1}\\
			&\cov\Big\{\hat g(x_j^*), \hat g(x) \Big\} = \frac{\{\mu^2(x) + \sigma^2\} \nu_0f(x)}{Nh} + o\left(\frac{1}{Nh}\right),\label{eq:covgjgj+1}\\
			&\cov\Big\{\hat f(x_j^*), \hat g(x) \Big\} = \frac{\nu_0 \mu(x) f(x)}{Nh} + o\left(\frac{1}{Nh}\right)\label{eq:covfjgj+1}\\
			&\cov\Big\{\hat g(x_j^*), \hat f(x) \Big\} = \frac{\nu_0 \mu(x) f(x)}{Nh} + o\left(\frac{1}{Nh}\right)\label{eq:covgjfj+1}.
		\end{align}
		Moreover, we have
		\begin{align}
			&\var\Big\{Q_1(x_j^*)\Big\} = \left(\frac{\nu_0 \sigma^2}{f(x)}\right)\frac{1}{Nh} + o\left(\frac{1}{Nh}\right) \label{eq:varQ1xl},\\
			&\cov\Big\{Q_1(x_j^*), Q_1(x)  \Big\} = \left(\frac{\nu_0 \sigma^2}{f(x)}\right)\frac{1}{Nh} + o\left(\frac{1}{Nh}\right) \label{eq:covQ1jQ1j+1},
		\end{align}
		where $Q_1(\cdot)$ is defined in \eqref{eq:Q1xl}.
	\end{lemma}
	
	\begin{proof}
		Since the proofs of \eqref{eq:covfjfj+1}-\eqref{eq:covgjfj+1} are very similar, we prove \eqref{eq:covfjfj+1} and \eqref{eq:covfjgj+1} for examples.
		
		\noindent\textbf{Proof of \eqref{eq:covfjfj+1}.} Note that $ $
		\begin{align*}
			&\cov\Big\{\hat f(x_j^*), \hat f(x) \Big\}
			=\frac{1}{Nh^2} \cov\left\{  K\left(\frac{X_i - x_j^*}{h}\right), K\left(\frac{X_i - x}{h}\right)\right\}\\
			=&\frac{1}{Nh^2}  \left[E \left\{K\left(\frac{X_i - x_j^*}{h}\right) K\left(\frac{X_i - x}{h}\right)\right\}- \left\{EK\left(\frac{X_i - x_j^*}{h}\right)\right\}\left\{EK\left(\frac{X_i - x}{h}\right)\right\} \right] .
		\end{align*}
		By some calculations, we have
		\begin{align*}
			&E \left\{K\left(\frac{X_i - x_j^*}{h}\right) K\left(\frac{X_i - x}{h}\right)\right\} \\
			=& \int K\left(\frac{u - x_j^*}{h} \right)K\left(\frac{u - x}{h} \right) f(u)du
			=h\int K\left(t + \frac{x - x_j^*}{h} \right) K(t) f(x+ht)dt\\
			=&h\int K\left(t + \frac{x - x_j^*}{h} \right) K(t) \left\{f(x) + \dot f(\tilde x) ht \right\}dt=\nu_0 f(x) h \{1 + o(1) \},
		\end{align*}
		where $\tilde x$ is between $x$ and $x + ht$, $\nu_0 = \int K^2(u)du$, and the last line follows from the dominated convergence theorem and the fact that $ |(x-x_j^*) / h| \to 0$. Since $EK\{(X_i - x_j ^*) / h \} = O(h)$, we can obtain \eqref{eq:covfjfj+1} by using above results.
		
		\noindent\textbf{Proof of \eqref{eq:covfjgj+1}.} Note that $\cov\{\hat f(x_j^*), \hat g(x) \}  =$
		\begin{align*}
			& \frac{1}{N^2h^2} \cov\left\{ \sum_{i=1}^N K\left(\frac{X_i - x_j^*}{h}\right), \sum_{i=1}^N K\left(\frac{X_i - x}{h}\right)Y_i\right\}\\
			=&\frac{1}{Nh^2} \cov\left[  K\left(\frac{X_i - x_j^*}{h}\right), K\left(\frac{X_i - x}{h}\right)\left\{\mu(X_i) + \varepsilon_i\right\}\right]\\
			=&\frac{1}{Nh^2} \Bigg[ E \left\{K\left(\frac{X_i - x_j^*}{h}\right) K\left(\frac{X_i - x}{h}\right)\mu(X_i)\right\}\\
			&~~~~~~~~~~- \left\{EK\left(\frac{X_i - x_j^*}{h}\right)\right\}\left\{EK\left(\frac{X_i - x}{h}\right)\mu(X_i)\right\}  \Bigg].
		\end{align*}
		Similarly, we can compute that
		\begin{align*}
			&E \left\{K\left(\frac{X_i - x_j^*}{h}\right) K\left(\frac{X_i - x}{h}\right)\mu(X_i)\right\}\\
			=& h\int K\left(t + \frac{x - x_j^*}{h} \right) K(t)\mu(x + ht) f(x + ht) dt\\
			=& \nu_0 \mu(x)f(x)h \{1 + o(1) \}.
		\end{align*}
		Further note that $EK\{(X_i - x_j^*) / h \} = O(h)$ and $EK\left\{(X_i - x) / h \right\}\mu(X_i) = O(h)$.
		Thus, \eqref{eq:covfjgj+1} follows from above results.
		
		\noindent\textbf{Proof of \eqref{eq:varQ1xl}.} From \eqref{eq:Q1xl} and the proofs of \eqref{eq:varfm}, \eqref{eq:vargm}, and \eqref{eq:covfmgm}, we can obtain that $\var\Big\{Q_1(x_j^*) \Big\} =$
		\begin{align*}
			& \frac{\mu^2(x_j^*)}{f^2(x_j^*)} \var\Big\{\hat f(x_j^*)\Big\} + \frac{1}{f^2(x_j^*)} \var\Big\{\hat g(x_j^*) \Big\}  - 2\frac{\mu^2(x_j^*)}{f^2(x_j^*)} \cov\Big\{\hat f(x_j^*), \hat g(x_j^*)\Big\}\\
			=&\frac{\mu^2(x)}{f^2(x)}\frac{\nu_0 f(x)}{Nh}\{1+ o(1) \} + \frac{1}{f^2(x)}\frac{\{\mu^2(x) + \sigma^2\} \nu_0f(x)}{Nh} \{1+ o(1) \}\\
			&~~ -2 \frac{\mu(x)}{f^2(x)}\frac{\nu_0 \mu(x) f(x)}{Nh}\{1+ o(1) \}
			=\left\{\frac{\nu_0 \sigma^2}{f(x)}\right\}\frac{1}{Nh} + o\left(\frac{1}{Nh}\right),
		\end{align*}
		where the second line follows from the facts that $f(\cdot)$ and $\mu(\cdot)$ are continuous and the condition that $|x-x_j ^*|\to 0$. 
		
		\noindent\textbf{Proof of \eqref{eq:covQ1jQ1j+1}.} Similar to the proof of \eqref{eq:varQ1xl}, from \eqref{eq:Q1xl} and \eqref{eq:covfjfj+1}-\eqref{eq:covgjfj+1} we can obtain that $\cov\Big\{Q_1(x_j^*), Q_1(x) \Big\} =$
		\begin{align*}
			& \frac{\mu(x_j^*)\mu(x)}{f(x_j^*)f(x)} \cov\Big\{\hat f(x_j^*), \hat f(x) \Big\} + \frac{1}{f(x_j^*) f(x)}\cov\Big\{\hat g(x_j^*), \hat g(x) \Big\} \\
			& ~~-\frac{\mu(x_j^*)}{f(x_j^*)f(x)} \cov\Big\{\hat f(x_j^*), \hat g(x) \Big\} - \frac{\mu(x)}{f(x_j^*)f(x)} \cov\Big\{\hat g(x_j^*), \hat f(x) \Big\}\\
			=&\frac{\mu^2(x)}{f^2(x)}\frac{\nu_0 f(x)}{Nh}\{1+ o(1) \} + \frac{1}{f^2(x)}\frac{\{\mu^2(x) + \sigma^2\} \nu_0f(x)}{Nh} \{1+ o(1) \} \\
			& ~~ - \frac{\mu(x)}{f^2(x)}\frac{\nu_0 \mu(x) f(x)}{Nh}\{1+ o(1) \} - \frac{\mu(x)}{f^2(x)}\frac{\nu_0 \mu(x) f(x)}{Nh}\{1+ o(1) \}\\
			=& \left\{\frac{\nu_0 \sigma^2}{f(x)}\right\}\frac{1}{Nh} + o\left(\frac{1}{Nh}\right).
		\end{align*}
		This completes the proof the lemma. 
	\end{proof}

	\begin{lemma}\label{lemma:correlations}
		Consider three sequences of random variables $U_{n1}$, $U_{n2}$, and $V_n$. If $\cor(U_{n1},V_n) \to 1$ and $\cor(U_{n2},V_n) \to 1$ as $n\to \infty$, then we have $\cor(U_{n1},U_{n2}) \to 1$ as $n\to \infty$.
	\end{lemma}
	\begin{proof}
		Write $\rho_{n1} = \cor(U_{n1}, V_n)$,  $\rho_{n2} = \cor(U_{n2}, V_n)$, and $\rho_{n3} =  \cor(U_{n1}, U_{n2})$. Then we know that the correlation matrix of $(V_n, U_{n1}, U_{n2})$, denoted by
		\begin{align*}
			\bC_n = \begin{bmatrix}
				1 & \rho_{n1} & \rho_{n2} \\
				\rho_{n1} & 1 & \rho_{n3} \\
				\rho_{n2} & \rho_{n3} & 1
			\end{bmatrix},
		\end{align*}
		should be positive-semidefinite.
		Thus, its determinant should satisfy $|\bC_n| = 1 + 2 \rho_{n1}\rho_{n2}\rho_{n3} - \rho_{n1}^2 -\rho_{n2}^2-\rho_{n3}^2 \ge 0$ for each $n\ge 1$. 
		Note that $\lim_{n\to\infty} \rho_{n1} = 1$ and $\lim_{n\to\infty} \rho_{n2} = 1$.
		This implies that
		\begin{align*}
			\lim_{n\to\infty} | \bC_n | = \lim_{n\to\infty} (-1 + 2\rho_{n3} - \rho_{n3}^2 )= -\lim_{n\to\infty} (\rho_{n3} - 1)^2 \ge 0.
		\end{align*}
		Then we must have $\lim_{n\to\infty} \rho_{n3}=\lim_{n\to \infty} \cor(U_{n1},U_{n2}) = 1$. This completes the proof of the lemma. 
	\end{proof}

	\begin{lemma}\label{lemma:difference}
		Consider $K+1$ sequences of random variables $U_{n1}, \dots, U_{nK}$ and $V_n$. Suppose that $\var(U_{nk}) = \var(V_{n})\{ 1+ o(1)\}$, $\cov(U_{nk}, V_{n}) = \var(V_{n})\{ 1+ o(1)\}$, and $(EU_{nk} - EV_n) / \sqrt{\var(V_n)}  \to 0$ as $n\to \infty$ for each $1\le k \le K$.
		If $K$ finite numbers $q_k\, (1\le k\le K)$ satisfy $\sum_{k=1}^Kq_k = 1$, then we have
		\begin{align*}
			\sum_{k=1}^K q_k U_{nk} - V_n= o_p\Big(\sqrt{\var(V_n)}\Big).
		\end{align*}
	\end{lemma}
	
	\begin{proof}
		Let $\tilde U_{nk} = (U_{nk} - EU_{nk})/\sqrt{\var(V_n)}$ and $\tilde V_{nk} = (V_{nk} - EV_{nk})/\sqrt{\var(V_n)}$. Then we have 
		\begin{align*}
			\frac{1}{\sqrt{\var(V_n)}} \left\{\sum_{k=1}^K q_k U_{nk} - V_n\right\} =& \sum_{k=1}^K q_k \tilde U_{nk} - \tilde V_n +  \sum_{k=1}^K q_k \frac{EU_{nk} - EV_n}{\sqrt{\var(V_n)}}
			=&\sum_{k=1}^K q_k \tilde U_{nk} - \tilde V_n + o(1),
		\end{align*}
		where we have used $(EU_{nk} - EV_n) / \sqrt{\var(V_n)}  \to 0$ as $n\to\infty$ in the last equality. Thus, it suffices to show that $\mE_n = \sum_{k=1}^K q_k \tilde U_{nk} - \tilde V_n= o_p(1)$. We can compute that
		\begin{align}\label{eq:mE^2}
			E\mE_n^2 = E\left(\sum_{k=1}^K q_k \tilde U_{nk}\right)^2 + E\tilde V_n^2 - 2 \sum_{k=1}^K q_k \cov(\tilde U_{nk}, \tilde V_n).
		\end{align}
		Note that 
		\begin{align}\label{eq:sum_q_k^2}
			E\left(\sum_{k=1}^K q_k \tilde U_{nk}\right)^2 = \sum_{k_1=1}^K \sum_{k_2=1}^K q_{k_1} q_{k_2} \cov(\tilde U_{nk_1},\tilde U_{nk_2}).
		\end{align}
		If $k_1=k_2$, we should have $\cov(\tilde U_{nk_1},\tilde U_{nk_2}) = \var(U_{kn_1}) / \var(V_n) = 1+o(1)$. We next consider the case $k_1\ne k_2$. 
		In fact,
		\begin{align}
			&\cov(\tilde U_{nk_1},\tilde U_{nk_2}) = \cov( U_{nk_1},U_{nk_2}) / \var(V_n) \nonumber\\
			=& \cor(U_{nk_1},U_{nk_2})  \sqrt{\var(U_{nk_1}) \var(U_{nk_2})} \Big/ \var(V_n) \nonumber \\
			=& \cor(U_{nk_1},U_{nk_2}) \{1+o(1)\} \label{eq:Uk1_Uk2}
		\end{align}
		where we have used the condition $\var(V_{nk}) = \var(V_n) \{1+o(1)\}$ for each $1\le k \le K$ in the last equality.
		Note that $\cor(U_{nk},V_n) = \cov(U_{nk},V_n) /\sqrt{\var(U_{nk})\var(V_n)} = 1+o(1)$ for each $1\le k \le K$, since $\cov(U_{nk},V_n)=\var(V_n)\{1+o(1) \}$ and $\var(U_{nk})=\{1+o(1) \} $. 
		Consequently, $\cor(U_{nk_1},V_n)\to 1$ and $\cor(U_{nk_2},V_n)\to 1$ as $n\to \infty$.
		By Lemma \ref{lemma:correlations}, we conclude that $\cor(U_{nk_1},U_{nk_2})\to 1$ as $n\to \infty$. 
		Then by \eqref{eq:Uk1_Uk2}, we obtain that $\cov(\tilde U_{nk_1},\tilde U_{nk_2}) = 1+o(1)$. 
		Therefore, $\cov(\tilde U_{nk_1},\tilde U_{nk_2}) =  1+o(1)$ for each pair $1\le k_1, k_2 \le K$.
		Thus, by \eqref{eq:sum_q_k^2} and the condition $\sum_{k=1}^K q_k=1$, we can derive that 
		\begin{align*}
			E\left(\sum_{k=1}^K q_k \tilde U_{nk}\right)^2 = \sum_{k_1=1}^K \sum_{k_2=1}^K q_{k_1} q_{k_2} \{ 1+o(1)\} = \left(\sum_{k=1}^K q_k\right)^2\{ 1+o(1)\} = 1+o(1).
		\end{align*}
		Further note that $E\tilde V_n^2 = 1$, and $\cov(\tilde U_{nk}, \tilde V_n) = \cov(U_{nk}, V_{n}) / \var(V_n) = 1+o(1)$. Together with above results and \eqref{eq:mE^2}, we should have $ E\mE_n^2 = o(1)$. Then we conclude that $\mE_n = o_p(1)$. This completes the proof the lemma.

	\end{proof}

	\begin{lemma} \label{lemma:fg_nu}
		Under the conditions assumed in Theorem \ref{thm:GPA_nu}, we have
		\begin{align}
			&E \hat f =  f + \frac{\kappa_{\nu+1} }{(\nu+1)! } f^{(\nu+1)}  h^{\nu+1}  + o(h^{\nu+1}) \label{eq:Ef_nu}\\
			&\var(\hat f) = \frac{\nu_0  f}{Nh} + O\left(\frac{1}{N}\right), \label{eq:varf_nu} \\
			&E \hat g = \mu f + \frac{\kappa_{\nu+1} }{(\nu+1)! } \left\{\sum_{s=0}^{\nu+1} \binom{\nu+1}{s} \mu^{(s)} f^{(\nu+1-s)} \right\} h^{\nu+1} + o(h^{\nu+1}), \label{eq:Eg_nu}\\
			&\var(\hat g) = \frac{(\mu^2 + \sigma^2)\nu_0 f }{Nh} + O\left(\frac{1}{N}\right), \label{eq:varg_nu}\\
			&\cov(\hat f, \hat g) = \frac{\nu_0 \mu f }{Nh} +  O\left(\frac{1}{N} \right), \label{eq:covfg_nu}
		\end{align}
	\end{lemma}
	
	\begin{proof}
		Recall that 
		\begin{align*}
			&\hat f= \hat f(x) = \frac{1}{Nh} \sum_{i=1}^N K\left(\frac{X_i-x}{h} \right),\ \hat g= \hat g(x) = \frac{1}{Nh} \sum_{i=1}^N K\left(\frac{X_i-x}{h} \right)Y_i.
		\end{align*}
		We then proof the lemma as follows.
		
		\noindent\textbf{Proof of \eqref{eq:Ef_nu}}. Note that $E \hat f = h^{-1} E K\{(X_i-x)/h \} $. In fact, 
		\begin{align*}
			&EK\left(\frac{X_i-x}{h} \right) = \int K\left(\frac{u-x}{h}\right) f(u) du  = h \int K(t) f(x+ht) dt \\
			=&h \int \left\{\sum_{s=0}^{\nu} \frac{f^{(s)}(x)}{s!} (ht)^s + \frac{f^{(\nu+1)}(\tilde x)}{(\nu+1)!} (ht)^{\nu+1} \right\} K(t) dt\\
			=& hf + h^{\nu+2} \frac{f^{(\nu+1)}(x) }{(\nu+1)! } \kappa_{\nu+1}  + h^{\nu+2}  \frac{1}{(\nu+1)! } \int \left\{f^{(\nu+1)}(\tilde x) - f^{(\nu+1)}(x) \right\} t^{\nu+1} K(t) dt\\
			=& hf + h^{\nu+2} \frac{f^{(\nu+1)}(x) }{(\nu+1)! } \kappa_{\nu+1}  + o(h^{\nu+2}),
		\end{align*}
		where $\tilde x$ is between $x$ and $x+ht$, and $\kappa_{\nu+1} = \int u^{\nu+1} K(u) du$. Here, we have used the properties of higher order kernel $K(\cdot)$ in the fourth equality and the dominated convergence theorem in the last equality. Then, we can immediately obtain \eqref{eq:Ef_nu}.
		
		\noindent\textbf{Proof of \eqref{eq:varf_p}}. Note that $\var(\hat f)  = N^{-1} h^{-2}E\xi_i^2$, where $\xi_i = K\{(X_i - x) / h\}  - EK\{(X_i - x) / h\}$. In fact, 
		\begin{align*}
			&EK^2\left(\frac{X_i- x}{h} \right) = \int K^2\left(\frac{u- x}{h} \right) f(u)du
			= h\int K^2(t) f(x+ht)dt\\
			=& h\int K^2(t) \left\{f(x)+ \dot f(\tilde x)ht   \right\}dt=\nu_ 0^p f h + O(h^2),
		\end{align*}
		where $\nu_0 = \int K^2(u) du$. By proof of \eqref{eq:Ef_nu}, we know that $\left[ EK\left\{(X_i- x)/h\right\}\right]^2 = O(h^2)$.
		Then we have $E\xi_i^2 = \nu_0 f h  +  O(h^2) $, and thus \eqref{eq:varg_nu} follows.

		\noindent\textbf{Proof of \eqref{eq:Eg_nu}.}
		Note that $E \hat g = h^{-1} E [K \{(X_i-x)/h  \} Y_i ] = h^{-1} E[K \{(X_i-x)/h  \} \mu(X_i) ]  $, since $\varepsilon_i$ is independent of $X_i$. Note that $g = \mu f$ is also $(\nu+1)$-times continuously differentiable under the assumed conditions. Similar to the proof of \eqref{eq:Ef_nu}, we can obtain that 
		\begin{align*}
			&EK\left(\frac{X_i-x}{h} \right) \mu(X_i) = \int K\left(\frac{u-x}{h}\right) \{\mu(u) f(u)\} du  = h \int K(t) g(x+ht) dt \\
			=&h \int \left\{\sum_{s=0}^{\nu} \frac{g^{(s)}(x)}{s!} (ht)^s + \frac{g^{(\nu+1)}(\tilde x)}{(\nu+1)!} (ht)^{\nu+1} \right\} K(t) dt\\
			=& hg + h^{\nu+2} \frac{g^{(\nu+1)}(x) }{(\nu+1)! } \kappa_{\nu+1}  + o(h^{\nu+2}).
		\end{align*}
		By Leibniz rule, we have 
		\begin{align*}
			g^{(\nu+1)} = (\mu f)^{(\nu+1)} = \sum_{s=0}^{\nu+1} \binom{\nu+1}{s} \mu^{(s)} f^{(\nu+1-s)}.
		\end{align*}
		Thus, \eqref{eq:Eg_nu} follows immediately.

		\noindent\textbf{Proof of \eqref{eq:varg_nu}}. Note that $\var(\hat g)  = N^{-1} h^{-2}E\zeta_i^2$, where $\zeta_i = K\{(X_i - x) / h\}Y_i  - EK\{(X_i - x) / h\}Y_i$. Similarly, we can compute that
		\begin{align*}
			&E\left\{K\left(\frac{X_i- x}{h} \right) Y_i \right\}^2 = E\left[K\left(\frac{X_i- x}{h} \right) \Big\{\mu(X_i) + \varepsilon_i \Big\}  \right]^2 \\
			=& E\left\{K^2\left(\frac{X_i- x}{h} \right) \mu^2(X_i)  \right\} +  \sigma^2 EK^2\left(\frac{X_i- x}{h} \right) \\
			=&\int K^2\left(\frac{u- x}{h} \right) \mu^2(u) f(u)du + \sigma^2 \int K^2\left(\frac{u- x}{h} \right) f(u)du \\
			=& h \int K^2(t) \mu^2(x+ht) f(x+ht)dt + h \sigma^2 \int K^2(t) f(x+ht)dt\\
			=& h (\mu^2 + \sigma^2) \nu_0 f  + O(h^2).
		\end{align*} 
		By proof of \eqref{eq:Eg_nu}, we know that $\left[ EK\left\{(X_i- x)/h\right\}\right]^2 = O(h^2)$.
		Then we have $E\zeta_i^2 = (\mu^2 + \sigma^2) \nu_0 f h+ O(h^2)$, and thus \eqref{eq:varf_nu} follows.

		\noindent\textbf{Proof of \eqref{eq:covfg_nu}}. Note that $\cov(\hat f, \hat g) = N^{-1} h^{-2}E \xi_i \zeta_i$, where $\xi_i = K\{(X_i - x) / h\}  - EK\{(X_i - x) / h\}$ and $\zeta_i = K\{(X_i - x) / h\}Y_i  - EK\{(X_i - x) / h\}Y_i$. Similarly, we can compute that
		\begin{align*}
			&EK^2\left(\frac{X_i- x}{h} \right) Y_i =EK^2\left(\frac{X_i- x}{h} \right) \mu(X_i) \\
			=& h \int K^2(t) \mu(x+ht) f(x+ht)dt
			= h \nu_0 \mu f  + O(h^2),
		\end{align*} 
		By proofs of \eqref{eq:Ef_nu} and \eqref{eq:Eg_nu}, we know that $\left[ EK\left\{(X_i- x)/h\right\}\right] \left[ EK\left\{(X_i- x)/h\right\} Y_i\right]  = O(h^2)$.
		Then we have $E\xi_i \zeta_i = \nu_0 \mu f h+ O(h^2)$, and thus \eqref{eq:covfg_nu} follows.

	\end{proof}

	\begin{lemma} \label{lemma:fg_p}
		Under the conditions assumed in Theorem \ref{thm:GPA_multi}, we have
		\begin{align}
			&E \hat f =  f + \frac{\kappa_2 }{2 }  \tr(\ddot f)  h^{2}  + o(h^2) \label{eq:Ef_p}\\
			&\var(\hat f) = \frac{\nu_0^p  f}{Nh^p} + O\left(\frac{1}{Nh^{p-1}}\right), \label{eq:varf_p} \\
			&E \hat g = \mu f + \frac{\kappa_2 }{2 } \left\{\mu \tr(\ddot f) + f \tr(\ddot \mu) + 2 \dot \mu^\top \dot f \right\} h^2 + o(h^{\nu+1}), \label{eq:Eg_p}\\
			&\var(\hat g) = \frac{(\mu^2 + \sigma^2)\nu_0^p f }{Nh^p} + O\left(\frac{1}{N h^{p-1}}\right), \label{eq:varg_p}\\
			&\cov(\hat f, \hat g) = \frac{\nu_0^p \mu f }{Nh^p} +  O\left(\frac{1}{N h^{p-1}} \right), \label{eq:covfg_p}
		\end{align}
	\end{lemma}
	
	\begin{proof}
		Note that 
		\begin{align*}
			\hat f= \hat f(\bx) = \frac{1}{Nh^p} \sum_{i=1}^N \bK\left(\frac{X_{i}-\bx}{h} \right),\ \hat g= \hat g(\bx) = \frac{1}{Nh^p} \sum_{i=1}^N  \bK\left(\frac{X_{i}-\bx}{h} \right)Y_i,
		\end{align*}
		where $\bK(\bt) = \prod_{s=1}^p K(t_s)$ is the multivariate kernel function.
		Recall that condition \ref{cond:kernel} assumes $K(\cdot)$ is a symmetric probability density function. Thus, we should have
		\begin{align*}
			\int  \bK(\bt)  d\bt =& \prod_{s=1}^p \int K(t_s)   d t_s = 1,\\
			\int \bt \bK(\bt)  d\bt =& \left\{ \int t_s K(t_s)   d t_s,\ 1\le s \le p \right\}^\top = \boldsymbol{0} \in \mR^p, \\
			\int t_{s_1} t_{s_2} \bK(\bt)  d\bt =& 
			\begin{cases} 0 ,& \textup{ if } s_1 \ne s_2\\
				\kappa_2  ,& \textup{ if } s_1 = s_2
			\end{cases}\\
			\int \bK^2(\bt)  d\bt = &\prod_{s=1}^p \int K^2(t_s)   d t_s  = \nu_0^p,
		\end{align*}
		where $ \kappa_2= \int t^2 K(t)dt $ and $\nu_0 = \int K^2(t) dt$.
		We then proof the lemma as follows.
		
		\noindent\textbf{Proof of \eqref{eq:Ef_p}}. Note that $E \hat f = h^{-p} E \bK\{(X_{i}-\bx)/h \} $. In fact, 
		\begin{align*}
			&E \bK\left(\frac{X_{i}-\bx}{h} \right) = \int  \bK\left(\frac{\bu-\bx}{h}\right) f(\bu) d\bu  = h^p \int \bK(\bt) f(\bx+h \bt) d\bt \\
			=&h^p \int \bK(t) \left\{ f(\bx) + h \bt^\top \dot f(\bx) + (h^2/2) \bt^\top \ddot f(\tilde\bx) \bt  \right\} d \bt\\
			=&h^p f(\bx) + (h^{p+2} / 2) \left[ \int \bK(t) \bt^\top \ \ddot f(\bx)  \bt d\bt +  \int \bK(t) \bt^\top \Big\{\ddot f(\tilde\bx) - \ddot f(\bx)  \Big\} \bt d\bt   \right]\\
			=& h^p\left\{ f + \frac{\kappa_2 }{2 }  \tr(\ddot f )  h^{2}  + o(h^2) \right\},
		\end{align*}
		where $\tilde\bx$ is between $\bx$ and $\bx + h \bt$.
		Here, we have used the condition $\ddot f $ is continuous and the dominated convergence theorem in the last equality.  Then \eqref{eq:Ef_p} immediately follows.
		
		\noindent\textbf{Proof of \eqref{eq:varf_p}}. 
		Note that $\var(\hat f) = N^{-1} h^{-2p}E\xi_i^2$, where $\xi_i = \bK\{(X_i - \bx) / h\}  - E \bK\{(X_i - \bx) / h\}$. In fact, 
		\begin{align*}
			&E \bK^2\left(\frac{X_i- \bx}{h} \right) = \int \bK^2\left(\frac{\bu- \bx}{h} \right) f(\bu)d \bu
			= h^p \int \bK^2(\bt) f(\bx+h \bt)d\bt\\
			=& h^p \int \bK^2(\bt) \left\{ f(\bx)+ h\bt^\top \dot f(\tilde \bx)   \right\}dt=\nu_ 0^p f h^p + O(h^{p+1}).
		\end{align*}
		By proof of \eqref{eq:Ef_p}, we know that $\left[ E \bK\left\{(X_i- \bx)/h\right\}\right]^2 = O(h^{2p})$.
		Then we have $E\xi_i^2 = \nu_0^p f h^p  +  O(h^{p+1}) $, and thus \eqref{eq:varf_p} follows.
		
		\noindent\textbf{Proof of \eqref{eq:Eg_p}}. 
		Note that $E \hat g = h^{-p} E [\bK \{(X_i-\bx)/h  \} Y_i ] = h^{-p} E[\bK \{(X_i-\bx)/h  \} \mu(X_i) ]  $, since $\varepsilon_i$ is independent of $X_i$. Note that $g = \mu f$ has continuous second order derivative $\ddot g$ under the assumed conditions. Similar to the proof of \eqref{eq:Ef_p}, we can obtain that 
		\begin{align*}
			&E\bK\left(\frac{X_i-\bx}{h} \right) \mu(X_i) = \int \bK\left(\frac{\bu-\bx}{h}\right) \{\mu(\bu) f(\bu)\} d\bu  = h^{p} \int \bK(t) g(\bx+h\bt) d\bt \\
			=&h^p\left\{ g + \frac{\kappa_2 }{2 }  \tr(\ddot g)  h^{2}  + o(h^2) \right\}
		\end{align*}
		By chain rule, we have $\ddot g =  \mu \ddot f + \ddot \mu f + \dot \mu \dot f^\top + \dot f \dot \mu^\top$.
		Thus, $\tr(\ddot g)  = \mu \tr(\ddot f) + f \tr(\ddot \mu) + 2 \dot \mu^\top \dot f$, and \eqref{eq:Eg_p} follows immediately.
		
		\noindent\textbf{Proof of \eqref{eq:varg_p}}. Note that $\var(\hat g)  = N^{-1} h^{-2p}E\zeta_i^2$, where $\zeta_i = \bK\{(X_i - \bx) / h\}Y_i  - E\bK\{(X_i - \bx) / h\}Y_i$. Similarly, we can compute that
		\begin{align*}
			&E\left\{\bK\left(\frac{X_i- \bx}{h} \right) Y_i \right\}^2 = E\left[\bK\left(\frac{X_i- \bx}{h} \right) \Big\{\mu(X_i) + \varepsilon_i \Big\}  \right]^2 \\
			=& E\left\{\bK^2\left(\frac{X_i- \bx}{h} \right) \mu^2(X_i)  \right\} +  \sigma^2 E\bK^2\left(\frac{X_i- \bx}{h} \right) \\
			=&\int \bK^2\left(\frac{\bu- \bx}{h} \right) \mu^2(\bu) f(\bu)d\bu + \sigma^2 \int \bK^2\left(\frac{\bu- \bx}{h} \right) f(\bu)d\bu \\
			=& h^p \int \bK^2(\bt) \mu^2(\bx+h\bt) f(\bx+h\bt)d\bt + h^p \sigma^2 \int \bK^2(\bt) f(\bx+h\bt)d\bt\\
			=&h^p (\mu^2 + \sigma^2) \nu_0^p f   + O(h^{p+1}).
		\end{align*} 
		By proof of \eqref{eq:Eg_p}, we know that $\left[ E\bK\left\{(X_i- \bx)/h\right\}\right]^2 = O(h^{2p})$.
		Then we have $E\zeta_i^2 = (\mu^2 + \sigma^2) \nu_0^p f h^p+ O(h^{p+1})$, and thus \eqref{eq:varg_p} follows.

		\noindent\textbf{Proof of \eqref{eq:covfg_p}}. Note that $\cov(\hat f, \hat g) = N^{-1} h^{-2p}E \xi_i \zeta_i$, where $\xi_i = \bK\{(X_i - \bx) / h\}  - E\bK\{(X_i - \bx) / h\}$ and $\zeta_i = \bK\{(X_i - \bx) / h\}Y_i  - E\bK\{(X_i - \bx) / h\}Y_i$. Similarly, we can compute that
		\begin{align*}
			&E\bK^2\left(\frac{X_i- \bx}{h} \right) Y_i =E\bK^2\left(\frac{X_i- \bx}{h} \right) \mu(X_i) \\
			=& h^p \int \bK^2(t) \mu(\bx+h\bt) f(\bx+h\bt)d\bt
			= h^p \nu_0^p \mu f  + O(h^{p+1}),
		\end{align*} 
		By proofs of \eqref{eq:Ef_p} and \eqref{eq:Eg_p}, we know that $\left[ E\bK\left\{(X_i- \bx)/h\right\}\right] \left[ EK\left\{(X_i- \bx)/h\right\} Y_i\right]  = O(h^{2p})$.
		Then we have $E\xi_i \zeta_i = \nu_0^p \mu f h^p+ O(h^{p+1})$, and thus \eqref{eq:covfg_p} follows.
		
	\end{proof}

		

	\section{Extension to Non-compactly Supported Covariate} \label{append:non_compact}
	
	The proposed GPA method can be extended to the case where the covariate $X_i$ is not compactly supported. Without loss of generality, assume that the probability density function of $X_i$ is supported on the whole real line $\mR$. In such cases, we can consider a closed interval with a slowly diverging size, for example, $\mI_N = [-\log N, \log N]$. 
	Similarly, we can set $[2 \log N] J + 1$ equally spaced grid points $x_j^*$ on $\mI_N$, so that the inter-grid point distance remains $\Delta = 2\log N / ([ 2\log N] J) =O (1/J)$. Here, $[ r ]$ denotes the integer part of $r\in\mR$. It is worth noting that for any fixed $x\in\mR$, we should have $x\in \mI_N$ for sufficiently large $N$. 
	In this case, we should have $x \in [x_j^*, x_{j+1}^*]$ for some $j$. We can then conduct linear interpolation at $x$, following \eqref{eq:interpolation}, to obtain the GPA estimator $\hat \mu_{\GPA}(x)$ as
	\begin{align*} \label{eq:interpolation}
		\hat\mu_\textup{GPA}(x) = & \frac{x_{j+1}^* - x}{\Delta} \hat \mu(x_j^* ) + \frac{ x- x_j^*}{\Delta} \hat \mu(x_{j+1}^*).
	\end{align*}
	If $x\not \in \mI_N$, we simply set $\hat\mu_\textup{GPA}(x)  = 0$.
	Then, it can be shown that the asymptotic results in Theorem \ref{thm:GPA_asymptotic} remain valid; the proof is provided at the end of this section.
	This implies that we need $Jh\to\infty$ such that the resulting GPA estimator can be as efficient as the global estimator. Recall that $h = CN^{-1/5}$ for some constant $C>0$. Thus, the number of grid points should satisfy $[2 \log N] J + 1 \gg N^{1/5} \log N$. This means that the required number of grid points is only $\log N$ times that of the case with compact support.
	
	\textbf{Proof of Theorem \ref{thm:GPA_asymptotic} for non-compactly supported covariate.}  The proof is almost the same as that of Theorem \ref{thm:GPA_asymptotic} in Appendix \ref{append:thm:GPA_asymptotic}. Here, we provide a brief outline.
	
	First, note that $f(\cdot)$ and $\mu(\cdot)$ and their first two derivatives are continuous and hence bounded in a neighborhood of $x$ under condition \ref{cond:smoothness}.
	Then we can verify that the conclusions for the global estimator $\hat \mu(x)$ in Theorem \ref{thm:global&OS} (a) remain valid, provided conditions \ref{cond:smoothness}--\ref{cond:bandwidth} and $f(x)>0$. Thus, we should have $\hat \mu(x) - \mu(x) = Q_1 + Q_2 + \mQ$, where $\sqrt{Nh}\Big\{Q_1(x) - B(x)h^2\Big\} \to_d \mN\Big(0, V(x)\Big)$, $E Q_2 = O\left(h/N+h^4\right)$, $\var(Q_2)=O\left\{1/(Nh)^2\right\}$, and $ \mQ = O_p\left\{ 1/ (Nh)^{3/2} \right\}$. 
	
	To prove the asymptotic results in the theorem, it suffices to deal with the case that $N$ is sufficiently large such that $x\in\mI_N=[-\log N, \log N]$.
	Then we can assume that $x\in[x_j^*, x_{j+1}^*]$ for some $j$.
	Since $f(x)>0$, $f(\cdot)$ is continuous, and $\max_{k=j,j+1} |x_j^*-x| \le \Delta\to 0$, we should have $f(x_k^*) > 0$ for $k=j,j+1$ as well. Thus, the conclusions in Theorem \ref{thm:global&OS} (a) also hold for $\hat\mu(x_k^*),\ k=j,j+1$. 
	Then by the linear interpolation formula, we should have 
	$\hat\mu_\textup{GPA}(x) - \mu(x) = \tilde{Q}_0  + \tilde{Q}_1 + \tilde Q_2 + \tilde \mQ $, where 
	$\tilde Q_0 = \tilde{Q}_0(x) = q(x) \mu(x_j^* ) + \{1-q(x)\} \mu(x_{j+1}^*) - \mu(x)$,
	$\tilde Q_1 = \tilde{Q}_1(x) = q(x) Q_1(x_j^* ) + \{1-q(x)\} Q_1(x_{j+1}^*)$, 
	$\tilde Q_2 = \tilde{Q}_2(x) = q(x) Q_2(x_j^* ) + \{1-q(x)\}  Q_2(x_{j+1}^*)$, and 
	$\tilde \mQ   = \tilde \mQ(x)   = q(x)\mQ(x_j^* ) + \{1-q(x)\}  \mQ(x_{j+1}^*) $. Here, $q(x) = (x_{j+1}^* - x) / \Delta \in [0,1]$ and $1-q(x) = (x-x_j^*) / \Delta$ are the linear interpolation coefficients. 
	Recall that $\max_{k=j,j+1}|x_j^*-x| \le \Delta = O(1/J)$ and $Jh\to\infty$. Then, we can use the same arguments in Appendix \ref{append:thm:GPA_asymptotic} to obtain the desired results. This completes the proof.

	\section{Discussion of Basis Expansion Method}
	\label{append:basis_expansion}
	
	Another promising way to deal with prediction issue for nonparametric regression model is basis expansion method \citep{hastie2009elements}. 
	Specifically, we can approximate the mean function $\mu(\cdot)$ by the linear combination of a set of appropriately selected basis functions $B_k(x)$ with $1\le k\le K$. 
	Then we have $\mu(x) \approx \sum_{k=1}^K \alpha_k B_k(x)$, where $\{\alpha_k\}_{k=1}^K$ is a set of coefficients needs to be estimated. Once $\alpha_k$ is estimated (denoted by $\hat \alpha_k$), we can then predict $\mu(X^*)$ by $\hat \mu(X^*) = \sum_{k=1}^K \hat \alpha_k B_k(X^*)$ for any newly observed $X^*$. In this case, no more estimation is needed.
	
	Nevertheless, we should remark that there is also a significant price paid by this method. That is how to estimate $\balpha = (\alpha_1,\dots,\alpha_K)^\top \in \mR^K$. One natural way is the method of least squares estimation, which is obtained by minimizing the loss function $\mL(\balpha) = \sum_{i=1}^N \big\{Y_i - \balpha^\top B(X_i)  \big\}^2$, where $B(x) = (B_1(x),\dots,B_K(x))^\top \in \mR^K$ . This leads to $\hat  \balpha = (\bB^\top \bB)^{-1} \bB^\top \bY$, where $\bB = \big(B(X_1), \dots, B(X_N)\big)^\top \in \mR^{N\times K}$, and $\bY = (Y_1,\dots,Y_N)^\top \in \mR^N$. 
	If all the data points are placed on one single machine, then the calculation of $\hat \balpha$ is straightforward. However, if the data are distributed on different local machines, then how to compute $\hat \balpha$ becomes a tricky issue. There exist two possible solutions. 
	
	{\sc Solution 1. (Then one-shot type estimator)} We can compute the local estimator on $m$-th machine by minimizing $\mL_m(\balpha) = \sum_{i=1}^N \big\{Y_i - \balpha^\top B(X_i)  \big\}^2$. This leads to the $m$-th local estimator $\hat  \balpha_m = (\bB_m^\top \bB_m)^{-1} \bB_m^\top \bY_m$, where $\bB_m = \{B(X_i):i\in \mS_m \}^\top \in \mR^{n\times K}$ and $\bY_m = \{Y_i:i\in \mS_m \}^\top \in \mR^n$. Subsequently, the central machine collects these local estimators and constructs the one-shot (OS) estimator as $\hat \balpha_{\textup{OS}} = M^{-1} \sum_{m=1}^M \hat  \balpha_m$.
	The merit of this method is low communication cost since only a $K$-dimensional vector $\hat\alpha_m$ needs to be transmitted from each local machine to the central one. However, it requires that the data should be randomly distributed across different machines. Otherwise, the local estimators $\hat\alpha_m\ (1\le m\le M)$ could be seriously biased. This might make the final estimator $\hat\alpha_\textup{OS}$ also seriously biased. In contrast, our method is free from such a problem. 
	
	{\sc Solution 2. (The moment assembling method)} By this method, each local machine should report to the central machine two moment estimators. They are, respectively, $\bB_m^\top \bB_m $ and $\bB_m^\top \bY_m$. Then, the central machine can calculate $\hat \balpha$ as $\hat \balpha = \big(\sum_{m=1}^M \bB_m^\top \bB_m \big)^{-1} \sum_{m=1}^M \bB_m^\top \bY_m$.
	Note that $\bB_m^\top \bB_m$ is a $K\times K$ matrix. Moreover, we need $K\to\infty$ as $N\to \infty$. 
	Suppose we use second-order spline basis to approximate the twice continuously differentiable mean function $\mu(\cdot)$.
	Then by the theory of \cite{shen1998local}, the optimal $K$ should be of the order $K = O(N^{1/5})$. Thus, the communication cost of transmitting $\bB_m^\top \bB_m $ is of the order $O(N^{2/5})$. This leads to an expensive communication cost. In contrast, the communication cost of our GPA method only needs to be any order larger than $O(N^{1/5})$.
	

	\section{Additional Numerical Experiments}
	\label{append:comparison}
	
	In this subsection, we compare the GPA methods with other methods.
	Specifically, we consider the one-shot type kernel estimator $\hat\mu_{\OS}$, the linear interpolation based GPA estimator $\hat\mu_{\GPA}$, the cubic polynomial interpolation based estimator $\hat\mu_{\PGPA, 3}$, the one-shot type basis expansion estimator $\hat\mu_{\BE, \OS}$, the moment assembling (MA) type basis expansion estimator $\hat\mu_{\BE, \MA}$, and the RKHS based one-shot type kernel ridge
	regression (KRR) estimator $\hat\mu_{\KRR, \OS}$. 
	For the cubic polynomial interpolation based estimator $\hat\mu_{\PGPA, 3}$, we use the 4th-order kernel function $K(u) = (45/32)(1-7u^2/3)(1-u^2)I(|u|\le 1) $.
	The two types of basis expansion estimators are briefly discussed in the Appendix \ref{append:basis_expansion}. 
	For the one-shot type KRR estimator $\hat\mu_{\KRR, \OS}$, we use the reproducing kernel function $R(s,t) = (-1/4!) B_4(\textup{frac}\{s-t\})$ \citep{wahba1990spline}, where $\textup{frac}\{x\}$ denotes the fractional part of $x$, and $B_4(\cdot)$ is the 4th Bernoulli polynomials \citep{olver2010nist}. 
	We use the pilot sample approach described in Section \ref{subsec:bandwidth} to choose the tuning parameters for each method, those are, the bandwidth for the GPA estimators, the number of basis for the basis expansion estimators, and the regularization parameter for the KRR estimator. 
	
	We generate $X_i$ from the uniform distribution $\text{Unif}(0,1)$, and consider a mean function $\mu_3(x) = 24\sqrt{x(1-x)} \sin\{2.1\pi /(x+0.05)\}$ from \cite{fan1996local}. 
	Once $X_i$ is obtained, $Y_i$ is generated according to $Y_i = \mu(X_i) + \varepsilon_i$ with $\varepsilon_i$ simulated from the standard normal distribution. 
	In this experiment, we consider two different training sample sizes $N = (5\times 10^4, 1\times 10^5)$ and two different testing sample sizes $N^*=(1\times 10^4, 2\times 10^4)$. We fix the number of machines $M=50$, and the pilot sample size $n_0 = 2,000$. We further consider two different local sample allocation strategies, i.e., random partition and nonrandom partition, as in Section \ref{subsec:dist_KE}.
 Once the whole training sample has been partitioned, we employ different methods to make predictions for the testing sample and then compute the corresponding RMSE values. In addition, the computation and communication times for both the training phase (i.e. $T_\textup{cp}$ and $T_\textup{cm}$) and the prediction phase (i.e. $P_\textup{cp}$ and $P_\textup{cm}$) are recorded. We replicate each experiment $B=100$ times. The RMSE values and various time costs over $100$ replications for different sample allocation strategies are then averaged.
 The detailed results are given in Tables \ref{tab:RMSE_Time} and \ref{tab:RMSE_Time_2}.

		

	\begin{table}[htbp]
		\caption{The averaged RMSE values are computed based $100$ random replications for different estimators by the random partition strategy. $T_\textup{cp}$ and $T_\textup{cm}$ stand for the computation and communication times for the training phase, respectively. $P_\textup{cp}$ and $P_\textup{cm}$ stand for the computation and communication times for the prediction phase, respectively. }\label{tab:RMSE_Time}
	\centering
	\begin{tabular}{>{\centering\arraybackslash}p{0.15\textwidth}| >{\centering\arraybackslash}p{0.12\textwidth} *{6}{>{\centering\arraybackslash}p{0.08\textwidth}} }
		\toprule
		Sample Size  &  & $\hat\mu_{\OS}$ & $\hat\mu_{\GPA}$ & $\hat\mu_{\PGPA, 3}$ & $\hat\mu_{\BE, \OS}$ & $\hat\mu_{\BE, \MA}$&$\hat\mu_{\KRR, \OS}$\\
		\hline
		&RMSE                                      &  0.172 & 0.167 & 0.144 & 0.298 & 0.266 & 0.240 \\
		$N=5\times10^4$   &$T_\textup{cp}$ (sec.)  &  0.000 & 0.169 & 0.088 & 0.079 & 0.098 & 0.330 \\
		&$T_\textup{cm}$ (sec.)                    &  0.000 & 0.191 & 0.154 & 0.155 & 0.485 & 0.000 \\
		$N^*=1\times10^4$ &$P_\textup{cp}$ (sec.)  &  1.007 & 0.003 & 0.103 & 0.048 & 0.043 & 2.159 \\
		&$P_\textup{cm}$ (sec.)                    &  0.384 & 0.000 & 0.000 & 0.000 & 0.000 & 0.554 \\
		\hline
		&RMSE                                      &  0.131 & 0.121 & 0.110 & 0.185 & 0.172 & 0.118 \\
		$N=1\times10^5$   &$T_\textup{cp}$ (sec.)  &  0.000 & 0.449 & 0.235 & 0.199 & 0.253 & 1.309 \\
		&$T_\textup{cm}$ (sec.)                    &  0.000 & 0.222 & 0.193 & 0.198 & 0.507 & 0.000 \\
		$N^*=2\times10^4$ &$P_\textup{cp}$ (sec.)  &  4.690 & 0.004 & 0.113 & 0.086 & 0.075 & 9.290 \\
		&$P_\textup{cm}$ (sec.)                    &  0.802 & 0.000 & 0.000 & 0.000 & 0.000 & 1.085 \\
		\bottomrule
	\end{tabular}
\end{table}

Table \ref{tab:RMSE_Time} reports the detailed simulation results for the random partition strategy.
By Table \ref{tab:RMSE_Time} we obtain the following interesting findings. 
First, the RMSE values for all distributed estimators decrease as the training sample size increases. This verifies the consistency of these estimators for the random partition strategy.
Second, we find that nonzero time cost due to communication (i.e., $P_\textup{cm}$) is required by the one-shot kernel estimator $\hat\mu_{\OS}$ and the KRR estimator $\hat\mu_{\KRR, \OS}$ for prediction. In contrast, zero communication cost is needed by the other methods (including the two GPA estimators) since efficient prediction can be done solely by the central machine. 
Furthermore, we observe that the communication times of two GPA estimators are lower than that of the basis expansion estimator $\hat\mu_{\BE, \MA}$, which is computed by the moment assembling. 
For instance, when sample size $N=5\times 10^4$, the communication times (i.e., $T_\textup{cm}$) of the two GPA estimators are less than $0.2$ seconds, whereas that of the estimator $\hat\mu_{\BE, \MA}$ is over $0.4$ seconds.
This finding supports our discussion in Appendix \ref{append:basis_expansion}, where we highlight that the moment assembling method could incur higher communication costs than the GPA method.
Lastly, we find that the RMSE values of the cubic polynomial interpolated estimator $\hat\mu_{\PGPA, 3}$ are smaller than that of the linear interpolation based GPA estimator $\hat\mu_{\GPA}$. This also corroborates our theoretical findings in Theorem \ref{thm:GPA_nu} very well. That is the higher order interpolation can lead to a more accurate estimator if the underlying mean function is sufficiently smooth.

Table \ref{tab:RMSE_Time_2} reports the results for the nonrandom partition strategy.
By Table \ref{tab:RMSE_Time_2} we find that NA values are reported for the one-shot kernel estimator $\hat\mu_{\OS}$. 
This is because the support of $X_i$s on the local machine may not cover all $X_i^*$s in the testing sample in this case. Consequently, these machines cannot make predictions for these testing observations, and thus $\hat\mu_{\OS}$ is not well defined.
Moreover, extremely large RMSE values are obtained for both the $\hat\mu_{\BE, \OS}$ and $\hat\mu_{\KRR, \OS}$. This is because both the $\hat\mu_{\BE, \OS}$ and $\hat\mu_{\KRR, \OS}$ are inconsistent with the nonrandom partition strategy.
Lastly, the two GPA estimators and the moment assembling type basis expansion estimator (i.e., $\hat\mu_{\BE, \MA}$) are not affected by the nonrandom partition.
Furthermore, similar to the results in  Table \ref{tab:RMSE_Time}, both GPA estimators require significantly less time cost for communication time during the training phase. 

\begin{table}[htbp]
	\caption{The averaged RMSE values are computed based $100$ random replications for different estimators by the nonrandom partition strategy. $T_\textup{cp}$ and $T_\textup{cm}$ stand for the computation and communication times for the training phase, respectively. $P_\textup{cp}$ and $P_\textup{cm}$ stand for the computation and communication times for the prediction phase, respectively.}\label{tab:RMSE_Time_2}
\centering
\begin{tabular}{>{\centering\arraybackslash}p{0.15\textwidth}| >{\centering\arraybackslash}p{0.12\textwidth} *{6}{>{\centering\arraybackslash}p{0.08\textwidth}} }
	\toprule
	Sample Size  &  & $\hat\mu_{\OS}$ & $\hat\mu_{\GPA}$ & $\hat\mu_{\PGPA, 3}$ & $\hat\mu_{\BE, \OS}$ & $\hat\mu_{\BE, \MA}$&$\hat\mu_{\KRR, \OS}$\\
	\hline
	&RMSE                                      &  NA & 0.167 & 0.144 & 105.033 & 0.266 & 16.085 \\
	$N=5\times10^4$   &$T_\textup{cp}$ (sec.)  &  NA & 0.175 & 0.096 &   0.081 & 0.130 &  0.321 \\
	&$T_\textup{cm}$ (sec.)                    &  NA & 0.176 & 0.157 &   0.152 & 0.442 &  0.000 \\
	$N^*=1\times10^4$ &$P_\textup{cp}$ (sec.)  &  NA & 0.003 & 0.098 &   0.048 & 0.038 &  2.141 \\
	&$P_\textup{cm}$ (sec.)                    &  NA & 0.000 & 0.000 &   0.000 & 0.000 &  0.552 \\
	\hline
	&RMSE                                      &  NA & 0.121 & 0.110 & 61.323 & 0.172 & 22.069 \\
	$N=1\times10^5$   &$T_\textup{cp}$ (sec.)  &  NA & 0.385 & 0.194 &  0.195 & 0.301 &  1.159 \\
	&$T_\textup{cm}$ (sec.)                    &  NA & 0.227 & 0.193 &  0.201 & 0.516 &  0.000 \\
	$N^*=2\times10^4$ &$P_\textup{cp}$ (sec.)  &  NA & 0.005 & 0.119 &  0.085 & 0.074 &  6.487 \\
	&$P_\textup{cm}$ (sec.)                    &  NA & 0.000 & 0.000 &  0.000 & 0.000 &  0.911 \\
	\bottomrule
\end{tabular}
\end{table}

\end{appendix}
}\fi

\end{document}